\documentclass[12pt,a4paper]{article}
    \usepackage[cp1251]{inputenc}
    \usepackage[english]{babel}
\usepackage{algorithm}
\usepackage[noend]{algpseudocode}    

    \usepackage{amsmath,amssymb,amsthm}
\usepackage{stmaryrd}
\usepackage{tikz}
\usepackage{multirow}
\usepackage{hyperref} 

\usetikzlibrary{positioning,shapes.gates.logic.US}
\usepackage{mathtools}

\usetikzlibrary{arrows}
\usetikzlibrary{graphs}

\usepackage{tikz-qtree}

\usepackage{xypic}
\usetikzlibrary{shapes}
\usetikzlibrary{arrows.meta}
     \usepackage{graphicx}
     \usepackage{wrapfig}
\newcommand{\bleqinn}[2]{%
  \ooalign{%
    \cr
    $#1\leq$\cr\raisebox{.2ex}{$#1\color{black!55}\blacktriangleleft$}}%
}

\newcommand{\us}{\mathrel{\lhd}}

\newcommand{\useq}{\mathrel{\trianglelefteq}}

\newcommand{\sus}{
\mathrel{\color{black!55}\blacktriangleleft}}

\newcommand{\suseq}{
\mathrel{\mathpalette\bleqinn\relax}}

\newcommand{\uuus}{\mathrel{\us\!\!\us\!\!\us}}

\newcommand{\szigzag}{\mathrel{\supset\!\!\!\sus\!\!\supset\!\!\!\sus}}

\newcommand{\DrawAnd}[3]{ 
\node      (x#3 1)             at (#1-1.5,#2-0.5)                 {$x_{#3}^{1}$};
\node      (x#3 0)             at (#1-1.5,#2+0.5)                 {$x_{#3}^{0}$};
\node[andnode]     (and#3)   at (#1,#2)                    {AND};
\draw[->] (x#3 0) to (and#3.138);
\draw[->] (x#3 1) to (and#3.222);
}
\newcommand{\DrawBlock}[4]{ 
\node[rnode,text width=0.75cm,align=center,scale = 0.6]     (b#3)   at (#1,#2)  {#4};
}
\newcommand{\DrawBlockDashed}[4]{ 
\node[rnode,dashed,text width=0.75cm,align=center,scale = 0.6]     (b#3)   at (#1,#2)  {#4};
}
\newcommand{\DrawAndDotted}[3]{ 
\node      (x#3 1)             at (#1-1.5,#2-0.5)                 {$x_{#3}^{1}$};
\node      (x#3 0)             at (#1-1.5,#2+0.5)                 {$x_{#3}^{0}$};
\node[andnodedashed]     (and#3)   at (#1,#2)                    {AND};
\draw[->,dashed] (x#3 0) to (and#3.138);
\draw[->,dashed] (x#3 1) to (and#3.222);
}

  \DeclareMathOperator\CSP{CSP}

\DeclareMathOperator{\proj}{pr}

  \DeclareMathOperator\oneinthree{1IN3'}
  \DeclareMathOperator\QCSP{QCSP}

  \DeclareMathOperator{\Var}{Var}

\newcommand{\zv}{\mathbf z}

\newcommand{\uv}{\mathbf u}

\newcommand{\xv}{\mathbf x}

\newcommand{\join}{\ovee}

\renewcommand{\le}{\leqslant}
\renewcommand{\ge}{\geqslant}

\theoremstyle{definition}
\theoremstyle{plain}

\newtheorem{thm}{Theorem}
\newtheorem{problem}{Question}
\newtheorem{conj}{Conjecture}
\newtheorem{lem}[thm]{Lemma}
\newtheorem{sublem}{Claim}[subsection]
\newtheorem{cor}[thm]{Corollary}

\newenvironment{sketch}{\begin{proof}[\textsc{Sketch of the proof:}]}{\end{proof}}

\newtheorem*{THMFindSmallTreeTHM}{Theorem~\ref{THMFindSmallTree}}

\newtheorem*{THMNonemptyReductionIsZigzagTHM}{Theorem~\ref{THMNonemptyReductionIsZigzag}}

\newtheorem*{THMFindUniversalSubuniverseTHM}{Theorem~\ref{THMFindUniversalSubuniverse}}

\newtheorem*{THMFindSmallerReductionTHM}{Theorem~\ref{THMFindSmallerReduction}}

\newtheorem*{THMMightyTupleIPSpaceHardnessTHM}{Theorem~\ref{THMMightyTupleIPSpaceHardness}}

\newtheorem*{LEMUniversalSubuniverseImpliesLEM}{Lemma~\ref{LEMUniversalSubuniverseImplies}}

\newtheorem*{LEMMightyTupleTwoImpliesLEM}{Lemma~\ref{LEMMightyTupleTwoImplies}}

\newtheorem*{LEMMightyTupleFiveImpliesLEM}{Lemma~\ref{LEMMightyTupleFiveImplies}}

\newtheorem*{LEMTHMIdempotantClassificationLEM}{Lemma~\ref{LEMTHMIdempotantClassification}}

\newtheorem*{LEMMightyTupleTwoThreeFourEquivalenceLEM}{Lemma~\ref{LEMMightyTupleTwoThreeFourEquivalence}}

\title{$\Pi_2^{P}$ vs PSpace Dichotomy for the Quantified Constraint Satisfaction Problem}
\author{Dmitriy Zhuk\thanks{The author is funded by the European Union (ERC, POCOCOP, 101071674). Views and opinions expressed are however those of the author(s) only and do not necessarily reflect those of the European Union or the European Research Council Executive Agency. Neither the European Union nor the granting authority can be held responsible for them.}}
 \date{}
\begin{document}
\maketitle
{\footnotesize
\tableofcontents
}
\begin{abstract}
The Quantified Constraint Satisfaction Problem is the 
problem of evaluating a sentence with both quantifiers, over relations from some constraint language, with conjunction as the only connective.
We show that for any constraint language 
on a finite domain the Quantified Constraint Satisfaction Problem 
is either in $\Pi_{2}^{P}$, or PSpace-complete.
Additionally, we build a constraint language 
on a 6-element domain such that the Quantified Constraint Satisfaction Problem 
over this language is $\Pi_{2}^{P}$-complete.
\end{abstract}

\section{Introduction}

The \emph{Quantified Constraint Satisfaction Problem} $\QCSP(\Gamma)$ is the generalization of the \emph{Constraint Satisfaction Problem} $\CSP(\Gamma)$ which, given the latter in its logical form, augments its native existential quantification with universal quantification. That is, $\QCSP(\Gamma)$ is the problem to evaluate a sentence of the form $\forall x_1 \exists y_1 \ldots \forall x_n \exists y_n \ \Phi$, where $\Phi$ is a conjunction of relations from the \emph{constraint language} $\Gamma$, all over the same finite domain $A$.
Since the resolution of the Feder-Vardi ``Dichotomy'' Conjecture, classifying the complexity of $\CSP(\Gamma)$, for all finite  $\Gamma$, between P and NP-complete \cite{BulatovFVConjecture,BulatovProofCSP,ZhukFVConjecture,zhuk_CSP_Dichotomy_JACM}, a desire has been building for a classification for $\QCSP(\Gamma)$. Indeed, since the  classification of the \emph{Valued CSPs} was reduced to that for CSPs \cite{KolmogorovKR17}, the QCSP remains the last of the older variants of the CSP to have been systematically studied but not classified. More recently, other interesting open classification questions have appeared such as that for \emph{Promise CSPs} \cite{BrakensiekG18} and finitely-bounded, homogeneous infinite-domain CSPs \cite{BartoP16}. 


\subsection{Complexity of the QCSP}

While $\CSP(\Gamma)$ remains in NP for any finite $\Gamma$, $\QCSP(\Gamma)$ can be PSpace-complete, as witnessed by \emph{Quantified 3-Satisfiability} or \emph{Quantified Graph 3-Colouring} (see \cite{BBCJK}). It is well-known that the complexity classification for QCSPs embeds the classification for CSPs: if $\Gamma+1$ is $\Gamma$ with the addition of a new isolated element not appearing in any relations, then $\CSP(\Gamma)$ and $\QCSP(\Gamma+1)$ are polynomially equivalent. Thus, and similarly to the Valued CSPs, the CSP classification will play a part in the QCSP classification. 
For a long time the complexities P, NP-complete, and PSpace-complete were the only complexity classes that could be achieved by $\QCSP(\Gamma)$ \cite{BBCJK,hubie-sicomp,HubieSIGACT,Meditations,QC2017}.
Nevertheless, in \cite{QCSPMonstersSTOC,QCSP_Monsters_JACM} 
a constraint language $\Gamma$ on a 3-element domain was discovered 
such that $\QCSP(\Gamma)$ is coNP-complete. Combining 
this language with an NP-complete language the authors also built a
DP-complete constraint language on a 4-element domain and a
$\Theta_{2}^{P}$-complete language on a 10-element domain \cite{QCSPMonstersSTOC,QCSP_Monsters_JACM}.
Discovering these exotic complexity classes ruined hope to obtain
a simple and complete classification of the complexity of the QCSP for 
all constraint languages on a finite domain.
On the other hand, the possibility to express those complexity classes 
by fixing a constraint language makes the 
QCSP a powerful tool for studying complexity classes between P and PSpace.  
Finding a concrete border between complexity classes in terms of 
constraint languages may shed some light on the fundamental differences between them,  and may bring us closer to understanding why P and PSpace are different (if they are).

The exotic complexity classes appeared only on domains of size at least 4, 
while on a domain of size 2 we have a complete classification between P and PSpace-complete, and on a domain of size 3 we have a partial classification between P, NP-complete, coNP-complete, and PSpace-complete.

\begin{thm}[\cite{Schaefer}]
Suppose $\Gamma$ is a constraint language on
$\{0,1\}$. 
Then $\QCSP(\Gamma)$ is in P if 
$\CSP(\Gamma\cup\{x=0,x=1\})$ is in P, 
$\QCSP(\Gamma)$ is PSpace-complete otherwise.
\end{thm}

\begin{thm}[\cite{QCSPMonstersSTOC,QCSP_Monsters_JACM}]\label{THMThreeElementIdempotentClassificaition}
Suppose $\Gamma\supseteq \{x=a\mid a\in A\}$ is a constraint language on $\{0,1,2\}$.
Then $\QCSP(\Gamma)$ is either in P, or NP-complete, or coNP-complete, or 
PSpace-complete.
\end{thm}

The statement proved in \cite{QCSPMonstersSTOC,QCSP_Monsters_JACM} is stronger than Theorem \ref{THMThreeElementIdempotentClassificaition} as the authors provide  necessary and sufficient conditions for the $\QCSP(\Gamma)$ to be in each of these classes.
Notice that 
for the QCSP we do not know 
a simple trick that allows us to find an equivalent constraint language 
with all constant relations
$\{x=a\mid a\in A\}$
for a constraint language without.
Recall that for the usual CSP we first consider the core of the language and then 
safely add all the constant relations to it \cite{jeavons1998algebraic,CSPconjecture}.
For the QCSP reducing the domain is not an option as 
the universal quantifier lives on the whole domain.
That is why, Theorem \ref{THMThreeElementIdempotentClassificaition},
has been proved only 
for constraint languages with all constant relations, 
and a complete classification 
for all constraint languages 
on a 3-element domain is wide open.

\subsection{Reduction to CSP}\label{SUBSECTIONReductionTOCSPInIntro}
It is natural to try to reduce the QCSP to its older brother CSP. 
In fact, any QCSP instance 
$\exists y_{0}\forall x_1\exists y_1\dots\forall x_n \exists y_n\; \Psi$ can be viewed as a CSP instance of an exponential size. 
If a QCSP-sentence holds, then there exists a winning strategy for the Existential Player (EP) defined by Skolem functions, i.e.,
$y_{i} = f_{i}(x_1,\dots,x_{i})$. 
We encode every value of $f_{i}(a_1,\dots,a_{i})$ 
by a new variable $y_{i}^{a_1,\dots,a_{i}}$, and 
for any play of the Universal Player (UP)
we list all the constraints that have to be satisfied (see Section \ref{SUBSECTIONInducedCSPInstances} for more details).

Clearly, this procedure gives us nothing algorithmically, because the obtained CSP instance is of exponential size.
Nevertheless, we might ask whether it is necessary to  
look at the whole instance to learn that it does not hold, which can be formulated as follows.
We say that 
the UP \emph{wins on $S\subseteq A^{n}$}
in $\exists y_{0}\forall x_1\exists y_1\dots\forall x_n \exists y_n\; \Psi$
if the instance $$\exists y_0 \forall x_1\exists y_1\dots
\forall x_n\exists y_n(
(x_1,\dots,x_n)\in S\rightarrow \Psi)$$ does not hold.

\begin{problem}\label{QUESTIONSizeOfUPWinningSet}
For a No-instance of $\QCSP(\Gamma)$ with $n$ universal variables, what is the minimal $S\subseteq A^{n}$ such that the UP wins on $S$?
\end{problem}

In this paper we answer this fundamental question by showing that unless the problem is PSpace-hard, 
the set $S$ can be chosen of polynomial size. 
Notice that for the PSpace-hard case we should not expect $S$ to be of non-exponential size, as it would 
send our problem to some class below PSpace. 

It would be even better if the set $S$, on which the UP wins, 
could be fixed for all No-instances or could be calculated efficiently. We can ask the following question.

\begin{problem}
What is the minimal 
$S\subseteq A^{n}$ 
such that 
for any No-instance of $\QCSP(\Gamma)$
with $n$ universal variables
the UP wins on $S$?
\end{problem}

If $S$ can always be chosen of polynomial-size and  
can be computed efficiently, 
then $\QCSP(\Gamma)$ immediately goes to 
the complexity class NP, as it is reduced to a polynomial-size CSP instance that can be efficiently computed.
Surprisingly, all the problems $\QCSP(\Gamma)$ known to be in NP by 2018 satisfy the above property \cite{AU-Chen-PGP,Meditations,QC2017}.
In fact, as it is shown in \cite{zhukPGPReductionArxiv}, 
for all constraint languages whose polymorphisms satisfy the Polynomially Generated Powers (PGP) Property,
the set $S$ can be chosen to be very simple:
there exists $k$ such that 
the UP wins in any No-instance 
on the set of
all tuples having at most 
$k$ switches, where a switch in a tuple 
$(a_1,\dots,a_n)$ is a pair $(a_{i},a_{i+1})$ 
such that $a_i\neq a_{i+1}$.
Moreover, as it was shown in \cite{ZhukGap2015}, if polymorphisms do not 
satisfy the PGP property, they satisfy the Exponential Generated Powers (EGP) Property, which automatically implies that 
such a polynomial-size $S$ cannot exist (at least if the number of existential variables is not limited).

Surprisingly, in \cite{QCSPMonstersSTOC,QCSP_Monsters_JACM} 
two constraint languages on a 3-element domain were discovered such that 
the QCSP over these languages is solvable in polynomial time,
but they do not satisfy the PGP property and, therefore, we cannot fix a polynomial-size $S$.
Nevertheless, for every instance 
we can efficiently calculate a polynomial-size $S$ such that 
if the UP can win, the UP wins on $S$.
We can formulate the following open question.

\begin{problem}
Suppose $\QCSP(\Gamma)$ is in NP. Is it true that 
for any instance of $\QCSP(\Gamma)$ with $n$ universal variables there exists a polynomial-time computable 
set $S\subseteq A^{n}$ such that 
the UP can win if and only if the UP wins on $S$?
\end{problem}

\section{Main Results}\label{SECTIONMainResults}

\subsection{$\Pi_{2}^{P}$ vs PSpace Dichotomy}

The main result of this paper comes from Question \ref{QUESTIONSizeOfUPWinningSet} from the introduction. 
We show that 
if $\QCSP(\Gamma)$ is not PSpace-complete and the UP has a winning strategy in a concrete QCSP instance, 
then this winning strategy can be chosen to be rather simple.
We cannot expect the winning strategy for the UP to be polynomial-time computable because this would 
imply that $\QCSP(\Gamma)$ is in NP, 
and we know that $\QCSP(\Gamma)$ can be coNP-complete
\cite{QCSP_Monsters_JACM}. 
Nevertheless, as we show in the next theorem, 
the UP wins in any No-instance on a set $S$ of polynomial size, that is, we can restrict the UP to polynomially many possible moves and he still wins.

\begin{thm}\label{THMMainUPRestriction}
Suppose $\Gamma$ is a constraint language on a finite set $A$, $\QCSP(\Gamma)$ is not PSpace-hard.
Then for any No-instance
$\exists y_0 \forall x_1\exists y_1\dots
\forall x_n\exists y_n\Psi$ of $\QCSP(\Gamma)$ there exists 
    $S\subseteq A^{n}$ with $|S|\le|A|^{2}\cdot (n\cdot |A|)^{2^{{2|A|}^{|A|+1}}}$ such that  
    $$\exists y_0 \forall x_1\exists y_1\dots
\forall x_n\exists y_n(
(x_1,\dots,x_n)\in S\rightarrow \Psi)$$ does not hold.
\end{thm}

In other words, the above theorem states that 
unless $\QCSP(\Gamma)$ is PSpace-hard, 
for any No-instance the UP wins on a set $S$ of polynomial-size (notice that the domain $A$ is fixed).
If the polynomial-size set $S$ is fixed,
then to confirm that the instance does not hold we need 
to check all the strategies of the EP
defined on prefixes of the words (tuples) from $S$, which is also a polynomial-size set.
Thus, if $\QCSP(\Gamma)$ is not PSpace-hard, 
then to solve an instance 
$\exists y_0 \forall x_1\exists y_1\dots
\forall x_n\exists y_n\Psi$
we need to check that for all $S\subseteq A^{n}$ with $|S|\le|A|^{2}\cdot (n\cdot |A|)^{2^{{2|A|}^{|A|+1}}}$ 
there exists a winning strategy for the EP for the restricted problem, 
which sends the problem to the complexity class $\Pi_2^{P}$.
In fact,  
$\Pi_{2}^{P}$ is the class of problems $\mathcal U$
that can be defined as 
$$\mathcal U(Z) = \forall X ^{\color{black!40}|X|<p(|Z|)} \exists Y
^{\color{black!40}|Y|<q(|Z|)}
\mathcal V(X,Y,Z),$$
for some $\mathcal V\in \mathrm{P}$ and some polynomials $p$ and $q$.
In our case the set $S$ plays the role of $X$ and the restricted Skolem functions play the role of $Y$.
Then, in $\mathcal V$ we need to check for every tuple from $S$ (play of the UP) that the 
corresponding strategy of the EP works, which is obviously computable in polynomial time. 
Thus, we have the following Dichotomy Theorem.

\begin{thm}
Suppose $\Gamma$ is a constraint language on a finite set. 
Then $\QCSP(\Gamma)$ is 
\begin{itemize}
    \item PSpace-complete or
    \item in $\Pi_2^{P}$.
\end{itemize}
\end{thm}

\subsection{What is inside $\Pi_{2}^{P}$?}\label{SUBSECTIONWhatIsInside}

We show that the gap between 
Pspace and $\Pi_2^{P}$ cannot be enlarged, 
and there is a constraint language whose QCSP is $\Pi_2^P$-complete.

\begin{thm}\label{THMPiTwoCompleteLanguage}
There exists $\Gamma$ on a 6-element domain
such that 
$\QCSP(\Gamma)$ is $\Pi_2^{P}$-complete.
\end{thm}

Thus, we already have 7 complexity classes that can be expressed 
as the QCSP for some constraint language:
P, NP, coNP, DP, $\Theta_{2}^{P}$, $\Pi_{2}^{P}$, and PSpace.
In Figure \ref{fig:ComplexityClasses}
we show all the complexity classes known to be expressible as the QCSP and inclusions between them, 
where the edge is solid if we know that there are no classes between them, and 
dotted otherwise.

\begin{problem}
Are there any other complexity classes up to polynomial 
reduction that can be expressed as
$\QCSP(\Gamma)$ for some $\Gamma$ on a finite set? 
\end{problem}

In fact, we want to prove or disprove the following dichotomy claims:

\begin{problem}
Suppose $\Gamma$ is a constraint language on a finite set. Is it true that
\begin{enumerate}
    \item $\QCSP(\Gamma)$ is either $\Pi_2^{P}$-hard, or in 
    $\Theta_{2}^{P}$?
    \item $\QCSP(\Gamma)$ is either $\Theta_2^{P}$-hard, or in 
    $\mathrm{DP}$?
    \item $\QCSP(\Gamma)$ is either $\mathrm{DP}$-hard, or in 
    $\mathrm{NP}\cup \mathrm{coNP}$?
    \item $\QCSP(\Gamma)$ is either $\mathrm{NP}$-hard, or in 
    $\mathrm{coNP}$?
    \item $\QCSP(\Gamma)$ is either $\mathrm{coNP}$-hard, or in 
    $\mathrm{NP}$?
    \item $\QCSP(\Gamma)$ is either in $\mathrm{P}$, or
    $\mathrm{NP}$-hard, or $\mathrm{coNP}$-hard? 
\end{enumerate}
\end{problem}

\begin{figure}
    \centering

    \begin{tikzpicture}[scale = 1]


\node at (0,2) (p){\LARGE P};
\node at (2.2,3)(np){\LARGE NP}; 
\node at (2.2,1)(conp){\Large coNP}; 
\node at (4.4,2)(dp){\LARGE DP}; green,[rotate =6]
\node at (6.6,1)(theta){\LARGE $\Theta_{2}^{P}$}; 
\node at (9,1)(pi){\LARGE $\Pi_{2}^{P}$}; 
\node at (10,3.3-0.2)(pspace){\Large PSPACE}; 
\draw [line width=0.9mm, -stealth, black!50,dashed] (p) to[in=190,out=40] (np);
\draw [line width=0.9mm, -stealth, black!50,dashed] (p) to[in=160,out=-30] (conp);
\draw [line width=0.9mm, -stealth, black!50,dashed] (conp) to[in=190,out=40] (dp);
\draw [line width=0.9mm, -stealth, black!50,dashed] (np) to[in=140,out=-10] (dp);
\draw [line width=0.9mm, -stealth, black!50,dashed] (dp) to[in=120,out=10] (theta);
\draw [line width=0.9mm, -stealth, black!50,dashed] (theta) to[in=170,out=10] (pi);

\draw [line width=0.9mm, -stealth, black!50] (pi) to[in=250,out=60] (pspace);

\end{tikzpicture}
    \caption{Complexity classes expressible as $\QCSP(\Gamma)$ for some $\Gamma$.}
    \label{fig:ComplexityClasses}
\end{figure}
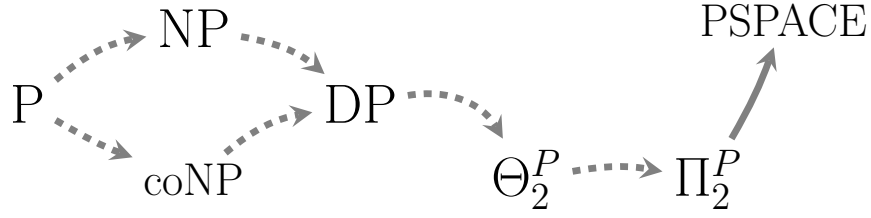

It is not hard to 
build an example showing that 
we cannot just move all universal quantifiers left 
to reduce $\QCSP(\Gamma)$ to a $\Pi_2^{P}$-sentence even 
if $\QCSP(\Gamma)$ is in $\Pi_2^{P}$.
Nevertheless, it is still not clear whether a smarter 
polynomial reduction 
to a $\Pi_2$-sentence over
the same language exists.
We denote the modification of $\QCSP(\Gamma)$ in which 
only $\Pi_2$-sentences are allowed 
by $\Pi_2$-$\QCSP(\Gamma)$. 
Then this question can be formulated as follows.

\begin{problem}
Suppose $\Gamma$ is a constraint language on a finite set and $\QCSP(\Gamma)$ is $\Pi_2^{P}$-complete.
Is it true that $\Pi_2$-$\QCSP(\Gamma)$ is
$\Pi_2^{P}$-complete?
\end{problem}




One may also ask whether 
it is sufficient to consider only $\Pi_2$-sentences 
for all complexity classes but PSpace.

\begin{problem}
Suppose $\Gamma$ is a constraint language on a finite domain and $\QCSP(\Gamma)$ is in $\Pi_2^{P}$.
Is it true that $\Pi_2$-$\QCSP(\Gamma)$ is
polynomially equivalent to 
$\QCSP(\Gamma)$?
\end{problem}

A positive answer to this question would make a complete classification 
of the complexity of $\QCSP(\Gamma)$ for each $\Gamma$ much closer.
Checking a $\Pi_2$-sentence is equivalent to 
solving a Constraint Satisfaction Problem for every evaluation of 
universal variables,
but if we need to check exponentially many of them, it does not give us 
an efficient algorithm.
It is very similar to Question \ref{QUESTIONSizeOfUPWinningSet} from the introduction on whether 
the UP can win only playing strategies from 
a polynomial-size subset, but for the 
$\Pi_{2}$-sentence the situation is much easier as 
the UP plays first and the EP just reacts.

Earlier Hubie Chen 
noticed \cite{AU-Chen-PGP} that in some cases it is sufficient to check only polynomially 
many evaluations to guarantee that the $\Pi_2$-sentence 
holds, 
which implies that the problem is equivalent to the CSP and belongs to NP.
Precisely, this reduction works for constraint languages 
satisfying the Polynomially Generated Powers (PGP) Property
already mentioned in the introduction. 
These are languages such that 
all the tuples of $A^{n}$ can be generated from 
polynomially many tuples by applying polymorphisms 
of $\Gamma$ coordinate-wise.
Notice that in the PGP case this polynomial set of evaluations can be chosen independently 
of the instance and can be calculated efficiently, as it is just the set of all tuples with at most $k$ switches \cite{ZhukGap2015}. This gives us a very simple polynomial reduction to CSP \cite{AU-Chen-PGP}.

Sometimes a similar strategy works even if the polymorphisms of the constraint language do not satisfy the PGP property:
two such constraint languages were presented
in \cite{QCSP_Monsters_JACM}. 
The polynomial algorithm for them works as follows.
First, by solving many CSP instances it calculates the polynomial
set of tuples (evaluations of the universal variables).
Then, again by solving CSP instances, it checks that the quantifier-free part of the instance is satisfiable for every tuple (evaluation)  it found.
This gives us a Turing reduction to the CSP, 
and if the CSP is solvable in polynomial time, it gives us a polynomial 
algorithm. This idea completed the classification of 
the complexity of the QCSP for all constraint languages on a 3-element set containing all constant relations \cite{QCSP_Monsters_JACM}, and we hope that a generalization of this idea will lead to a complete 
classification of the complexity inside $\Pi_2^{P}$.

\subsection{PSpace-complete languages}
The complexity of the CSP for 
a (finite) constraint language $\Gamma$ has a very simple characterization 
in terms of polymorphisms. 
Precisely, $\CSP(\Gamma)$ is solvable in polynomial time if
$\Gamma$ admits a cyclic polymorphism, 
and it is NP-complete otherwise \cite{BulatovFVConjecture,BulatovProofCSP,ZhukFVConjecture, zhuk_CSP_Dichotomy_JACM}.
It is also known that the complexity of 
$\QCSP(\Gamma)$ is determined by surjective polymorphisms of $\Gamma$ \cite{BBCJK},
but we are not aware of a nice characterization of 
 $\Pi_2^{P}$-membership in terms of polymorphisms, 
 moreover polymorphisms do not play any role in this paper.
Nevertheless, we have a nice characterization in terms of relations. 
It turned out all the PSpace-hard cases are similar in the sense that they can express certain relations giving us PSpace-hardness. 
We say that a constraint language $\Gamma$ 
\emph{q-defines} a relation $R$ 
if there exists a 
quantified conjunctive formula 
over $\Gamma$ that defines the relation $R$.
Similarly, we say that 
$\Gamma$ \emph{q-defines} a set $S$ of relations if it 
q-defines each relation from $S$. In this case we also say 
that $S$ is \emph{q-definable over $\Gamma$.}
It is an easy observation that 
$\QCSP(\Gamma_1)$ can be (LOGSPACE) reduced 
to $\QCSP(\Gamma_2)$ if 
$\Gamma_2$ q-defines $\Gamma_1$ \cite{BBCJK}. 

To formulate the classification of all PSpace-complete languages, 
we introduce the notion of a mighty tuple.
Suppose 
$k\ge 0$, $m\ge 1$, 
$Q\subseteq A^{|A|+k+m+2}$, 
$B,C,D\subseteq A^{|A|+k+m+1}$, 
$\Delta\subseteq A^{|A|+k}$.
The relation $Q$ can be viewed as 
a binary relation having three additional parameters 
$\zv\in A^{|A|}$, $\delta\in A^{k}$, and 
$\alpha\in A^m$.
Similarly, 
$\Delta$ is a $k$-ary relation 
with an additional parameter $\zv\in A^{|A|}$, 
$B,C,D$ are unary relations
with additional parameters
$\zv\in A^{|A|}$ and $\delta\in A^{k}$.
By 
$\prescript{\zv}{}\Delta$,
$\prescript{\zv}{\delta}Q^{\alpha}$, 
$\prescript{\zv}{\delta}D$,
$\prescript{\zv}{\delta}B$,
$\prescript{\zv}{\delta}C$
we denote the respective $k$-ary, binary, and unary relations 
where these parameters are fixed.
For instance, we say that 
$(a,b)\in \prescript{\zv}{\delta}Q^{\alpha}$ 
if 
$(\zv,\delta,\alpha,a,b)\in Q$.
Denote 
\begin{align}\label{EQForallInterpretation}
\prescript{\zv}{\delta}Q^{\forall}(y_1,y_2) &= 
\forall x \;
\prescript{\zv}{\delta}Q^{x,x,\dots,x}(y_1,y_2),\\
\label{EQForallForallInterpretation}
\prescript{\zv}{\delta}Q^{\forall\forall}(y_1,y_2) &= 
\forall x_1\dots\forall x_m \;
\prescript{\zv}{\delta}Q^{x_1,x_2,\dots,x_m}(y_1,y_2).
\end{align}
A tuple $(Q,D,B,C,\Delta)$ is called \emph{a mighty tuple I} if 
\begin{enumerate}
    \item $\prescript{\zv}{}\Delta\neq\varnothing$ for every $\zv\in A^{|A|}$;
    \item $\prescript{\zv}{\delta}B$, $\prescript{\zv}{\delta}C$, and $\prescript{\zv}{\delta}D$ are nonempty 
    for every $\zv\in A^{|A|}$ and $\delta\in\prescript{\zv}{}\Delta $;
    \item $\prescript{\zv}{\delta}Q^{\alpha}$ is an equivalence relation 
    on $\prescript{\zv}{\delta}D$ for every  $\zv\in A^{|A|}$, $\delta\in \prescript{\zv}{}\Delta$, and $\alpha\in A^{m}$;
   \item  $\prescript{\zv}{\delta}Q^{\forall}=\prescript{\zv}{\delta}D\times
    \prescript{\zv}{\delta}D$
    for every $\zv\in A^{|A|}$ and $\delta\in \prescript{\zv}{}\Delta $;
    \item $\prescript{\zv}{\delta}B$ and $\prescript{\zv}{\delta}C$
    are equivalence classes of $\prescript{\zv}{\delta}Q^{\forall\forall}$;    
    \item there exists $\zv\in A^{|A|}$ such that 
   $\prescript{\zv}{\delta}B\neq \prescript{\zv}{\delta}C$  for every $\delta\in\prescript{\zv}{}\Delta$. 
\end{enumerate}

The idea behind this definition is as follows.
For fixed $\zv$ and $\delta$ we have a parameterized binary relation $\prescript{\zv}{\delta}Q$, 
which is the full equivalence relation if 
$\alpha$ is a constant tuple and some equivalence relation otherwise. 
Relations $B$ and $C$ are just two equivalence classes that 
the EP has to connect by a complicated formula over $Q$ and 
the UP is trying to prevent this by choosing the parameters $\alpha$.

In the next two theorems and later in the paper we assume that 
$\mathrm{PSpace}\neq \Pi_{2}^{P}$. 

\begin{thm}\label{THMNonIdempotantClassification}
Suppose $\Gamma$ is a constraint language on a finite set $A$.
Then the following conditions are equivalent:
\begin{enumerate}
\item $\QCSP(\Gamma)$ is PSpace-complete;
\item  there exists a 
mighty tuple I q-definable over $\Gamma$.
\end{enumerate} 
\end{thm}

For constraint languages containing 
all constant relations we get an easier characterization.

\begin{thm}\label{THMIdempotantClassification}
Suppose $\Gamma\supseteq \{x=a\mid a\in A\}$ is a constraint language on a finite set $A$.
Then the following conditions are equivalent:
\begin{enumerate}
\item $\QCSP(\Gamma)$ is PSpace-complete;
\item  there exist an equivalence relation $\sigma$ on $D\subseteq A$ and $B,C\subsetneq A$
such that 
$B\cup C = A$ and 
    $\Gamma$
    q-defines the relations 
    $(y_{1},y_{2}\in D)\wedge(\sigma(y_1,y_2)\vee (x\in B))$ and 
    $(y_{1},y_{2}\in D)\wedge(\sigma(y_1,y_2)\vee (x\in C))$.
\end{enumerate} 
\end{thm}


The above theorems show that all the hardness cases have the same nature.
In the next section we provide a sketch of a proof 
of the PSpace-hardness in Theorem \ref{THMIdempotantClassification} 
for the case when $A = \{+,-,0,1\}$, 
$D = \{+,-\}$,
$\sigma$ is the equality on $D$,
$B = \{+,-,1\}$ and $C =\{+,-,0\}$,
but one may check that the same proof works word for word for the arbitrary $A$, $B$, $C$, $D$, and $\sigma$.
Moreover, as we show in Section \ref{SUBSECTIONPSPACEHARDNESS}
a very similar reduction 
works for the general case in Theorem \ref{THMNonIdempotantClassification}.

\subsection{Idea of the proof}

The proof of Theorem \ref{THMMainUPRestriction}, 
which is the main result of the paper, comes from the 
exponential-size CSP instance we discussed in 
Section \ref{SUBSECTIONReductionTOCSPInIntro}. 
Even though we cannot actually run any algorithm on it, 
one may ask whether 
it is solvable by local consistency methods.
Surprisingly, 
unless $\QCSP(\Gamma)$ is PSpace-hard, a slight modification of the exponential-size CSP instance 
can be solved even by arc-consistency, which is the biggest discovery of this paper (see Theorem \ref{THMOneConsistentReductionImpliesASolution}). 
The reader should not think that the trick is hidden in the modification, as we just replace the constraint language $\Gamma$ by 
the relation that is defined by 
the quantifier-free part of the instance (see Section \ref{SUBSECTIONInducedCSPInstances}).

This result does not give immediate consequences on the complexity of the QCSP as the instance is still of exponential size.
Nevertheless, we prove that 
unless the $\QCSP(\Gamma)$ is PSpace-hard, 
the arc-consistency can show that the instance has no solutions 
only by looking at the polynomial part of it, 
and this is the second main discovery of the paper (see Corollary \ref{CORBigTreeImpliesMightyTupleI} and Theorem \ref{THMFindSmallTree}).

Notice that most of the previous results on the complexity of the QCSP 
were proved for constraint languages with all constant relations 
$x=a$ \cite{AU-Chen-PGP,QC2017,QCSPMonstersSTOC}. Here, we obtain results for the general case replacing 
constant relations by $|A|$ new variables that are universally quantified at the very beginning, and therefore 
can be viewed  as external parameters.
The price we pay for the general case is that 
all the relations and instances are parameterized by two additional parameters $\zv$ and $\delta$, which you already saw in the classification of all PSpace-complete languages.

\subsection{Are there other complexity classes?}

As we now know,
$\QCSP(\Gamma)$ can be solvable in polynomial time, 
NP-complete, coNP-complete, DP-complete, 
$\Theta_{2}^{P}$-complete, $\Pi_{2}^{P}$-complete, and PSpace-complete.  
Knowing this, most of the readers probably expect 
infinitely many other complexity classes up to polynomial equivalence that can be expressed via the QCSP by fixing the constraint language.
In our opinion it is highly possible that these 7 complexity classes are everything we can attain, 
as we mostly expected new classes between $\Pi_{2}^{P}$ and PSpace, and now we know that there are none.
In this section we share our speculations on the question.

First, let us formulate 
what each of the classes means from the game theoretic point of view.
A QCSP instance is a game between the Universal Player (he) and the Existential Player (she):
he tries to make the quantifier-free part false, and 
she tries to make it true \cite{arora2009computational}. 
We say that a move of a player is \emph{trivial} 
if the optimal move can be calculated in polynomial time.
Then, those complexity classes just show how much they can interact with each other.

\noindent $\mathbf{P}$:  
    the play of both players is trivial;
    
   \noindent  $\mathbf{NP}$: only the EP plays, the play of 
    the UP is trivial;
    
   \noindent  $\mathbf{coNP}$: only the UP plays, the play of 
    the EP is trivial;
    
   \noindent  $\mathbf{DP}=\mathbf{NP}\wedge\mathbf{coNP}$:
each player plays their own game. Yes-instance: EP wins, UP loses;

\noindent $\mathbf{\Theta_{2}^{P}}
=(\mathbf{NP}\vee\mathbf{coNP})\wedge \dots\wedge (\mathbf{NP}\vee\mathbf{coNP})$: each player plays many games (no interaction), the result is a boolean combination of the results of those games;

\noindent $\mathbf{\Pi_{2}^{P}}$: the UP plays first,  then the EP plays; 

\noindent $\mathbf{PSpace}$: they play against each other (no restrictions).

We do not expect anything new between P and DP
because we do not have anything between P and NP for the CSP.
As the conjunction is given for free in the QCSP, 
whenever we can combine NP and coNP by anything else than conjunction, we expect to obtain a disjunction and, therefore, get $\Theta_{2}^{P}$.
Thus, the only
place where we expect new classes is between 
$\Theta_{2}^{P}$ and $\Pi_{2}^{P}$.
Nevertheless, even here we cannot imagine an interaction which is weaker than in $\Pi_{2}^{P}$,
and we must have some interaction 
as the class $\Theta_{2}^{P}$ is the strongest class where we just combine the results of the independent games.

Notice that everything we wrote in this section is only speculation, and we need a real proof of all of the dichotomies formulated in 
Section \ref{SUBSECTIONWhatIsInside}.

    


\subsection{Structure of the paper}

The rest of the paper is organized as follows. 
In Section \ref{SectionWithSimpleHardnessReduction} we 
show a concrete constraint language on a 4-element domain
whose complexity of the QCSP is PSpace-complete.
Then, in Section \ref{SECTIONPiTwoCompleteLanguage} 
we present a constraint language on a 6-element domain whose QCSP 
is $\Pi_{2}^{P}$-complete.

In Section \ref{SECTIONMainProof} 
we provide necessary definitions and derive all the main results of the 
paper from the statements that are proved later.
Here, 
we define 5 tuples of relations, called mighty tuples, such that 
the QCSP over any of them is PSpace-hard. 
In Section \ref{SECTIONFindingSolution}
we prove all the necessary statements 
under the assumption that 
a mighty tuple is not q-definable.
Finally, in Section \ref{SECTIONHardnessResults}
we prove PSpace-hardness for a mighty tuple I, 
show the equivalence and reductions between mighty tuples.

\section{The most general PSpace-hard constraint language}\label{SectionWithSimpleHardnessReduction} 

In this section for a concrete constraint language $\Gamma$ 
on a 4-element domain we show how to 
reduce the complement of Quantified-3-CNF 
to $\QCSP(\Gamma)$ and therefore prove PSpace-hardness. 
This constraint language is important because 
a similar reduction works for all the PSpace-hard 
cases.
 
\begin{lem}
Suppose 
$A = \{+,-,0,1\}$,
$R_{0}(y_1,y_2,x) = (y_1,y_2\in\{+,-\})\wedge (x=0\rightarrow y_1=y_2)$,
$R_{1}(y_1,y_2,x) = (y_1,y_2\in\{+,-\})\wedge (x=1\rightarrow  y_1=y_2)$,
$\Gamma = \{R_0,R_1,\{+\},\{-\}\}$.
Then $\QCSP(\Gamma)$ is PSpace-hard.
\end{lem}

This example may be viewed as the weakest constraint language 
whose QCSP is PSpace-hard. 
It can be seen from the definition that 
the UP should only play $x$-variables (the last coordinates of $R_{0}$ and $R_1$)
and the EP should only play $y$-variables (the first two coordinates),
and originally we did not see how they can interact with each other and, therefore, did not expect the QCSP to be PSpace-hard. 
Below we demonstrate on a concrete example 
how the EP can control the moves of the UP and therefore, 
interact in the area of the UP.

\begin{sketch}
We build a reduction from 
the complement of the Quantified-3-CNF.
Let the sentence be
$$\neg (\exists x_1\forall x_2\exists x_3\;\;((x_1\vee \overline x_2\vee x_3)\wedge 
(\overline x_1\vee x_2\vee \overline x_3)
\wedge 
(x_1\vee \overline x_2\vee \overline x_3)).$$
Instead of formulas we draw graphs whose 
vertices are variables and edges are relations.
$R_0(y_1,y_2,x)$ is drawn as
a red edge from $y_1$ to $y_2$ labeled with $x$.
Similarly 
$R_1(y_1,y_2,x)$ is drawn as
a blue edge from $y_1$ to $y_2$ labeled with $x$. 
If a vertex has no name, then we assume that the variable is 
existentially quantified after all other variables are quantified.
If the vertex is marked with $+$ or $-$, then we assume 
that the corresponding variable is equal to $+$ or $-$ respectively.
In Figure \ref{ExampleOfFormulasEncoding}
you can see an example of a graph and the corresponding formula.

\begin{figure}
    \centering
\begin{tikzpicture}[scale = 1]
\node at (2,1)[black,scale =1.0](x) {
$y_1$}; 
\node at (8,1)[black,scale =1.0](y) {$+$};
\draw[thick,blue!90,text = black, line width=0.5mm, in =170, out =10] (x) 
to node[above,scale =1.0]{$x_1$} node[fill=white,scale=0.5,text = blue]{1}(4,1);
\draw[thick,red!90,text = black, line width=0.5mm, in =170, out =10] (4,1) 
to node[above,scale =1.0]{$x_2$} 
node[fill=white,scale=0.5,text = red]{0}(6,1);
\draw[thick,blue!90,text = black, line width=0.5mm, in =170, out =10] (6,1) 
to node[above,scale =1.0]{$x_3$} node[fill=white,scale=0.5,text = blue]{1}(y);
\end{tikzpicture}
    \caption{A graph for
    $\exists u_1\exists u_2 \exists u_3
R_{1}(y_1,u_1,x_1)\wedge 
R_{0}(u_1,u_2,x_2)\wedge 
R_{1}(u_2,u_3,x_3)\wedge (u_3=+)$}
    \label{ExampleOfFormulasEncoding}
\end{figure}

Thus, these graphs can be viewed as electrical circuits 
where the ends of an edge are connected (equal) whenever 
the variable written on it has the corresponding value.  
Then the encoding of 
the quantifier-free part
$$((x_1\vee \overline x_2\vee x_3)\wedge 
(\overline x_1\vee x_2\vee \overline x_3)
\wedge 
(x_1\vee \overline x_2\vee \overline x_3))$$
is shown in Figure \ref{QuantifierFreePartEncoding}.
If we assume that all the $x$-variables are from $\{0,1\}$,
which will be the case, 
then 
the formula in Figure \ref{QuantifierFreePartEncoding}
holds if and only if 
the 3-CNF does not hold. 
In fact, if the 3-CNF holds then $+$ is connected (equal) to 
$-$ through three edges, which gives a contradiction.
If the formula in Figure \ref{QuantifierFreePartEncoding}
holds, then at some point we go from $+$ to $-$, which 
means that the corresponding clause does not hold.

\begin{figure}
    \centering
\begin{tikzpicture}[scale = 1]
\node at (1,1)[black,scale =1.0](x) {
$+$}; 
\node at (8.5,1)[black,right, scale =1.0](y) {$-$};

\draw[thick,blue!90,text = black, line width=0.5mm, in =130, out =40] (x) to node[above,scale =1.0]{$x_1$} node[fill=white,scale=0.5,text = blue]{1}(3.5,1);

\draw[thick,blue!90,text = black, line width=0.5mm, in =240, out =-40] (x) to node[above,scale =1.0]{$x_3$}node[fill=white,scale=0.5,text = blue]{1} (3.5,1);

\draw[thick,red!90,text = black, line width=0.5mm, in =185, out =-5]
(x) to node[above,scale =1.0]{$x_2$} node[fill=white,scale=0.5,text = red]{0}(3.5,1);

\draw[thick,red!90,text = black, line width=0.5mm, in =130, out =50] (3.5,1) to node[above,scale =1.0]{$x_1$} node[fill=white,scale=0.5,text = red]{0}(6,1);
\draw[thick,red!90,text = black, line width=0.5mm, in =240, out =-60] (3.5,1) 
to node[above,scale =1.0]{$x_3$} node[fill=white,scale=0.5,text = red]{0}(6,1);
\draw[thick,blue!90,text = black, line width=0.5mm, in =185, out =-5] (3.5,1) 
to node[above,scale =1.0]{$x_2$} node[fill=white,scale=0.5,text = blue]{1}(6,1);

\draw[thick,blue!90,text = black, line width=0.5mm, in =140, out =50] (6,1) 
to node[above,scale =1.0]{$x_1$} node[fill=white,scale=0.5,text = blue]{1}(y);
\draw[thick,red!90,text = black, line width=0.5mm, in =220, out =-60] (6,1) 
to node[above,scale =1.0]{$x_3$} node[fill=white,scale=0.5,text = red]{0}(y);
\draw[thick,red!90,text = black, line width=0.5mm, in =185, out =-5] (6,1) 
to node[above,scale =1.0]{$x_2$} 
node[fill=white,scale=0.5,text = red]{0}(y);
\end{tikzpicture}
    \caption{A graph expressing
    $\neg((x_1\vee \overline x_2\vee x_3)\wedge 
(\overline x_1\vee x_2\vee \overline x_3)
\wedge 
(x_1\vee \overline x_2\vee \overline x_3)).$}
    \label{QuantifierFreePartEncoding}
\end{figure}

If we add universal quantifiers to the formula in Figure \ref{QuantifierFreePartEncoding}
we get 
\begin{align*}
\forall x_1\forall x_2\forall x_3\;\;\neg((x_1\vee \overline x_2\vee x_3)\wedge &
(\overline x_1\vee x_2\vee \overline x_3)
\wedge 
(x_1\vee \overline x_2\vee \overline x_3))=\\
&\neg (\exists x_1\exists x_2\exists x_3\; 
((x_1\vee \overline x_2\vee x_3)\wedge 
(\overline x_1\vee x_2\vee \overline x_3)
\wedge 
(x_1\vee \overline x_2\vee \overline x_3)))    
\end{align*}
Notice that it does not make sense
for the UP 
to play values $+$ and $-$ because 
the relations $R_0$ and $R_1$ hold whenever 
the last coordinate is from $\{+,-\}$.
Thus, we already encoded the complement to 3-CNF-Satisfability, 
which means that $\QCSP(\Gamma)$ is coNP-hard.

To show PSpace-hardness we need to add existential quantifiers.
We cannot just add $\exists x_2$ because the obvious choice for the 
EP would be $+$ or $-$.
As shown in Figure \ref{Q3CNFExampleEncoding},
whenever we want to add $\exists x_2$, we add a new existential variable 
$y_2$ and universally quantify $x_2$.
The goal of the UP is to connect $+$ and $-$. 
Hence, if the EP plays $y_2=+$, then 
the only reasonable choice for the UP is to play 
$x_2 = 1$;
and if 
the EP plays $y_2=-$, then 
the UP must play 
$x_2 = 0$.
Thus, the EP controls the moves of the UP, which is equivalent to 
the EP playing on the set $\{0,1\}$.

\begin{figure}
    \centering
\begin{tikzpicture}[scale = 1]
\node at (1+2.5+2.5,1)[black,scale =1.0](u2){};
\node at (1+2.5+2.5+2.5,3)[black,scale =1.0](u3){};
\node at (1+2.5+2.5+2.5,1)[black,scale =1.0](u4){};

\node at (-0.5,2)[black,scale =1.0] {$\forall x_1\exists y_2 \forall x_2\forall x_3$};

\node at (1+2.5,3)[black,left,scale =1.0](plus) {
$+$}; 
\node at (1+2.5,1)[black,left, scale =1.0](minus) {$-$};
\draw[thick,red!90,text = black, line width=0.5mm, in =140, out =-15] (plus) to node[above right,scale =1.0]{$x_2$} 
node[fill=white,scale=0.5,text = red]{0}
(u2);
\draw[thick,blue!90,text = black, line width=0.5mm, in =190, out =-10] (minus) to node[above,scale =1.0]{$x_2$} node[fill=white,scale=0.5,text = blue]{1}(u2);
\node at (1+2.5+2.5,3)[black,scale =1.0](y1){$y_2$};

\draw[thick,blue!90,text = black, line width=0.5mm, in =140, out =40] (u2) to node[above,scale =1.0]
{$x_1$} node[fill=white,scale=0.5,text = blue]{1} (u4);
\draw[thick,red!90,text = black, line width=0.5mm, in =185, out =-5] (u2) to node[above,scale =1.0]{$x_2$} 
node[fill=white,scale=0.5,text = red]{0}
(u4)
;
\draw[thick,blue!90,text = black, line width=0.5mm, in =220, out =-40] (u2) to node[above,scale =1.0]{$x_3$} 
node[fill=white,scale=0.5,text = blue]{1}(u4);

\draw[thick,blue!90,text = black, line width=0.5mm, in =140, out =40] (y1) to node[above,scale =1.0]{$x_1$} node[fill=white,scale=0.5,text = blue]{1}(u3);
\draw[thick,red!90,text = black, line width=0.5mm, in =185, out =-5] (y1) to node[above,scale =1.0]{$x_2$} node[fill=white,scale=0.5,text = red]{0}(u3);
\draw[thick,red!90,text = black, line width=0.5mm, in =220, out =-40] (y1) to node[above,scale =1.0]{$x_3$} 
node[fill=white,scale=0.5,text = red]{0}(u3);

\draw[thick,red!90,text = black, line width=0.5mm, in =0, out =0] (u4) to node[right,scale =1.0]{$x_3$} node[fill=white,scale=0.5,text = red]{0}(u3);
\draw[thick,blue!90,text = black, line width=0.5mm, in =185+90+10, out =-15+90] (u4) to node[right,scale =1.0]{$x_2$} node[fill=white,scale=0.5,text = blue]{1}(u3);
\draw[thick,red!90,text = black, line width=0.5mm, in =140+90-5, out =40+90+5] (u4) to node[right,scale =1.0]{$x_1$} node[fill=white,scale=0.5,text = red]{0}(u3);
\end{tikzpicture}
    \caption{A graph expressing
    $\forall x_1\exists x_2\forall x_3\;\;\neg((x_1\vee \overline x_2\vee x_3)\wedge 
(\overline x_1\vee x_2\vee \overline x_3)
\wedge 
(x_1\vee \overline x_2\vee \overline x_3))$}
    \label{Q3CNFExampleEncoding}
\end{figure}

Thus, we encoded 
\begin{align*}
\forall x_1\exists x_2\forall x_3\;\;\neg((x_1\vee \overline x_2\vee x_3)&\wedge 
(\overline x_1\vee x_2\vee \overline x_3)
\wedge 
(x_1\vee \overline x_2\vee \overline x_3))=\\
&\neg (\exists x_1\forall x_2\exists x_3\;\;((x_1\vee \overline x_2\vee x_3)\wedge 
(\overline x_1\vee x_2\vee \overline x_3)
\wedge 
(x_1\vee \overline x_2\vee \overline x_3)),
\end{align*}
which is a complement to the Quantified-3-CNF.
\end{sketch}

As we saw in 
Theorem \ref{THMIdempotantClassification} 
two relations equivalent to $R_0$ and $R_1$ 
can be q-defined from $\Gamma$ 
whenever $\Gamma$ contains all constant relations 
and $\QCSP(\Gamma)$ is not in $\Pi_2^{P}$.
The criterion for the general case 
is formulated using mighty tuples, 
but as you can see 
in the proof of Theorem 
\ref{THMMightyTupleIPSpaceHardness}
in Section \ref{SUBSECTIONPSPACEHARDNESS}
we still build two relations $R_{0}$ and $R_{1}$
and use almost the same construction to prove PSpace-hardness.

\section{$\Pi_2^{P}$-complete constraint language}
\label{SECTIONPiTwoCompleteLanguage}

In this section we define a concrete constraint language $\Gamma$ on 
a 6-element domain $A = \{0,1,2,0',1',2'\}$
such that 
$\QCSP(\Gamma)$ is $\Pi_{2}^{P}$-complete.

First, we define two ternary relations $\mathrm{AND}_{2}$ and
$\mathrm{OR}_{2}$ 
corresponding to the operations $\wedge$ and $\vee$ on $\{0,1\}$.
If one of the first two coordinates is from
$\{2,0',1',2'\}$, then the remaining elements can be chosen arbitrary, 
i.e., 
\begin{align*}
\{2,0',1',2'\}\times A\times A\subseteq &\mathrm{AND}_{2}, 
&\{2,0',1',2'\}\times A\times A\subseteq \mathrm{OR}_{2},\\
A\times \{2,0',1',2'\}\times A\subseteq &\mathrm{AND}_{2},
&A\times \{2,0',1',2'\}\times A\subseteq \mathrm{OR}_{2}.
\end{align*}
If $a,b\in\{0,1\}$, then 
$(a,b,c)\in \mathrm{AND}_{2}\Rightarrow
(a\wedge b=c)$ and 
$(a,b,c)\in \mathrm{OR}_{2}\Rightarrow
(a\vee b=c)$.
In other words 
$$\mathrm{AND}_{2}\cap (\{0,1\}\times\{0,1\}\times A) = 
\begin{pmatrix}
0 & 0 & 1 & 1 \\
0 & 1 & 0 & 1 \\
0 & 0 & 0 & 1
\end{pmatrix},
\mathrm{OR}_{2}\cap (\{0,1\}\times\{0,1\}\times A) = 
\begin{pmatrix}
0 & 0 & 1 & 1 \\
0 & 1 & 0 & 1\\ 
0 & 1 & 1 & 1
\end{pmatrix},$$
where each matrix should be understood 
as the set of tuples written as columns.
The remaining four relations are defined by 
\begin{align*}
\oneinthree &= \{(2',2',2'),(1',0',0'),(0',1',0'),(0',0',1')\},\\    
\delta_{0} &= \{1\}\times\{0',2'\}\cup(A\setminus\{1\})\times \{0',1',2'\},\\
\delta_{1} &= \{1\}\times\{1',2'\}\cup(A\setminus\{1\})\times \{0',1',2'\},\\
\epsilon &= \{0\}\times \{0',1'\}\cup(A\setminus\{0\})\times \{0',1',2'\}.
\end{align*}
The relation $\oneinthree$ is the usual relation $\mathrm{1IN3}$ on $\{0',1'\}$ 
with an additional tuple $(2',2',2')$.
The relations 
$\delta_{0}$, $\delta_{1}$ and $\epsilon$ 
can also be viewed as
\begin{align*}
\delta_{0}(x,y) &= (y\in\{0',1',2'\})\wedge (x = 1\Rightarrow y\neq 1'),\\
\delta_{1}(x,y) &= (y\in\{0',1',2'\})\wedge (x = 1\Rightarrow y\neq 0'),\\
\epsilon(x,y) &= (y\in\{0',1',2'\})\wedge (x = 0\Rightarrow y\neq 2').
\end{align*}

Note that $\delta_{1}$ can be derived from 
$\delta_{0}$ and $\oneinthree$ by the formula 
$$\delta_{1}(x,y) = \exists u_1 \exists u_2\exists u_3 \;
\delta_{0}(x,u_1)\wedge \oneinthree(y,u_1,u_2)\wedge 
\oneinthree(u_2,u_2,u_3).$$


Let 
$\Gamma = \{\mathrm{AND}_{2},\mathrm{OR}_{2},\oneinthree,\delta_{0},\delta_{1},\epsilon\}$.


\begin{lem}
$\QCSP(\Gamma)$ is $\Pi_{2}^{P}$-hard.
\end{lem}

\begin{proof}
First, we derive the $\mathrm{OR}$ relation of larger arity by  
$$\mathrm{OR}_{n+1}(x_1,\dots,x_{n+1},y)= 
\exists y' \; \mathrm{OR}_{n}(x_1,\dots,x_{n},y')
\wedge \mathrm{OR}_{2}(y',x_{n+1},y).$$
Then we assume that  
$\mathrm{OR}_{n}$ is in our language.
To prove the $\Pi_{2}^{P}$-hardness we build a reduction 
from $\Pi_{2}$-$\QCSP(\mathrm{1IN3})$, 
where $\mathrm{1IN3} = \{(1,0,0),(0,1,0),(0,0,1)\}$.
Let 
$$\Phi = \forall x_1 \dots\forall x_m 
\exists x_{m+1}\dots\exists x_{n}\;
\mathrm{1IN3}(x_{i_1},x_{j_1},x_{k_1})\wedge \dots\wedge
\mathrm{1IN3}(x_{i_s},x_{j_s},x_{k_s}).$$
Since we can always add dummy variables, we assume that $\Phi$ has at least two universally quantified variables 
for the general construction to make sense.
The problem of checking whether $\Phi$ holds is $\Pi_{2}^{P}$-complete \cite{arora2009computational}.
We will encode $\Phi$ by the following instance of $\QCSP(\Gamma)$.
\begin{align*}
\Psi = \forall x_{1}^{0}\forall x_{1}^1\dots\forall x_{m}^0\forall x_{m}^{1}
\;\; \exists x_{1}\exists x_2\dots\exists x_n\;\;\exists z_{1}\dots \exists z_{n}\;\;\exists z\;\;
\bigwedge\limits_{i=1}^{m}\left(\delta_{0}(x_{i}^0,x_{i})\wedge \delta_{1}(x_{i}^1,x_{i})\right) \wedge &\\
\;
  \bigwedge\limits_{i=1}^{m}
 \mathrm{AND}_{2}(x_i^{0},x_i^{1},z_{i})
 \wedge \mathrm{OR}_{m}(z_{1}\dots,z_{n},z)\wedge 
\bigwedge\limits_{i=1}^{n} \epsilon(z,x_{i})
\wedge 
\bigwedge\limits_{\ell=1}^{s} \oneinthree(x_{i_\ell},x_{j_\ell}&,x_{k_\ell})
\end{align*}

\begin{figure}
    \centering
\begin{tikzpicture}[
trnode/.style={isosceles triangle,isosceles triangle apex angle=120,draw,draw=green!60},
andnodedashed/.style={regular polygon,regular polygon sides=3, draw,draw=green!60!black, shape border rotate= 30, fill=green!20, very thick, minimum size=7mm,inner sep = -3pt,dashed},
andnode/.style={regular polygon,regular polygon sides=3, draw,draw=green!60!black, shape border rotate= 30, fill=green!20, very thick, minimum size=7mm,inner sep = -3pt},
rnode/.style={rectangle, draw=red!60, fill=red!5, very thick, minimum size=5mm},
]


\DrawAnd{1}{7}{1}
\DrawAnd{1}{5}{2}

\DrawAndDotted{1}{2.5}{i}

\DrawAnd{1}{0}{m}

\node  (xx10)             at (4,7.5)                 {$x_{1}^{0}$};
\DrawBlock{4}{6}{10}{$=1$ $\downarrow$ $\neq 1'$}
\draw[->] (xx10) to (b10);

\node  (xx11)             at (5,7.5)                 {$x_{1}^{1}$};
\DrawBlock{4+1}{6}{11}{$=1$ $\downarrow$ $\neq 0'$}
\draw[->] (xx11) to (b11);

\node  (xx20)             at (6.2,7.5)                 {$x_{2}^{0}$};
\DrawBlock{6.2}{6}{20}{$=1$ $\downarrow$ $\neq 1'$}
\draw[->] (xx20) to (b20);

\node  (xx21)             at (7.2,7.5)                 {$x_{2}^{1}$};
\DrawBlock{7.2}{6}{21}{$=1$ $\downarrow$ $\neq 0'$}
\draw[->] (xx21) to (b21);

\node  (xxi0)             at (8.4+0.2,7.5)                 {$x_{i}^{0}$};
\DrawBlockDashed{8.4+0.2}{6}{i0}{$=1$ $\downarrow$ $\neq 1'$}
\draw[->,dashed] (xxi0) to (bi0);

\node  (xxi1)             at (9.4+0.2,7.5)                 {$x_{i}^{1}$};
\DrawBlockDashed{9.4+0.2}{6}{i1}{$=1$ $\downarrow$ $\neq 0'$}
\draw[->,dashed] (xxi1) to (bi1);

\node  (xxn0)             at (9.4+1.2+0.2+0.2,7.5)                 {$x_{m}^{0}$};
\DrawBlock{9.4+1.2+0.2+0.2}{6}{n0}{$=1$ $\downarrow$ $\neq 1'$}
\draw[->] (xxn0) to (bn0);

\node  (xxn1)             at (1+9.4+1.2+0.2+0.2,7.5)                 {$x_{m}^{1}$};
\DrawBlock{1+9.4+1.2+0.2+0.2}{6}{n1}{$=1$ $\downarrow$ $\neq 0'$}
\draw[->] (xxn1) to (bn1);

\DrawBlock{4.2}{2}{b1}{$=0$ $\downarrow$ $\neq 2'$}
\draw[->] (4.2,3.5) to (bb1);
\DrawBlock{6.4}{2}{b2}{$=0$ $\downarrow$ $\neq 2'$}
\draw[->] (6.4,3.5) to (bb2);
\DrawBlockDashed{8.4+0.2+0.2}{2}{bi}{$=0$ $\downarrow$ $\neq 2'$}
\draw[->,dashed] (8.4+0.2+0.2,3.5) to (bbi);
\DrawBlock{9.4+1.2+0.2+0.2+0.2}{2}{bn}{$=0$ $\downarrow$ $\neq 2'$}
\draw[->] (9.4+1.2+0.2+0.2+0.2,3.5) to (bbn);

\DrawBlock{9.4+1.2+0.2+0.2+0.2+1.6}{2}{bnp}{$=0$ $\downarrow$ $\neq 2'$}
\draw[->] (9.4+1.2+0.2+0.2+0.2+1.6,3.5) to (bbnp);
\draw[->] (bbnp) to (9.4+1.2+0.2+0.2+0.2+1.6,0.58);

\DrawBlockDashed{9.4+1.2+0.2+0.2+0.2+2.4}{2}{bnpp}{$=0$ $\downarrow$ $\neq 2'$}

\draw[->,dashed] (9.4+1.2+0.2+0.2+0.2+2.4,3.5) to (bbnpp);
\draw[->,dashed] (bbnpp) to (9.4+1.2+0.2+0.2+0.2+2.4,0.58);

\DrawBlock{9.4+1.2+0.2+0.2+0.2+3.2}{2}{bm}{$=0$ $\downarrow$ $\neq 2'$}

\draw[->] (9.4+1.2+0.2+0.2+0.2+3.2,3.5) to (bbm);
\draw[->] (bbm) to (9.4+1.2+0.2+0.2+0.2+3.2,0.58);

\node (xxx1) at (4.81,0.42) {};
\draw[->,out=-70,in = 90] (b10.south) to (xxx1);
\draw[->] (b11) to (xxx1);
\draw[->,out=-70,in = 90] (bb1.south) to (xxx1);

\node (xxx2) at (6.87,0.42) {};
\draw[->,out=-70,in = 90] (b20.south) to (xxx2);
\draw[->] (b21) to (xxx2);
\draw[->,out=-70,in = 90] (bb2.south) to (xxx2);

\node (xxxi) at (9.3,0.42) {};
\draw[->,out=-70,in = 90,dashed] (bi0.south) to (xxxi);
\draw[->,dashed] (bi1) to (xxxi);
\draw[->,out=-70,in = 90,dashed] (bbi.south) to (xxxi);

\node (xxxn) at (11.75,0.42) {};
\draw[->,out=-70,in = 90] (bn0.south) to (xxxn);
\draw[->] (bn1) to (xxxn);
\draw[->,out=-70,in = 90] (bbn.south) to (xxxn);

\draw[->,out=-70,in = 90] (bbn.south) to (xxxn);

\node[rectangle, draw=blue!60, fill=blue!5, very thick, minimum size=5mm,text width=11.3cm]
at (9,0){
{\color{blue}
\;\;\;\;\;\;\;\;\;\;\;$x_1$\;\;\;\;\;\;\;\;\;\;\;\;\;\;
$x_2$\;\;\;\;\;\;\ldots\;\;\;\;\;
$x_i$\;\;\;\;\;\;\ldots\;\;\;\;
$x_m$\;\;\;
$x_{m+1}$\;$\ldots$\;\;$x_n$}
$\oneinthree(x_{i_1},x_{j_1},x_{k_1})\wedge
\oneinthree(x_{i_2},x_{j_2},x_{k_2})\wedge 
\dots\wedge \oneinthree(x_{i_s},x_{j_s},x_{k_s})
$};

\node[isosceles triangle,isosceles triangle apex angle=147,draw,draw=green!60!black,fill=green!20,very thick] (or) at (3,3.5) {$\mathrm{OR}_{m}$};

\draw[->] (and1) to (or.100);
\draw[->] (and2) to (or.112.5);
\draw[->,dashed] (andi) to (or.237);
\draw[->] (andm) to (or.260);

\draw[->] (or) to (5.5,3.5);
\draw[->] (5.5,3.5) to (7.5,3.5);
\draw[->,dashed] (7.5,3.5) to (10.5,3.5);

\draw[->] (10.5,3.5) to (12.8,3.5);
\draw[->,dashed] (12.8,3.5) to (9.4+1.2+0.2+0.2+0.2+3.2,3.5);




\end{tikzpicture}
    \caption{Reduction from $\Pi_{2}$-$\QCSP(\mathrm{1IN3})$ to 
    $\QCSP(\Gamma)$.}
    \label{PiTwoReductionDiagram}
\end{figure}
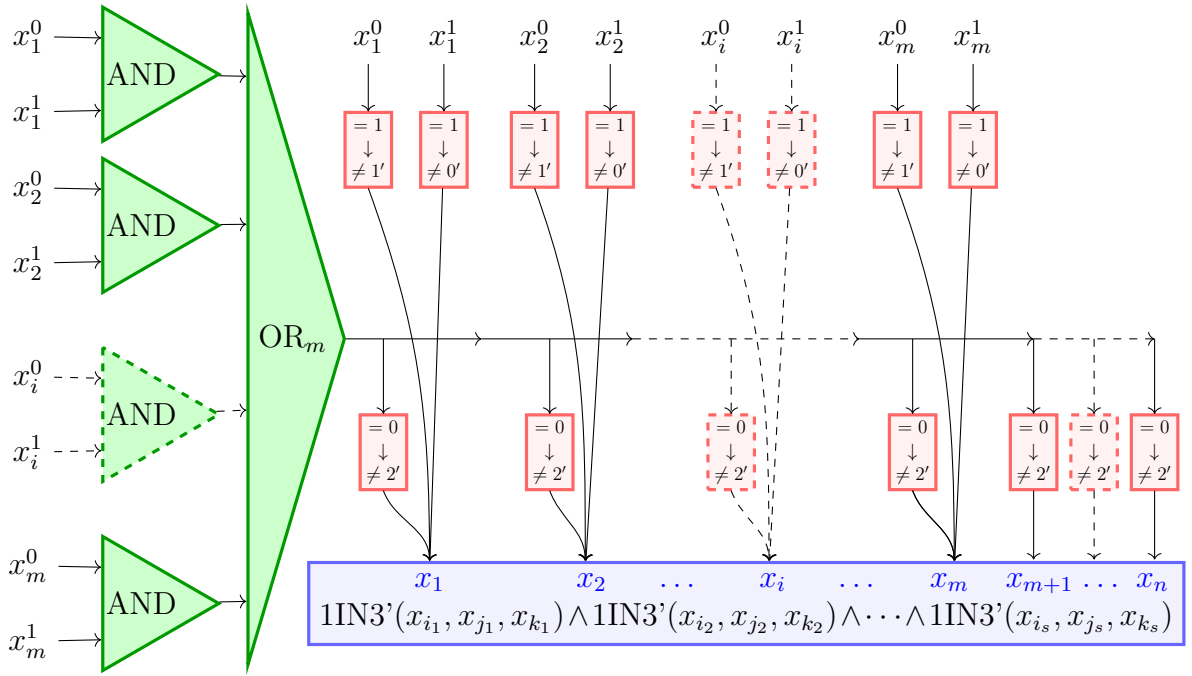

The quantifier free part of $\Psi$ is shown in Figure \ref{PiTwoReductionDiagram},  where triangle elements are $\mathrm{AND}_{2}$ 
and $\mathrm{OR}_{m}$, rectangular elements are $\delta_{0}$, $\delta_{1}$, 
and $\epsilon$, and the big block at the bottom is 
just the conjunction of the corresponding 
$\oneinthree$-relations.
The variable $x_{i}$ in $\Psi$ 
will take values from $\{0',1',2'\}$, 
which makes the use of universal quantifiers directly impossible. 
For the universal variables to be applicable we introduce two new variables 
$x_{i}^0$ and $x_{i}^1$ for each $x_{i}$ where $i\in\{1,2,\dots,m\}$. 
We expect exactly one of the two values $x_{i}^0$ and $x_{i}^1$ to be equal to 1.
$x_{i}^{0}=1$ means that $x_{i} = 0'$,  
$x_{i}^{1}=1$ means that $x_{i} = 1'$. 
Using the relations $\delta_{0}$ and $\delta_{1}$ we make 
$x_{i}$ equal to the value we need.
Notice that 
$x_{i}$ can also be equal to $2'$ and this value should be forbidden by $\epsilon$. 
Let us prove that $\Phi$ and $\Psi$ are equivalent.

$\Phi\Rightarrow \Psi$.
Suppose we have a winning strategy for the Existential Player (EP) in $\Phi$, 
let us define a winning strategy for the EP in $\Psi$.
If the Universal Player (UP) in $\Psi$ plays 
$x_{i}^0 = x_{i}^1 = 1$, or $x_{i}^0\notin\{0,1\}$, or $x_{i}^1\notin\{0,1\}$ for some $i$, then  
the winning strategy for the EP is to choose $x_{1} = \dots = x_{n} = 2'$.
Then the $\oneinthree$-block of $\Psi$ is satisfied as 
$(2',2',2',\dots,2')$ is its trivial solution.
Only value 1 restricts the second coordinate of the relations $\delta_{0}$ and $\delta_{1}$,
hence the best choice for the UP is to 
make $(x_{i}^0,x_{i}^1)\in\{(0,1),(1,0)\}$.
We interpret $(x_{i}^0,x_{i}^1)=(0,1)$ as $x_{i} = 1$ and 
$(x_{i}^0,x_{i}^1)=(1,0)$ as $x_{i} =0$.
Then the EP in $\Psi$ plays 
$x_{1},\dots,x_{m}$ according to $(x_{i}^0,x_{i}^1)$ and 
plays $x_{m+1},\dots,x_{n}$  
just copying the moves of the EP in $\Phi$
but $0'$ instead of $0$ and $1'$ instead of $1$.
Since the quantifier-free part of $\Phi$ is satisfied, 
the $\oneinthree$-block of $\Psi$ is also satisfied.

$\Psi\Rightarrow \Phi$. Suppose 
the UP in $\Phi$ plays 
$x_1,\dots,x_m$. 
Let the UP in $\Psi$ play $x_{i}^0 =1$, $x_{i}^1 =0$ if $x_{i} = 0$ and 
$x_{i}^0=0$, $x_{i}^1 =1$ if $x_{i} = 1$.
Then the EP in $\Psi$ should play 
only values from $\{0',1'\}$ for $x_{1},\dots,x_{n}$.
The EP in $\Phi$ just copies the moves of the EP in $\Psi$ playing $0$ instead of $0'$, and $1$ instead of $1'$.
The satisfiability of the $\oneinthree$-block of $\Psi$ 
implies the satisfiability of the quantifier-free part of $\Phi$.
\end{proof}

\begin{lem}
$\QCSP(\Gamma)$ is in $\Pi_{2}^{P}$.
\end{lem}

\begin{proof}
One of the definitions of the class 
$\Pi_{2}^{P}$ is $\mathrm{coNP}^{\mathrm{NP}}$, that is the class of problem solvable 
by a nondetermenistic Turing machine 
augmented by an oracle for some NP-complete problem \cite{arora2009computational}.
Thus, to prove the membership in $\Pi_{2}^P$, it is sufficient to show that an optimal strategy 
for the EP can be calculated in polynomial time using 
the NP-oracle. 

Suppose we have an instance 
$\forall x_1\exists y_1 \forall x_2 \exists y_2 \dots \forall x_{n} \exists y_{n} \Phi$.
Suppose 
the variables 
$x_{1},\dots,x_{i}$ and $y_{1},\dots,y_{i-1}$
are already evaluated with 
$a_1,\dots,a_{i}$ and $b_{1},\dots,b_{i-1}$, respectively. 
We need to calculate an optimal value $b_{i}$ for $y_{i}$, i.e.,
suppose the EP can still win in this position then
she should be able to win after making the move $y_{i} = b_{i}$.

We will explain the algorithm first and then we argue why it returns an optimal move.
Using the NP-oracle for every $d\in A$ we check the satisfiability of the instance
\begin{align}\label{FORMULAWITHONEs}
    \Phi
\wedge \bigwedge_{j=1}^{j=i} x_{j} = a_{j}
\wedge \bigwedge_{j=i+1}^{j=n} x_{j} = 1
\wedge \bigwedge_{j=1}^{j=i-1} y_{j} = b_{j}
\wedge y_{i} = d.
\end{align}
Thus, we just send the previous variables to their values, evaluate all further 
universal variables to 1 and the variable $y_{i}$ to $d$.
Let $D$ be the set of $d$ such that 
(\ref{FORMULAWITHONEs}) has a solution.
Then we choose an optimal value as follows:
\begin{enumerate}
    \item If $D$ is empty, then the EP cannot win and an optimal move does not exist.
    \item If $D=\{b\}$ for some $b$ then $y_{i}=b$ is the optimal move.
    \item If $c\in D\cap \{0',1',2\}$ then 
    $y_i = c$ is an optimal move.
\end{enumerate}

The algorithm is obviously polynomial. Hence it is sufficient to prove that
cases 1-3 cover all the possible cases and 
the move chosen by 3 is optimal.

The binary operation $g(x,y) = 
\begin{cases}
x, & \text{if $y=1$}\\
x, & \text{if $x=y$}\\
y, & \text{if $y\in\{0',1'\}$}\\
x, & \text{if $x\in\{0',1'\}$ and $y=2'$}\\
2, & \text{otherwise}
\end{cases}$ preserves all the relations from $\Gamma$ and all constant relations,
which follows from the following properties of $g$ and manual checking of some cases:
\begin{itemize}
    \item $g$ either returns the first variable or an element of $\{0',1',2\}$.
    \item $g$ preserves $\{0',1',2'\}$ and $g$ restricted to $\{0',1',2'\}$ returns 
    the last non-$2'$ value if it exists. 
    \item $g$ returns 1 only on the tuple $(1,1)$ and it returns $2'$ only on the tuple $(2',2')$.
\end{itemize}
This implies (see \cite{bartopolymorphisms}) that 
$g$ preserves the solution set of any instance of $\CSP(\Gamma)$.

The only cases that are not covered by cases 1-3 
are $\{0,1\}\subseteq D$,
$\{0,2'\}\subseteq D$, or 
$\{1,2'\}\subseteq D$.
Since $g(1,0)=g(0,2') = g(1,2')=2$, we have $2\in D$, which means that $D$ satisfies case 3.

It remains to show that any value
$c\in D\cap \{0',1',2\}$ is an optimal move for the EP.
Let 
$y_{j}:=f_{j}(x_{i+1},\dots,x_{j})$, where $j=i,i+1,\dots,n$, be the winning strategy 
for the EP.  
Then 
the tuple 
\begin{align}\label{FORMULASkolemFunctions}
    (a_{1},b_1,\dots,a_{i-1},b_{i-1},a_{i},f_{i}(), a_{i+1},
f_{i+1}(a_{i+1}),a_{i+2},f_{i+2}(a_{i+1},a_{i+2}),\dots,
a_{n},f_{n}(a_{i+1},\dots,a_n))
\end{align} 
is a solution of the quantifier-free part $\Phi$ 
for any $a_{i+1},\dots,a_{n}\in A$.
By the definition of 
$D$
the tuple
\begin{align}\label{FORMULAAfterOnes}
(a_{1},b_1,\dots,a_{i-1},b_{i-1},a_{i},c, 1,
c_{i+1},1,c_{i+2},\dots,
1,c_{n})
\end{align} 
is a solution of $\Phi$ for some $c_{i+1},\dots,c_{n}\in A$.
Applying $g$ 
to the tuples (\ref{FORMULASkolemFunctions})
and
(\ref{FORMULAAfterOnes}) coordinate-wise we derive that the tuple 
\begin{align*}
  (a_{1},b_1,\dots,a_{i-1},b_{i-1},a_{i},c,& \\ a_{i+1},
g(f_{i+1}(a_{i+1}),&c_{i+1}),a_{i+2},g(f_{i+2}(a_{i+1},a_{i+2}),c_{i+2}),\dots,
a_{n},g(f_{n}(a_{i+1},\dots,a_n),c_{n}))
\end{align*} 
is a solution of $\Phi$ for any $a_{i+1},\dots,a_{n}\in A$.
Hence $y_{i} = c$ is an optimal move for the EP, which completes the proof.
\end{proof}

Thus, we proved Theorem \ref{THMPiTwoCompleteLanguage} from 
Section \ref{SECTIONMainResults}.


\section{Proof of the main result}\label{SECTIONMainProof}

\subsection{Necessary definitions and notations}
In the paper we assume that the overall domain $A$ is finite and fixed. 
To simplify notations we even assume that $A = \{1,2,\dots,|A|\}$.
For a positive integer $m$ by $[m]$ we denote the set $\{1,2,\dots,m\}$.

For two binary relations 
$S_1$ and $S_2$ by 
$S_1+S_2$ we denote the composition of two relations, 
that is $S_1+S_2 =S$, 
where $S(x,y) = \exists z\; S_1(x,z)\wedge S_{2}(z,y)$.
By $S_{1}-S_{2}$ we denote the binary relation 
$S(x,y) = \exists z\; S_1(x,z)\wedge S_{2}(y,z)$,
that is, 
$S_{1}-S_{2} = S_{1} + S_{2}'$, where 
$S_{2}'$ is obtained from $S_{2}$ by switching the coordinates.
Similarly, we can write 
$U_1+S_2$ if $U_{1}$ is unary, 
that is $U_1+S_2 =U$, 
where $U(x) = \exists z\; U_1(z)\wedge S_{2}(z,x)$.

For a formula $\Phi$ and some free variables $u_1,\dots,u_s$ of this formula 
by 
$\Phi\downarrow \begin{smallmatrix}u_1& \dots & u_{s} \\
v_1 & \dots & v_{s} \end{smallmatrix}$
we denote the formula 
obtained from $\Phi$ 
by substituting each $u_i$ by $v_{i}$.

For a relation $R\subseteq A^{n}$ and $i\in[n]$
by $\proj_{i}(R)$ we denote the projection 
of $R$ onto the $i$-th coordinate.
Similarly for a constraint $C = R(u_1,\dots,u_s)$
by $\proj_{u_{i}}(C)$ we denote $\proj_{i}(R)$.

Sometimes it will be convenient to assume 
that some of the variables of a relation are external parameters.
Thus, a relation of arity $|A|+2$ is called
a $\zv$-parameterized binary relation, where 
$\zv\in A^{|A|}$.
Some relations have two or even three parameters. 
Thus, we may consider 
$(\zv,\delta,\alpha)$-parameterized binary relation $Q$, 
where $\zv\in A^{|A|}$, $\delta\in A^{k}$, and 
$\alpha \in A^{m}$, 
which is a relation of arity 
$|A|+k+m+2$.
To refer to the
binary relation for the fixed parameters 
$\zv,\delta$, and $\alpha$ 
we write 
$\prescript{\zv}{\delta}{Q^{\alpha}}$.
Sometimes we replace the $\alpha$-parameter
with $\forall$ or $\forall\forall$ meaning that 
we universally quantify this parameter in two 
different ways 
(see equations (\ref{EQForallInterpretation}) and (\ref{EQForallForallInterpretation})).

A parameterized unary relation is called 
\emph{nonempty} if it is nonempty for every choice of 
the parameters.
For an instance $\mathcal I$ by 
$\Var(\mathcal I)$ we denote the set of all variables appearing in this instance.

For a CSP instance $\mathcal I$ (conjunctive formula) 
and some variables $u_1,\dots,u_k\in\Var(\mathcal I)$  
by 
$\mathcal I(u_{1},\dots,u_{k})$ we denote the
set of all tuples $(a_1,\dots,a_k)$ such that 
$\mathcal I$ has a solution with 
$u_{i} = a_i$ for every $i$.
Thus, $\mathcal I(x_{1},\dots,x_{k})$
defines a $k$-ary relation.
Since this relation is defined from $\mathcal I$ by adding 
existential quantifiers, 
the relation $\mathcal I(x_{1},\dots,x_{k})$ 
is q-definable from relations in $\mathcal I$.

\subsection{Universal subset}

Suppose $S\subseteq W\subseteq A^{t}$ 
and $\Sigma$ is a set of relations on $A$.
We say that 
\emph{$S$ is a universal subset of $W$ over $\Sigma$}, denote 
$S\useq^{\Sigma} W$, if 
there exist $s$ and a relation $R\subseteq A^{t+s}$ q-definable from $\Sigma$ such that 
\begin{align*}
    S(y_1,\dots,y_t) &= \forall x_1 \dots \forall x_s \;
R(y_1,\dots,y_t,x_1,\dots,x_s),\\
W(y_1,\dots,y_t) &= \forall x \;
R(y_1,\dots,y_t,x,\dots,x).
\end{align*}
Notice that $S\useq^{\Sigma} W$
also implies that $\Sigma$ q-defines both $S$ and $W$.
To emphasize that $S$ and $W$ are different 
we write 
$S\us^{\Sigma} W$
instead of $S\useq^{\Sigma} W$.

\begin{lem}[proved in Section \ref{SUBSECTIONAuxiliaryStatements}]\label{LEMUniversalSubuniverseImplies}
Suppose 
\begin{align*}
W(y_1,\dots,y_t) &= \exists u_1 \exists u_{2} \dots\exists u_{\ell} 
\;\;\;W_{1}(z_{1,1},\dots,z_{1,n_1})\wedge \dots\wedge W_{m}(z_{m,1},\dots,z_{m,n_m}),\\
    S(y_1,\dots,y_t) &= \exists u_1 \exists u_{2} \dots\exists u_{\ell} 
\;\;\;S_{1}(z_{1,1},\dots,z_{1,n_1})\wedge \dots\wedge S_{m}(z_{m,1},\dots,z_{m,n_m}),
\end{align*} 
where each $z_{i,j}\in\{y_1,\dots,y_t,u_1,\dots,u_{\ell}\}$, 
and $S_{i}\useq^\Sigma W_{i}$ for every $i$.
Then $S\useq^\Sigma W$.
\end{lem}

For $k\ge 0$ we write $S\uuus_{k}^{\Sigma} W$ if 
$S\useq^{\Sigma} C_{1}\useq^{\Sigma} C_{2}\useq^{\Sigma}\dots\useq^{\Sigma} C_{k}\useq^{\Sigma} W$
for some relations $C_{1},\dots,C_{k}$ q-definable from 
$\Sigma$. We often omit $k$, if we do not want to specify the length of the sequence.
In Section \ref{SECTIONFindingSolution} we usually omit $\Sigma$ and write just $\useq$, $\us$, or $\uuus$ 
meaning that $\Sigma = \{R\}$.

It follows from the definition that 
for any $\alpha$-parameterized relation $Q$ 
we have 
$Q^{\forall\forall}\useq^{\{Q\}} Q^{\forall}$.


\subsection{Induced CSP Instances}\label{SUBSECTIONInducedCSPInstances}

Suppose $\Psi=\exists y_0 \forall x_1 \exists y_1 \dots
\forall x_{n} \exists y_{n} \Phi$ is 
an instance of $\QCSP(\Gamma)$, 
where $\Phi$ a quantifier-free conjunctive formula.
Let us show how to build an equivalent CSP instance of an exponential size.
If the sentence 
$\Psi$ holds, then there exist Skolem functions 
$f_{0},\dots,f_{n}$
defining a winning strategy for the existential 
player (EP), that is, she can play 
$y_{i} = f_{i}(x_1,\dots,x_{i})$.
Since it is a winning strategy, if the universal player (UP) plays 
$(x_{1},x_2,\dots,x_{n})=(a_1,a_2,\dots,a_{n})$ and 
the EP plays 
$(y_0,y_1,\dots,y_{n})=(f_{0}(),f_{1}(a_1),f_{2}(a_1,a_2),\dots,f_{n}(a_1,\dots,a_{n}))$,
then the obtained evaluation should satisfy 
the quantifier-free part $\Phi$.
We introduce 
an existential variable 
$y_{i}^{a_1,\dots,a_i}$ 
for every $i$ and $a_{1},\dots,a_{i}\in A$.
Then 
$\Psi$ is equivalent 
to the satisfiability of 
the CSP instance
$\bigwedge\limits_{a_1,\dots,a_{n}\in A}
    (\Phi\downarrow \begin{smallmatrix}
x_1 & \dots & x_{n} & y_1 &y_2 &\dots & y_{n}\\
    a_1& \dots & a_{n} & y_{1}^{a_1}&y_{2}^{a_{1},a_2}&
\dots&y_{n}^{a_1,\dots,a_{n}}\end{smallmatrix})$. 
Notice that this instance is of exponential size, that is 
why we cannot really use it in the algorithm. 
Since usually we do not have constant relations in our 
constraint language $\Gamma$ we replace the constants 
$1,2,\dots,|A|\in A$ by the respective universally quantified variables 
$z_{1},\dots,z_{|A|}$. 
Since in the paper 
we do not care about a concrete conjunctive formula $\Phi$, 
we usually define the relation $R\subseteq A^{2n+1}$ by  
$R(y_0,y_1,\dots,y_{n},x_1,\dots,x_n)=\Phi$ and work with it 
instead of constraints of $\Phi$.

In Section \ref{SECTIONFindingSolution}
the crucial idea is to show that
the induced exponential-size CSP instance
can be solved by arc-consistency, and to make 
it work we replace
the relations $R$ by stronger relations defined below.
For $R\subseteq A^{2n+1}$ and $m \in \{0,1,\dots,n-1\}$ put  
\begin{align*}
\mathcal W_R^{m}
(y_0,\dots,y_m,x_1,\dots,x_m)
=&
\forall x\exists y_{m+1}\exists y_{m+2}\dots\exists y_{n} 
R(y_0,\dots,y_n,x_1,\dots,x_m,x,x,\dots,x),\\
\mathcal S_R^{m}
(y_0,\dots,y_m,x_1,\dots,x_m)
=&
\forall x\exists y_{m+1}\forall x'
\exists y_{m+2}\dots\exists y_{n} 
R(y_0,\dots,y_n,x_1,\dots,x_m,x,x',\dots,x').
\end{align*}
Notice that $\mathcal S_R^{n-1}=\mathcal W_R^{n-1}$.
We set by definition that 
$\mathcal S_R^{n}=\mathcal W_R^{n}=R$. 
The crucial property of 
$\mathcal W_R^{m}$ and $\mathcal S_R^{m}$
is formulated in the following lemma.

\begin{lem}\label{LEMSIsUniversalInW}
Suppose $R\subseteq A^{2n+1}$ and $m\in\{0,1,\dots,n\}$, 
then 
$\mathcal {S}_R^{m}\useq^{\{R\}} \mathcal {W}_R^{m}$.
\end{lem}

\begin{proof}
We define a relation $Q$  
by
\begin{align*}
Q(y_0,\dots,y_m,&x_1,\dots,x_m,x,x^{1},\dots,x^{|A|}) = 
\\
\exists y_{m+1} \bigwedge\limits_{a\in A}
&\left(
\exists y_{m+2}\exists y_{m+3} \dots\exists y_{n}\;
R(y_0,\dots,y_{m},y_{m+1},y_{m+2},\dots,y_{n},x_1,\dots,x_m,x,x^{a},\dots,x^{a})
\right).
\end{align*}
Then the relation $Q$ witnesses that $\mathcal S_R^{m}\useq \mathcal W_R^{m}$:
\begin{align*}
\mathcal W_R^{m}
(y_0,\dots,y_m,x_1,\dots,x_m) 
&=\forall x \; Q(y_0,\dots,y_m,x_1,\dots,x_m,x,x,\dots,x), \\
\mathcal S_R^{m}
(y_0,\dots,y_m,x_1,\dots,x_m)  
&=\forall x \forall x^1\forall x^2\dots\forall x^{|A|} \; Q(y_0,\dots,y_m,x_1,\dots,x_m,x,x^{1},x^{2},\dots,x^{|A|}).\end{align*}%
\end{proof}

Thus, we can define a bunch of CSP instances equivalent to the original 
QCSP instance.  

\begin{lem}\label{LEMEquivalentCSPInstances}
Suppose $\exists y_0 \forall x_1 \exists y_1 \dots
\forall x_{n} \exists y_{n} \Phi$ is 
an instance of $\QCSP(\Gamma)$
and $R\subseteq A^{2n+1}$ is defined by 
$R(y_0,y_1,\dots,y_{n},x_1,\dots,x_n) = \Phi$. 
Then the following conditions are equivalent:
\begin{enumerate}
    \item $\exists y_0 \forall x_1 \exists y_1 \dots
\forall x_{n} \exists y_{n} \Phi$ holds;
    \item $\bigwedge\limits_{a_1,\dots,a_{n}\in A}
    (\Phi\downarrow \begin{smallmatrix}
x_1 & \dots & x_{n} & y_1 &y_2 &\dots & y_{n}\\
    a_1& \dots & a_{n} & y_{1}^{a_1}&y_{2}^{a_{1},a_2}&
\dots&y_{n}^{a_1,\dots,a_{n}}\end{smallmatrix})$ 
has a solution;
    \item $\bigwedge\limits_{a_1,\dots,a_{n}\in A}
R(y_{0},y_{1}^{a_1},y_{2}^{a_{1},a_2},\dots,y_{n}^{a_1,\dots,a_{n}},a_1,\dots,a_{n})$ 
has a solution;
    \item $\bigwedge\limits_{a_1,\dots,a_{n}\in A}
    (\Phi\downarrow \begin{smallmatrix}
x_1 & \dots & x_{n} & y_1 &y_2 &\dots & y_{n}\\z_{a_1}& \dots & z_{a_n} & y_{1}^{a_1}&y_{2}^{a_{1},a_2}&
\dots&y_{n}^{a_1,\dots,a_{n}}\end{smallmatrix})$
has a solution for every $z_1,\dots,z_{|A|}\in A$;
    \item $\bigwedge\limits_{a_1,\dots,a_{n}\in A}
R(y_{0},y_{1}^{a_1},y_{2}^{a_{1},a_2},\dots,y_{n}^{a_1,\dots,a_{n}},z_{a_1},\dots,z_{a_{n}})$
has a solution for every $z_1,\dots,z_{|A|}\in A$;
    \item $\bigwedge\limits_{\substack{ m\in\{0,1,\dots,n\}\\a_1,\dots,a_{m}\in A}}
\mathcal S_{R}^{m}(y_{0},y_{1}^{a_1},y_{2}^{a_{1},a_2},\dots,y_{m}^{a_1,\dots,a_{m}},
z_{a_1},\dots,z_{a_{m}})$
has a solution for every $z_1,\dots,z_{|A|}\in A$.
\end{enumerate}
\end{lem}
\begin{proof}
Trivially, we have 
$1\leftrightarrow 2\leftrightarrow 3$ and
$6\rightarrow 5\leftrightarrow 4\rightarrow 2$.
To complete the proof 
let us show that 
$3\rightarrow 6$.
Let us define a solution to 6 for a concrete 
$z_{1},\dots,z_{|A|}\in A$.
Let $y_{i}^{a_1,\dots,a_{i}} = b_{i}^{a_1,\dots,a_{i}}$ be a solution of 
3.
Then a solution to 6 can be defined by
$y_{i}^{a_1,\dots,a_{i}} = b_{i}^{z_{a_1},\dots,z_{a_{i}}}$. 
\end{proof}

Denote $\mathcal I_{R} =\bigwedge\limits_{\substack{ m\in\{0,1,\dots,n\}\\a_1,\dots,a_{m}\in A}}
\mathcal S_{R}^{m}(y_{0},y_{1}^{a_1},y_{2}^{a_{1},a_2},\dots,y_{m}^{a_1,\dots,a_{m}},
z_{a_1},\dots,z_{a_{m}})$, that is the equivalent CSP instance from item 6.
Notice that 
the variables $z_{1},\dots,z_{|A|}$ are viewed as
external parameters of the instance $\mathcal I_{R}$
and we call such instances \emph{$\zv$-parameterized}. 
That is why, we assume that
$z_{1},\dots,z_{|A|}$ are not in $\Var(\mathcal I_{R})$.
Also, when we refer to 
the constraint $$\mathcal S_{R}^{m}(y_{0},y_{1}^{a_1},y_{2}^{a_{1},a_2},\dots,y_{m}^{a_1,\dots,a_{m}},
z_{a_1},\dots,z_{a_{m}})$$
we usually omit these variables and write 
$\mathcal S_{R}^{m}(y_{0},y_{1}^{a_1},y_{2}^{a_{1},a_2},\dots,y_{m}^{a_1,\dots,a_{m}})$
instead.

\subsection{Mighty tuples}\label{SUBSECTIONMightyTuplesDefinitions}
In this subsection, we formulate five sufficient conditions 
for the QCSP over a constraint language $\Gamma$  
to be PSpace-hard. One of them, already defined in Section \ref{SECTIONMainResults}, is also a necessary condition.

\textbf{Mighty tuple I.}
A tuple $(Q,D,B,C,\Delta)$ is called \emph{a mighty tuple I} if 
$\Delta$ is a $\zv$-parameterized $k$-ary relation, 
$Q$ is a $(\zv,\delta,\alpha)$-parameterized binary relations, 
$D$, $B$, and $C$ are $(\zv,\delta)$-parameterized unary relations, and 
they 
satisfy the following conditions: 
\begin{enumerate}
    \item $\prescript{\zv}{}\Delta\neq\varnothing$ for every $\zv\in A^{|A|}$;
    \item $\prescript{\zv}{\delta}B$, $\prescript{\zv}{\delta}C$, and $\prescript{\zv}{\delta}D$ are nonempty 
    for every $\zv\in A^{|A|}$ and $\delta\in\prescript{\zv}{}\Delta $;
    \item $\prescript{\zv}{\delta}Q^{\alpha}$ is an equivalence relation 
    on $\prescript{\zv}{\delta}D$ for every  $\zv\in A^{|A|}$, $\delta\in \prescript{\zv}{}\Delta$ and $\alpha\in A^{m}$;

   \item  $\prescript{\zv}{\delta}Q^{\forall}=\prescript{\zv}{\delta}D\times
    \prescript{\zv}{\delta}D$
    for every $\zv\in A^{|A|}$ and $\delta\in \prescript{\zv}{}\Delta $;
    \item $\prescript{\zv}{\delta}B$ and $\prescript{\zv}{\delta}C$
    are equivalence classes of $\prescript{\zv}{\delta}Q^{\forall\forall}$;
    \item there exists $\zv\in A^{|A|}$ such that 
   $\prescript{\zv}{\delta}B\neq \prescript{\zv}{\delta}C$  for every $\delta\in\prescript{\zv}{}\Delta $. 
\end{enumerate}

In Section \ref{SECTIONHardnessResults} we prove the following theorem.

\begin{thm}\label{THMMightyTupleIPSpaceHardness}
Suppose $(Q,D,B,C,\Delta)$ is a mighty tuple I.
Then $\QCSP(\{Q,D,B,C,\Delta\})$ is PSpace-hard.
\end{thm}

In the paper instead of deriving a mighty tuple I, which is rather complicated,
we derive one of the easier tuples that we call
\emph{a mighty tuple II},
\emph{a mighty tuple III}, \emph{a mighty tuple IV}, 
and \emph{a mighty tuple V}.
They are defined below.

\textbf{Mighty tuple II.} 
A tuple $(Q,D,B,C)$ is called \emph{a mighty tuple II} if 
$Q$ is a $(\zv,\alpha)$-parameterized binary relation, 
$D$, $B$, and $C$ are $\zv$-parameterized unary relations,
and they 
satisfy the following conditions: 
\begin{enumerate}
    \item 
    $\prescript{\zv}{}B\neq \varnothing$
and        $\prescript{\zv}{}C\neq \varnothing$
     for every $\zv\in A^{|A|}$;
    \item $\prescript{\zv}{}{Q^{\alpha}}$ is an equivalence relations 
    on $\prescript{\zv}{}{D}$ for every $\zv$ and $\alpha$;
    \item $\prescript{\zv}{}{B}+\prescript{\zv}{}{Q^{\forall\forall}}=\prescript{\zv}{}{B}$ and $\prescript{\zv}{}{C}+\prescript{\zv}{}{Q^{\forall\forall}}=\prescript{\zv}{}{C}$ for every $\zv\in A^{|A|}$;
    \item 
    $\prescript{\zv}{}{B}+\prescript{\zv}{}{Q^{\forall}}=\prescript{\zv}{}{C}+\prescript{\zv}{}{Q^{\forall}}=\prescript{\zv}{}{D}$ for every $\zv\in A^{|A|}$;
     \item $\prescript{\zv}{}B\cap \prescript{\zv}{}C = \varnothing$ for some 
$\zv\in A^{|A|}$.
\end{enumerate}

\textbf{Mighty tuple III.} 
A tuple $(Q,B,C)$ is called \emph{a mighty tuple III} if 
$Q$ is a $(\zv,\alpha)$-parameterized binary relation, 
$B$ and $C$ are $\zv$-parameterized unary relations,
and they 
satisfy the following conditions: 
\begin{enumerate}
    \item 
    $\prescript{\zv}{}B\neq \varnothing$
and        $\prescript{\zv}{}C\neq \varnothing$
     for every $\zv\in A^{|A|}$;
    \item $\prescript{\zv}{}{B}+\prescript{\zv}{}{Q^{\forall\forall}}=\prescript{\zv}{}{B}$ for every $\zv\in A^{|A|}$;
    \item $\prescript{\zv}{}{Q^{\forall\forall}}+\prescript{\zv}{}{C}=\prescript{\zv}{}{C}$ for every $\zv\in A^{|A|}$;
   \item $\prescript{\zv}{}Q^{\forall}\cap (\prescript{\zv}{}B\times \prescript{\zv}{}C)\neq \varnothing$
    for every $\zv\in A^{|A|}$;
     \item $\prescript{\zv}{}B\cap \prescript{\zv}{}C = \varnothing$ for some 
$\zv\in A^{|A|}$.
\end{enumerate}

\textbf{Mighty tuple IV.} 
A tuple $(Q,D,B,C)$ is called \emph{a mighty tuple IV} if 
$Q$ is a $(\zv,\alpha)$-parameterized binary relation, 
$D$, $B$, and $C$ are $\zv$-parameterized unary relations,
and they 
satisfy the following conditions: 
\begin{enumerate}
    \item 
    $\varnothing\neq\prescript{\zv}{}B\subseteq \prescript{\zv}{}D$
and        $\varnothing\neq\prescript{\zv}{}C\subseteq \prescript{\zv}{}D$
     for every $\zv\in A^{|A|}$;
    \item $\prescript{\zv}{}{B}+\prescript{\zv}{}{Q^{\forall\forall}}=\prescript{\zv}{}{B}$ for every $\zv\in A^{|A|}$;
    \item 
    $\prescript{\zv}{}{B}+\prescript{\zv}{}{Q^{\forall}}=\prescript{\zv}{}{D}$ for every $\zv\in A^{|A|}$;
    \item 
     $\prescript{\zv}{}{D}+\prescript{\zv}{}{Q^{\forall\forall}}=\prescript{\zv}{}{D}$ for every $\zv\in A^{|A|}$;
     \item $\prescript{\zv}{}B\cap \prescript{\zv}{}C = \varnothing$ for some 
$\zv\in A^{|A|}$.
\end{enumerate}

As we prove in Section \ref{SUBSECTIONMightyTuplesTwoThreeFour} the existence of a mighty tuple II, a mighty tuple III, 
and a mighty tuple IV are equivalent:

\begin{lem}\label{LEMMightyTupleTwoThreeFourEquivalence}
Suppose $\Sigma$ is a set of relations on $A$.
Then the following conditions are equivalent:
\begin{enumerate}
    \item $\Sigma$ q-defines a mighty tuple II;
    \item $\Sigma$ q-defines a mighty tuple III;
    \item $\Sigma$ q-defines a mighty tuple IV.
\end{enumerate} 
\end{lem}

Moreover, each of them implies a mighty tuple I and therefore 
guarantees PSpace-hardness.

\begin{lem}\label{LEMMightyTupleTwoImplies}
Suppose $T$ is a mighty tuple of type II, III, or IV.
Then relations of $T$ q-define a mighty tuple I.
\end{lem}

\textbf{Mighty tuple V.} 
A tuple $(Q,D)$ is called \emph{a mighty tuple V} if 
$Q$ is a $(\zv,\alpha)$-parameterized binary relation, 
$D$ is a nonempty $\zv$-parameterized unary relation,
and they 
satisfy the following conditions: 
\begin{enumerate}
    \item  $\{(d,d)\mid d\in \prescript{\zv}{}D\}\subseteq \prescript{\zv}{}Q^{\forall}$ 
    for every $\zv\in A^{|A|}$; \;\quad($\prescript{\zv}{}Q^{\forall}$ is reflexive)
    \item $\proj_{1}(\prescript{\zv}{}Q^{\forall\forall})=\proj_{2}(\prescript{\zv}{}Q^{\forall\forall})=\prescript{\zv}{}D$ for every $\zv\in A^{|A|}$;
    \item $\prescript{\zv}{}Q^{\forall\forall}\cap \{(d,d)\mid d\in A\}=\varnothing$  
    for some $\zv\in A^{|A|}$. \quad ($\prescript{\zv}{}Q^{\forall\forall}$ is irreflexive)
\end{enumerate}

In Section \ref{SUBSECTIONMightyTupleV} 
we prove that a mighty tuple V implies a mighty tuple I.

\begin{lem}\label{LEMMightyTupleFiveImplies}
Suppose $(Q,D)$ is a mighty tuple V.
Then there exists a mighty tuple I q-definable from $\{Q,D\}$.
\end{lem}


\subsection{Reductions}

\emph{A $\zv$-parameterized reduction} $D^{(\top)}$ for 
a $\zv$-parameterized CSP instance 
$\mathcal I$ 
is a mapping that assigns a 
$\zv$-parameterized unary relation  
$D_{u}^{(\top)}$ to every variable $u$ of $\mathcal I$.
Then for any constraint $C=R(u_1,\dots,u_s)$
by $C^{(\top)}$ we denote 
the constraint $R^{(\top)}(u_1,\dots,u_s)$, 
where 
$\prescript{\zv}{}R^{(\top)}(u_1,\dots,u_s) = 
\prescript{\zv}{}R(u_1,\dots,u_s)\wedge \bigwedge_{i=1}^{s} 
(u_{i}\in \prescript{\zv}{}D_{u_i}^{(\top)})$.
A reduction $D^{(\top)}$ of $\mathcal I$ is called \emph{1-consistent} 
if  for
any constraint $C=R(u_1,\dots,u_s)$ in $\mathcal I$ 
and any $i\in \{1,2,\dots,s\}$ we have 
$\proj_{u_i}(C^{(\top)}) = D_{u_{i}}^{(\top)}$.

We say that a $\zv$-parameterized reduction $D^{(\bot)}$ 
is \emph{smaller than} a $\zv$-parameterized reduction $D^{(\top)}$ if 
$\prescript{\zv}{}D_{u}^{(\bot)}\subseteq \prescript{\zv}{}D_{u}^{(\top)}$
for every $u$ and $\zv$, and $D^{(\bot)}\neq D^{(\top)}$.
In this case we write $D^{(\bot)}\subsetneq D^{(\top)}$.
A reduction $D^{(\top)}$ is called \emph{nonempty} 
if $\prescript{\zv}{}D_{u}^{(\top)}$ is nonempty 
for every $u$ and $\zv$.

For a $\zv$-parameterized reduction 
$D^{(\top)}$ of $\mathcal I_{R}$ 
by 
$D^{(\top,0)}_{y_{m}^{a_1,\dots,a_{m}}}$ 
we denote the 
$\zv$-parameterized unary relation defined by
$$\left(\mathcal W_{R}^{m}(y_{0},y_{1}^{a_1},y_{2}^{a_{1},a_2},\dots,y_{m}^{a_1,\dots,a_{m}}) \wedge \bigwedge\limits_{i=0}^{m-1} 
(y_{i}^{a_{1},\dots,a_{i}}\in D_{y_{i}^{a_{1},\dots,a_{i}}}^{(\top)})\right) (y_{m}^{a_1,\dots,a_{m}}),$$
which is the set of all possible values
of $y_{m}^{a_1,\dots,a_{m}}$ in solutions of the 
conjunction in the brackets.
In other words, 
$D^{(\top,0)}_{u}$ is the restriction on $u$ we get  
by restricting 
the variables 
$u_{0},u_{1},\dots,u_{m-1}$  
in 
$\mathcal W_{R}^{m}(u_0,u_1,\dots,u_{m-1},u)$
to $D^{(\top)}$. 
Notice that 
if $D^{(\top)}$ is 1-consistent
then $D^{(\top)}_{u}\subseteq D^{(\top,0)}_{u}$ for every $u$.
A reduction $D^{(\top)}$ of $\mathcal I_{R}$ is called \emph{a universal reduction} if
$D^{(\top)}_{u}\uuus D^{(\top,0)}_{u}$ for every $u\in\Var(\mathcal I_{R})$.

\subsection{Proof of Theorems \ref{THMMainUPRestriction}, \ref{THMNonIdempotantClassification}, and \ref{THMIdempotantClassification}}

As we said before
for a given instance of the QCSP we define 
an equivalent $\zv$-parameterized exponential CSP
$\mathcal I_{R}$
whose only relations are 
$\mathcal S_{R}^{m}$.
The next theorem shows that 
either we can find a 1-consistent reduction of $\mathcal I_{R}$, or 
we can find a small subinstance without 
a solution for some $\zv$,  
or the QCSP over this language is PSpace-hard.

\begin{thm}\label{THMFindSmallTree}
Suppose $R\subseteq A^{2n+1}$.
Then one of the following conditions holds:
\begin{enumerate}
    \item there exists a $\zv$-parameterized
nonempty 1-consistent reduction for 
$\mathcal I_{R}$;
    \item  
there exists a subinstance  
$\mathcal J\subseteq \mathcal I_{R}$ with at most 
$(n\cdot |A|)^{2^{{2|A|}^{|A|+1}}}$ variables not having a solution for some $\zv\in A^{|A|}$;
\item there exists
    a mighty tuple III q-definable from $R$.
\end{enumerate}
\end{thm}

The next theorem as well as Theorem \ref{THMFindSmallTree}
is proved in Section \ref{SUBSECTIONFindingConsistentReduction}.

\begin{thm}\label{THMNonemptyReductionIsZigzag}
Suppose $R\subseteq A^{2n+1}$, 
$D^{(\top)}$ is an inclusion-maximal $\zv$-parameterized 
1-consistent nonempty reduction 
for $\mathcal I_{R}$.
Then 
$D^{(\top)}$ is a universal reduction.
\end{thm}

Then we consider the case when there exists a 1-consistent
universal reduction for $\mathcal I_{R}$. We will show that 
if the instance has no solutions, then we can find a smaller 1-consistent reduction.
We do this in two steps. First, in Section \ref{SUBSECTIONUniversalSubuniverseExists} 
we prove the following theorem that states that we can find 
a universal subset on some $D_{u}^{(\top)}$.

\begin{thm}\label{THMFindUniversalSubuniverse}
Suppose $R\subseteq A^{2n+1}$, $\mathcal I_{R}$ has no solutions for some $\zv$, 
$D^{(\top)}$ is a $\zv$-parameterized
universal 1-consistent reduction for 
$\mathcal I_{R}$.
Then one of the following conditions holds:
\begin{enumerate}
    \item 
    there exists a variable $u$ of $\mathcal I_{R}$ and
    a $\zv$-parameterized unary relation $B$ such that 
    $B\us D_{u}^{(\top)}$;
   \item  there exists a mighty tuple V q-definable from $R$.
\end{enumerate}
\end{thm}

Then we show (in Section \ref{SUBSECTIONFindingSmallerReduction}) that this universal subset can be extended to a smaller 
1-consistent reduction. 

\begin{thm}\label{THMFindSmallerReduction}
Suppose $R\subseteq A^{2n+1}$, 
$D^{(\top)}$ is a $\zv$-parameterized
universal 1-consistent reduction for 
$\mathcal I_{R}$,
$u\in \Var(\mathcal I_{R})$,
$B\us D_{u}^{(\top)}$ is 
    a $\zv$-parameterized nonempty unary relation. Then one of the following conditions holds:
\begin{enumerate}
    \item there exists a $\zv$-parameterized
universal 1-consistent reduction $D^{(\bot)}$ for 
$\mathcal I_{R}$ that is smaller than $D^{(\top)}$;
    \item  there exists a mighty tuple IV q-definable from $R$.
\end{enumerate}
\end{thm}
In both cases, there is an option that 
it cannot be done, but 
this implies that the QCSP is PSpace-hard.
Combining above theorems we obtain the fundamental fact that 
it is sufficient to run arc-consistency algorithm on $\mathcal I_{R}$ to be sure that it has a solution.

\begin{thm}\label{THMOneConsistentReductionImpliesASolution}
Suppose $R\subseteq A^{2n+1}$, 
there exists a $\zv$-parameterized
nonempty 
1-consistent reduction for 
$\mathcal I_{R}$. Then one of the following conditions holds:
\begin{enumerate}
    \item $\mathcal I_{R}$ has a solution for every $\zv\in A^{|A|}$;
    \item  there exists a mighty tuple  IV or V q-definable from $R$.
\end{enumerate}
\end{thm}

\begin{proof}
Assume the converse. Suppose 
there does not exist a mighty tuple IV or V q-definable from 
$R$ and $\mathcal I_{R}$ has no solutions for some $\zv$.
By Theorem \ref{THMNonemptyReductionIsZigzag} 
there exists a $\zv$-parameterized
universal 1-consistent reduction $D^{(1)}$ for 
$\mathcal I_{R}$. 
By Theorem \ref{THMFindUniversalSubuniverse} 
there exists a variable $u$ of $\mathcal I_{R}$ and
    a $\zv$-parameterized unary relation $B$ such that 
    $B\us D_{u}^{(1)}$.
By Theorem \ref{THMFindSmallerReduction} 
there exists a smaller $\zv$-parameterized 1-consistent universal 
reduction $D^{(2)}$. 
Then iteratively applying Theorems \ref{THMFindUniversalSubuniverse} and \ref{THMFindSmallerReduction} we 
build reductions 
$D^{(1)}, D^{(2)},\dots,D^{(s)}$.
Since we never stop and we cannot reduce domains forever, we 
get a contradiction.
\end{proof}

Summarizing above theorems we obtain the following theorem.

\begin{thm}\label{THMSmallWinningSetOrMightyTuple}
Suppose $\Gamma$ is a constraint language on a finite set $A$.
Then one of the following conditions holds:
\begin{enumerate}
\item for any No-instance $\exists y_0 \forall x_1\exists y_1\dots
\forall x_n\exists y_n\;\Psi$ of $\QCSP(\Gamma)$
there exists 
    $S\subseteq A^{n}$ with $|S|\le|A|^{2}\cdot (n\cdot |A|)^{2^{{2|A|}^{|A|+1}}}$ such that  
    $\exists y_0 \forall x_1\exists y_1\dots
\forall x_n\exists y_n(
(x_1,\dots,x_n)\in S\rightarrow \Psi)$ does not hold;
\item there exists a mighty tuple III, IV, or V q-definable over $\Gamma$.
\end{enumerate} 
\end{thm}

\begin{proof}
Below we assume that the second condition 
does not hold, that is,
there does not exist a mighty tuple III, IV, or V q-definable over 
$\Gamma$.

Let us prove condition 1.
Put $R(y_0,y_1,\dots,y_n,x_1,\dots,x_n) =\Psi$.
By Lemma \ref{LEMEquivalentCSPInstances},
$\mathcal I_{R}$ has no solutions for some $\zv\in A^{|A|}$.
Then by Theorem \ref{THMOneConsistentReductionImpliesASolution}
there does not exist 
a $\zv$-parameterized
nonempty 1-consistent reduction for 
$\mathcal I_{R}$. 
Applying Theorem \ref{THMFindSmallTree} to $\mathcal I_{R}$ we 
obtain that only case 2 is possible.




Thus, there exists  
a subinstance $\mathcal J\subseteq \mathcal I_{R}$ with at most 
$(n\cdot |A|)^{2^{{2|A|}^{|A|+1}}}$ variables not having a solution for some
$\zv=(b_1,\dots,b_{|A|})$. 
Let us define an appropriate $S\subseteq A^{n}$.
For a tuple $(a_1,\dots,a_i)$ and $i\in\{0,1,\dots,n-2\}$
by $E(a_1,\dots,a_i)$ we denote the set 
$$\{(b_{a_1},b_{a_2},\dots,b_{a_i},\underbrace{c,d,d,\dots,d}_{n-i})\mid 
c,d\in A\}.$$
Put
$E(a_1,\dots,a_n) = \{(b_{a_1},\dots,b_{a_n})\}$
and 
$E(a_1,\dots,a_{n-1}) = \{(b_{a_1},\dots,b_{a_{n-1}},c)\mid c\in A\}$.
Then 
put 
$S = \bigcup\limits_{y_{i}^{a_1,\dots,a_i} \in \Var(\mathcal J)} E(a_1,\dots,a_i)$.
We have 
$$|S|\le |A|^{2}\cdot|\Var(\mathcal J)|\le 
|A|^{2}\cdot (n\cdot |A|)^{2^{{2|A|}^{|A|+1}}}.$$
Observe that tuples 
from $S$ cover all the constraints of $\mathcal J$ for 
$\zv = (b_1,\dots,b_{|A|})$.
Therefore 
the CSP instance
$\bigwedge\limits_{(a_1,\dots,a_{n})\in S}
R(y_{0},y_{1}^{a_1},y_{2}^{a_{1},a_2},\dots,y_{n}^{a_1,\dots,a_{n}},a_1,\dots,a_{n})$ cannot be satisfied.
Notice that  
the existence of a winning strategy for the EP in
$$\exists y_0 \forall x_1\exists y_1\dots
\forall x_n\exists y_n(
(x_1,\dots,x_n)\in S\rightarrow \Psi)$$
is equivalent to the satisfiability 
of $\bigwedge\limits_{(a_1,\dots,a_{n})\in S}
R(y_{0},y_{1}^{a_1},y_{2}^{a_{1},a_2},\dots,y_{n}^{a_1,\dots,a_{n}},a_1,\dots,a_{n})$.
Hence, the sentence 
$\exists y_0 \forall x_1\exists y_1\dots
\forall x_n\exists y_n(
(x_1,\dots,x_n)\in S\rightarrow \Psi)$ does not hold, which completes the proof.
\end{proof}

Combining the above theorem with 
Theorem \ref{THMMightyTupleIPSpaceHardness} 
and Lemmas \ref{LEMMightyTupleTwoImplies}  and 
\ref{LEMMightyTupleFiveImplies},
we derive Theorem \ref{THMMainUPRestriction}.
Assuming that 
$\mathrm{PSpace}\neq\Pi_{2}^{P}$, we obtain the following classification of PSpace-complete languages that is 
a bit stronger than the classification in Theorem \ref{THMNonIdempotantClassification}:

\begin{thm}\label{THMPSpaceMightyClassification}
Suppose $\Gamma$ is a constraint language on a finite set $A$.
Then the following conditions are equivalent:
\begin{enumerate}
\item $\QCSP(\Gamma)$ is PSpace-complete;
\item $\Gamma$ q-defines a mighty tuple I;
\item $\Gamma$ q-defines a mighty tuple II or V.
\end{enumerate} 
\end{thm}

\begin{proof}
2 implies 1 by Theorem \ref{THMMightyTupleIPSpaceHardness}.
3 implies 2 by Lemmas \ref{LEMMightyTupleTwoImplies}  and 
\ref{LEMMightyTupleFiveImplies}.
Let us prove that 1 implies 3.
As we assumed that $\mathrm{PSpace}\neq\Pi_{2}^{P}$,
and the first case of 
Theorem \ref{THMSmallWinningSetOrMightyTuple}
implies $\Pi_{2}^{P}$-membership, 
we derive the second case of Theorem \ref{THMSmallWinningSetOrMightyTuple}, 
that is, there exists a mighty tuple III, IV, or V q-definable from $\Gamma$. 
Using lemma \ref{LEMMightyTupleTwoThreeFourEquivalence} 
there exists a mighty tuple II or V q-definable from $\Gamma$, which compltes the proof.
\end{proof}

For constraint languages containing all constant relations an easier classification of 
PSpace-complete languages is proved in Section \ref{SUBSECTIONClassificationIdempotentProof}. 

\begin{lem}\label{LEMTHMIdempotantClassification}
Suppose $\Gamma\supseteq\{x=a\mid a\in A\}$ is a set of relations on $A$.
Then the following conditions are equivalent:
\begin{enumerate}
    \item $\Gamma$ q-defines a mighty tuple I;
    \item $\Gamma$ q-defines a mighty tuple II;
\item  there exist an equivalence relation $\sigma$ on $D\subseteq A$ and $B,C\subsetneq A$
such that 
$B\cup C = A$ and 
    $\Gamma$
    q-defines the relations 
    $(y_{1},y_{2}\in D)\wedge(\sigma(y_1,y_2)\vee (x\in B))$ and 
    $(y_{1},y_{2}\in D)\wedge(\sigma(y_1,y_2)\vee (x\in C))$.
\end{enumerate}
\end{lem}

The above lemma together with 
Theorem \ref{THMMightyTupleIPSpaceHardness} 
implies Theorem \ref{THMIdempotantClassification}.


As we prove in Lemmas \ref{LEMMightyTupleTwoThreeFourEquivalence}
and \ref{LEMMightyTupleTwoImplies},
mighty tuples II, III, and IV are equivalent, 
and each of them implies a mighty tuple I.
Nevertheless, it is still not clear whether 
a mighty tuple I implies a mighty tuple II.
Moreover, we do not know how to q-define  
a mighty tuple II or V
from a mighty tuple I even though we know it is possible 
by Theorem \ref{THMPSpaceMightyClassification}.
Considering constraint languages with all constant relations 
from Lemma \ref{LEMTHMIdempotantClassification},
we can observe that a mighty tuple V cannot be derived from 
the relations 
   $\sigma(y_1,y_2)\vee (x\in B)$, 
    $\sigma(y_1,y_2)\vee (x\in C)$ and constant relations.
    Hence, it is not always true that 
a mighty tuple I implies a mighty tuple V, 
 but we can formulate the following conjecture.
    
\begin{conj}
$\QCSP(\Gamma)$ is PSpace-complete if and only if 
$\Gamma$ q-defines a mighty tuple II.
\end{conj}

\section{Finding a solution of  $\mathcal I_{R}$}\label{SECTIONFindingSolution}

\subsection{Definitions}

 We say that $u_{0}-C_{1}-u_{1}-\dots - C_{\ell}-u_{\ell}$ is
\emph{a path} in a CSP instance $\mathcal I$ if $C_{i}$ is a constraint of $\mathcal I$ and $u_{i},u_{i+1}\in\Var(C_{i})$
for every $i$.
The number $\ell$ is called \emph{the length} of the path.
We say that an instance is \emph{connected} if 
any two variables are connected by a path.
We say that an instance is \emph{a tree-instance} if it is connected and there is no path
$u_{0}-C_{1}-u_{1}-\dots -u_{\ell-1}-C_{\ell}-u_{\ell}$
such that $\ell\ge 2$, $u_{0} = u_{l}$, and all the constraints $C_{1},\ldots,C_{\ell}$ are different.

An instance $\mathcal I'$ is called \emph{a universal weakening} of 
$\mathcal I$ if $\mathcal I$ can be obtained from $\mathcal I'$
by replacing some constraint relations by 
their universal subsets over $\{R\}$.
We also denote it by $\mathcal I\useq \mathcal I'$.

A ($\zv$-parameterized) CSP instance $\mathcal I'$ is called \emph{a covering} of a 
($\zv$-parameterized) CSP instance $\mathcal I$ 
if there exists a mapping $\phi:\Var(\mathcal I')\to\Var(\mathcal I)$
such that for every constraint $S(u_1,\dots,u_t)$
of $\mathcal I'$ 
the constraint $S(\phi(u_1),\dots,\phi(u_t))$ is in 
$\mathcal I$.
We say that $\phi(u)$ is \emph{the parent of $u$}
and $u$ is \emph{a child} of $\phi(u)$.
The same child/parent terminology will also be applied to constraints.
An instance is called 
\emph{a tree-covering} if it is a covering and also a tree-instance.
Notice that reductions for an instance 
can be naturally extended to their 
coverings.

For a ($\zv$-parameterized) instance $\mathcal I$ and a ($\zv$-parameterized) reduction
$D^{(\top)}$
by $\mathcal I^{(\top)}$ we denote the instance 
whose constraints are reduced to $D^{(\top)}$.

All the variables of $\mathcal I_{R}$ can be drawn as a tree with 
a root $y_{0}$ and 
leaves $y_{n}^{a_1,\dots,a_n}$.
We assume that the root is at the top
and the leaves are at the bottom. Thus, whenever
we refer to \emph{a lowest}/\emph{highest} variable we mean
$y_{i}^{a_1,\dots,a_i}$ with the maximal/minimal $i$.
Also we say that a variable 
$y_{i}^{a_1,\dots,a_i}$ is from \emph{the $i$-th level}.

For a ($\zv$-parameterized) instance $\mathcal I$ and some variables 
$u_1,\dots,u_m\in\Var(\mathcal I)$ by 
$\mathcal I(u_1,\dots,u_{m})$ we denote the 
set of all tuples $(a_1,\dots,a_m)$ such that 
$\mathcal I$ has a solution with 
$u_{i} = a_i$ for every $i$.
Thus, $\mathcal I(u_{1},\dots,u_{m})$
defines an $m$-ary ($\zv$-parameterized) relation.
Notice that 
$\mathcal I(u_{1},\dots,u_{m})$ is q-definable over the relations 
appearing in $\mathcal I$.



Suppose $D^{(\top)}$ is a $\zv$-parameterized universal reduction 
of $\mathcal I_{R}$, that is, 
for every variable $u$
we have 
$D_{u}^{(\top)}\uuus D_{u}^{(\top,0)}$.
As we can repeat elements in the sequence $\uuus$ and 
any sequence longer than $|A|$ has repetitions,
we have $B\uuus_{|A|} C\Leftrightarrow B\uuus C$ for any $B,C\subseteq A$.
The sequence witnessing that 
$D_{u}^{(\top)} \uuus_{|A|} D_{u}^{(\top,0)}$
we denote by 
$$D_{u}^{(\top,|A|)}, D_{u}^{(\top,|A|-1)},\dots, D_{u}^{(\top,1)},D_{u}^{(\top,0)},$$
where $D_{u}^{(\top,|A|)}= D_{u}^{(\top)}$.
Also, notice that the reduction
$D^{(\top,i)}$ is defined independently on different variables, that is why we should not expect 
it to be 1-consistent. 

To simplify the explanation 
we give names to some constraints 
we will need later.
\begin{align*}
C_{S,\top}^{a_1,\dots,a_m} := \;\;\;\;\;&\left(
\mathcal S_{R}^{m}(y_0,y_1^{a_1},\dots,y_{m}^{a_1,\dots,a_{m}},
z_{a_1},\dots,z_{a_m}) 
\wedge \bigwedge\limits_{i=0}^{m} y_{i} \in D_{y_{i}^{a_1,\dots,a_i}}^{(\top)} \right)\\
C_{W,\top}^{a_1,\dots,a_m} := \;\;\;\;\;&\left(
\mathcal W_{R}^{m}(y_0,y_1^{a_1},\dots,y_{m}^{a_1,\dots,a_{m}},
z_{a_1},\dots,z_{a_m}) 
\wedge \bigwedge\limits_{i=0}^{m} y_{i} \in D_{y_{i}^{a_1,\dots,a_i}}^{(\top)} \right)\\
C_{W,\top,j}^{a_1,\dots,a_m} := \;\;\;\;\;&\left(
\exists y_{m}^{a_1,\dots,a_{m}} \; \mathcal W_{R}^{m}(y_0,y_1^{a_1},\dots,y_{m}^{a_1,\dots,a_{m}},
z_{a_1},\dots,z_{a_m}) 
\wedge \right.\\
\quad\quad\quad\quad\quad\quad&\quad\quad\quad\quad\quad\quad
\quad\quad\quad\quad\quad\quad\quad
\left.\bigwedge\limits_{i=0}^{m-1} y_{i} \in D_{y_{i}^{a_1,\dots,a_i}}^{(\top)} \wedge 
y_{m}^{a_1,\dots,a_{m}}\in D_{y_{m}^{a_1,\dots,a_m}}^{(\top,j)}\right)
\end{align*}

Notice that 
$C_{S,\top}^{a_1,\dots,a_m}$ and 
$C_{W,\top}^{a_1,\dots,a_m}$ have $m+1$ $y$-variables, but 
$C_{W,\top,j}^{a_1,\dots,a_m}$ has only $m$ $y$-variables.
Then $\mathcal I_{R}^{(\top)}$ is the 
instance 
consisting of the constraints 
$C_{S,\top}^{a_1,\dots,a_m}$, where 
$m\in\{0,1,\dots,n\}$ and 
$a_1,\dots,a_m\in A$.


By Lemma \ref{LEMUniversalSubuniverseImplies},
$C_{S,\top}^{a_1,\dots,a_m}\useq 
C_{W,\top}^{a_1,\dots,a_m}$
and 
$C_{W,\top,j+1}^{a_1,\dots,a_m}\useq
C_{W,\top,j}^{a_1,\dots,a_m}$.

\subsection{Auxiliary statements}\label{SUBSECTIONAuxiliaryStatements}

\begin{LEMUniversalSubuniverseImpliesLEM}
Suppose 
\begin{align*}
W(y_1,\dots,y_t) &= \exists u_1 \exists u_{2} \dots\exists u_{\ell} 
\;\;\;W_{1}(z_{1,1},\dots,z_{1,n_1})\wedge \dots\wedge W_{m}(z_{m,1},\dots,z_{m,n_m}),\\
    S(y_1,\dots,y_t) &= \exists u_1 \exists u_{2} \dots\exists u_{\ell} 
\;\;\;S_{1}(z_{1,1},\dots,z_{1,n_1})\wedge \dots\wedge S_{m}(z_{m,1},\dots,z_{m,n_m}),
\end{align*} 
where each $z_{i,j}\in\{y_1,\dots,y_t,u_1,\dots,u_{\ell}\}$, 
and $S_{i}\useq^\Sigma W_{i}$ for every $i$.
Then $S\useq^\Sigma W$.
\end{LEMUniversalSubuniverseImpliesLEM}
\begin{proof}
Let 
$S_{i}\useq^{\Sigma} W_{i}$ be  witnessed by 
a relation $R_{i}\subseteq A^{n_i+k_i}$.
Let $k = |A|$.
Define the relation $R$ witnessing that $S\useq^{\Sigma} W$ by 
\begin{align*}
R(y_1,\dots,y_t,x_{1},\dots,&x_{k}) =\\ 
\exists u_1 \exists u_{2} \dots\exists u_{\ell} 
&\left(\bigwedge\limits_{i=1}^{m}\;
\bigwedge\limits_{\phi\colon[k_{i}]\to[k]}
R_{i}(z_{i,1},\dots,z_{i,n_i},x_{\phi(1)},\dots,x_{\phi(k_{i})})
\wedge
\bigwedge\limits_{i=1}^{m}
W_{i}(z_{i,1},\dots,z_{i,n_i})\right)
\end{align*}
Notice that $W_{i}$ is q-definable from $R_{i}$, hence 
$R$ is q-definable over $\Sigma$.
\end{proof}

\begin{lem}[\cite{zhuk2021strong}, Lemma 5.6]\label{LEMExistenceOfTreeCoverings}
Suppose $D^{(\top)}$ is a reduction for an instance $\mathcal I$, 
$D^{(\bot)}$ is an inclusion maximal 1-consistent reduction 
for $\mathcal I$
such that $D_{u}^{(\bot)}\subseteq D_{u}^{(\top)}$ for every $u$. Then for every variable 
$y\in\Var(\mathcal I)$ there exists a 
tree-covering $\Upsilon_{y}$ of $\mathcal I$ 
such that 
$\Upsilon_{y}^{(\top)}(y)$ defines $D_{y}^{(\bot)}$.
\end{lem}

The above lemma can be generalized for 
$\zv$-parameterized reductions as follows:

\begin{lem}\label{LEMMaximalReductionGeneral}
Suppose $D^{(\top)}$ is a $\zv$-parameterized reduction for a $\zv$-parameterized instance $\mathcal I$, 
$D^{(\bot)}$ is an inclusion maximal $\zv$-parameterized 1-consistent reduction 
for $\mathcal I$
such that $D^{(\bot)}_{u}\subseteq D^{(\top)}_{u}$ for every $u$. Then for every variable 
$y\in\Var(\mathcal I)$ there exists a 
tree-covering $\Upsilon_{y}$ of $\mathcal I$ 
such that 
$\prescript{\zv}{}\Upsilon_{y}^{(\top)}(y)$ defines $\prescript{\zv}{}D_{y}^{(\bot)}$
for every $\zv$.
\end{lem}

\begin{proof}
By Lemma \ref{LEMExistenceOfTreeCoverings}
for every $\zv\in A^{|A|}$ and $v$ there exists a tree-covering $\Upsilon_{v,\zv}$ such that 
$\prescript{\zv}{}{\Upsilon_{v,\zv}^{(\top)}}(v)$ defines 
$\prescript{\zv}{}{D_{v}^{(\bot)}}$.
Let 
$\Upsilon_{v}$ be $\bigwedge_{\zv\in A^{|A|}}\Upsilon_{v,\zv}$, 
where we assume that the only common variable of 
$\Upsilon_{v,\zv_1}$ and $\Upsilon_{v,\zv_2}$,
if $\zv_1\neq \zv_2$, is $v$.
Then $\Upsilon_{v}$ is a tree-covering of $\mathcal I$ 
and $\prescript{\zv}{}{\Upsilon_{v}^{(\top)}}(v)$ defines 
$\prescript{\zv}{}{D_{v}^{(\bot)}}$ for every $\zv$.
\end{proof}

We can always take a trivial reduction $\prescript{\zv}{}D_{u}^{(\top)} = A$ for every 
$\zv$ and $u$, and derive the following lemma.

\begin{lem}\label{LEMMaximalReductionSimple}
Suppose 
$D^{(\bot)}$ is an inclusion maximal $\zv$-parameterized 1-consistent reduction 
for a $\zv$-parameterized instance $\mathcal I$. Then for every variable 
$u\in\Var(\mathcal I)$ there exists a 
tree-covering $\Upsilon_{u}$ of $\mathcal I$ 
such that 
$\prescript{\zv}{}\Upsilon_{u}(u)$ defines $\prescript{\zv}{}D_{u}^{(\bot)}$
for every $\zv$.
\end{lem}

\subsection{Finding a 1-consistent reduction}\label{SUBSECTIONFindingConsistentReduction}
In this section we prove that either there exists 
a 1-consistent reduction for $\mathcal I_{R}$, 
or there exists a polynomial size subinstance of $\mathcal I_{R}$ without a solution, or
we can build a mighty tuple III that guarantees PSpace-hardness. 
To be able to simplify our instance to 
a polynomial size we will need even stronger relations.
Put
\begin{align*}
    \mathcal {\widetilde W}_R^{m}
(y_0,\dots,y_m,x_1,\dots,x_m)
&=\bigwedge\limits_{i=0}^{m}
\mathcal {W}_R^{i}
(y_0,\dots,y_i,x_1,\dots,x_i),
\\
\mathcal {\widetilde S}_R^{m}
(y_0,\dots,y_m,x_1,\dots,x_m)
&=
\mathcal {S}_R^{m}
(y_0,\dots,y_m,x_1,\dots,x_m)
\wedge \bigwedge\limits_{i=0}^{m-1}
\mathcal {W}_R^{i}
(y_0,\dots,y_i,x_1,\dots,x_i).
\end{align*}

The following lemma follows immediately from the definition and 
Lemmas \ref{LEMSIsUniversalInW} and \ref{LEMUniversalSubuniverseImplies}.

\begin{lem}\label{LEMSTildeIsUniversalInWTilde}
Suppose $R\subseteq A^{2n+1}$, 
then  
$\mathcal {\widetilde{S}}_R^{m}\useq \mathcal {{\widetilde W}}_R^{m}$.
\end{lem}

Denote 
\begin{align*}
    \mathcal {\widetilde I}_{R} =\bigwedge\limits_{\substack{ m\in\{0,1,\dots,n\}\\a_1,\dots,a_{m}\in A}}
&\mathcal {\widetilde S}_{R}^{m}(y_{0},y_{1}^{a_1},y_{2}^{a_{1},a_2},\dots,y_{m}^{a_1,\dots,a_{m}},
z_{a_1},\dots,z_{a_{m}})\wedge \\
&\bigwedge\limits_{\substack{ m\in\{0,1,\dots,n\}\\a_1,\dots,a_{m}\in A}}
\mathcal {\widetilde W}_{R}^{m}(y_{0},y_{1}^{a_1},y_{2}^{a_{1},a_2},\dots,y_{m}^{a_1,\dots,a_{m}},
z_{a_1},\dots,z_{a_{m}}).
\end{align*}
Notice that $\mathcal {\widetilde I}_{R}$ is obtained 
from $\mathcal I_{R}$
by adding constraints that 
are satisfied by any solution of $\mathcal I_{R}$.
Hence $\mathcal {\widetilde I}_{R}$ is satisfiable if and only 
if $\mathcal I_{R}$ is satisfiable.

First, we prove some technical lemmas 
showing the connection of the length of a path and 
the size of a tree-covering.

\begin{lem}\label{LEMIndexOfAVariable}
Suppose 
$\mathcal T$ is a $\zv$-parameterized tree-instance such that 
$\prescript{\zv}{}{\mathcal T}$ has no solutions for 
some $\zv$, but
if we remove any constraint from $\mathcal T$ we get an instance with a solution for every $\zv$. 
Then each variable appears in $\mathcal T$ at most $|A|$ times.
\end{lem}

\begin{proof}
Let $u$ be some variable of $\mathcal T$.
Since $\mathcal T$ is a tree-instance, we can split it into tree-subinstances in $u$, that is, for any constraint $C$ containing $u$ 
we take the (maximal) tree-subinstance containing $C$ but not the other constraints containing $u$.
Let $\Phi_1,\dots,\Phi_{s}$ be 
the subinstances we obtain if we split the instance $\mathcal T$ in $u$. 
By 
$B_{i}$ we denote the $\zv$-parameterized unary relation defined by 
$\Phi_{i}(u)$.
Then there exists $\zv$ such that 
$\bigcap\limits_{i\in[s]} \prescript{\zv}{}B_{i}=\varnothing$.
Since removing any constraint from 
$\mathcal T$ gives an instance with a solution, 
$\bigcap\limits_{i\in[s]\setminus\{j\}} \prescript{\zv}{}B_{i}\neq\varnothing$
for every $j\in[s]$.
Therefore $s\le |A|$, which completes the proof.
\end{proof}

\begin{lem}\label{LEMMinPathInATree}
Suppose $\mathcal T$ is a tree-instance having $N\ge 2$ variables, the arity of every constraint of $\mathcal T$ is at most $n$, 
and every variable appears  at most $|A|$ times.
Then there exists a path in $\mathcal T$ of length at least 
$\lceil{\log_{k} (N\cdot (k-1)+1)}\rceil$, 
where $k = (n-1)\cdot |A|$.
\end{lem}

\begin{proof}
We prove even a stronger claim by induction on $N$. 
We prove that there exists a path of length
$\lceil{\log_{k} (N\cdot (k-1)+1)}\rceil$ starting with any variable $u$.

Suppose 
$u$ appears in constraints 
$C_{1},C_{2},\dots,C_{s}$, where $s\le|A|$.
Let $V$ be the set of all variables 
appearing in $C_{1},\dots,C_{s}$ except for $u$.
Notice that every variable 
$v\in V$ appears in exactly one constraint $C_{i}$.
By $\Phi_{v}$ we denote the (maximal) tree-subinstance of $\mathcal T$ 
containing all the constraints with $v$ but 
constraints from $\{C_{1},\dots,C_{s}\}$.
Then $\mathcal T = 
\bigwedge_{v\in V} \Phi_{v}\wedge \bigwedge_{i\in[s]} C_{i}$, and 
$\Phi_{v_1}$ and $\Phi_{v_2}$ do not share any variables
if $v_{1}\neq v_{2}$.
Since $|V|\le s\cdot (n-1)\le k$, 
there exists $v\in V$ such that
$\Phi_{v}$ contains at least 
$(N-1)/k$ variables.
By the inductive assumption $\Phi_{v}$ 
contains a path starting with $v$ 
of length at least
$\lceil\log_{k} (((N-1)/k)\cdot (k-1)+1)\rceil$.
Then $\mathcal T$ has a path of length at least
$$
1+ \lceil\log_{k} (((N-1)/k)\cdot (k-1)+1)\rceil
=\lceil\log_{k} ((N-1)\cdot (k-1)+k)\rceil
=\lceil\log_{k} (N\cdot (k-1)+1)\rceil.
$$
\end{proof}

Suppose $\mathcal T$ is a tree-covering of $\mathcal{\widetilde I}_{R}$. 
We define several transformations of $\mathcal T$, which we will apply to make it easier.

\begin{enumerate}
    \item[(w)]
    replace a constraint ${\mathcal {\widetilde S}}_{R}^{i}(u_0,u_1,\dots,u_{i})$ by 
    ${\mathcal {\widetilde W}}_{R}^{i}(u_0,u_1,\dots,u_{i})$; 
    \item[(s)] if $u_{i}$ appears only once in $\mathcal T$ 
    in a constraint ${\mathcal {\widetilde W}}_{R}^{i}(u_0,u_1,\dots,u_{i-1}, u_i)$ then replace this constraint by 
    ${\mathcal {\widetilde S}}_{R}^{i}(u_0,u_1,\dots,u_{i-1})$;
    \item[(r)] remove some constraint;
    \item[(j)] suppose $u_{i}\in\Var(\mathcal T)$ appears in constraints
    $Q(u_0,u_1,\dots,u_i,\dots,u_j)$ and
    ${\mathcal {\widetilde W}}_{R}^{i}(v_{0},v_{1},\dots,v_{i-1},u_i)$,
    where $Q\in\{{\mathcal {\widetilde W}}_{R}^{j},{\mathcal {\widetilde S}}_{R}^{j}\}$ and 
    $j\ge i$. 
    Then we identify the variables 
    $v_{k}=u_{k}$ for every $k\in\{0,1,\dots,i-1\}$
    and remove 
    the constraint 
    ${\mathcal {\widetilde W}}_{R}^{i}(v_{0},v_{1},\dots,v_{i-1},u_i)$.
    
\end{enumerate}

Notice that transformations 
(w) and (r) make the instance weaker (more solutions) 
but (s) and (j) make the instance stronger (less solutions).
The next lemma shows that if 
a tree-covering without a solution cannot be simplified using the transformations and it is 
large enough then 
we can cut it into three pieces satisfying nice properties.

\begin{lem}\label{LEMSplitPathIntoThreeParts}
Suppose  
\begin{enumerate}
    \item $R\subseteq A^{2n+1}$, where $n>0$;
    \item $\mathcal T$ is a tree-covering
of $\mathcal{\widetilde I}_{R}$ having  at least $(n\cdot |A|)^{2^{{2|A|}^{|A|+1}}}$ variables;
    \item  $\prescript{\zv_{0}}{}{\mathcal T}$ has no solutions for some $\zv_{0}\in A^{|A|}$;
    \item applying transformations (w) and (r) to $\mathcal T$ gives 
    an instance with a solution for every $\zv\in A^{|A|}$;
    \item transformations (s) and (j) are not applicable.
\end{enumerate}
Then $\mathcal T$ can be divided into 3 nonempty connected parts
$\mathcal I_1$, $\mathcal I_2$, and $\mathcal I_3$ 
such that 
\begin{itemize}
    \item[(l1)] the only common variable of $\mathcal I_1$ and 
$\mathcal I_2$ is a variable $u$;
    \item[(r1)] the only common variable of 
$\mathcal I_2$ and $\mathcal I_3$ is a variable $v$;
    \item[(l2)] 
    $\prescript{\zv}{}{\mathcal I_1(u)} = (\prescript{\zv}{}{\mathcal I_1}\wedge \prescript{\zv}{}{\mathcal I_2})(v)$
    for every $\zv\in A^{|A|}$;
    \item[(r2)] 
    $\prescript{\zv}{}{\mathcal I_3(v)} = (\prescript{\zv}{}{\mathcal I_2}\wedge \prescript{\zv}{}{\mathcal I_3})(u)$
    for every $\zv\in A^{|A|}$;
    \item[(m)] $\mathcal I_2$ contains a constraint 
    $\mathcal {\widetilde S}_{R}^{i}(v_0,\dots,v_{i})$ with $i<n-1$.
\end{itemize} 
\end{lem}

\begin{proof}
Let $N = |\Var(\mathcal T)|$, $k = n\cdot|A|$.
Notice that the arity of constraints in $\mathcal {\widetilde I}_{R}$ is 
at most $n+1$ if we ignore $\zv$-variables.
Since we cannot remove any constraints (property 4), Lemmas
\ref{LEMIndexOfAVariable} and \ref{LEMMinPathInATree} imply that there exists 
a path 
$u_0 - C_{1} -u_1 -C_2-\dots-C_{s}-u_{s}$ 
with 
$$s\ge \lceil{\log_{k} (N\cdot (k-1)+1)}\rceil\ge \log_{k} N
\ge 2^{{2|A|}^{|A|+1}}={(2^{|A|})}^{2|A|^{|A|}}> 1+{(2^{|A|}-1)}^{2|A|^{|A|}}.$$
For every $i\in[s-1]$ 
we split 
$\mathcal T$ into two tree-coverings 
$\Phi_{i}$ and $\Psi_{i}$ as follows.
If we split $\mathcal T$ in the variable $u_{i}$, then 
the part containing $C_{i+1}$ goes to $\Psi_{i}$ 
and all the remaining parts go to $\Phi_{i}$.
Thus, the only common variable 
of $\Phi_{i}$ and $\Psi_{i}$ is $u_{i}$.
Let $B_{i}$ and $C_{i}$ be the $\zv$-parameterized unary relations 
defined by $\Phi_{i}(u_{i})$ and $\Psi_{i}(u_{i})$,
respectively.

Since the transformation (r) cannot be applied, 
$\prescript{\zv}{}B_{i}$ and $\prescript{\zv}{}C_{i}$ are nonempty for every $i\in[s-1]$ and $\zv\in A^{|A|}$.
Since $\prescript{\zv_{0}}{}{\mathcal T}$ has no solutions, 
$\prescript{\zv_{0}}{}B_{i}\cap \prescript{\zv_{0}}{}C_{i}=\varnothing$ for every $i\in[s-1]$. 

There are exactly ${(2^{|A|}-1)}^{|A|^{|A|}}$ distinct nonempty 
$\zv$-parameterized unary relations.
Since $s> 1+{(2^{|A|}-1)}^{2|A|^{|A|}}$, there should be 
$1\le i<i'\le s-1$ such that 
$\prescript{\zv}{}B_{i} = \prescript{\zv}{}B_{i'}$ and 
$\prescript{\zv}{}C_{i} = \prescript{\zv}{}C_{i'}$
for every $\zv\in A^{|A|}$.
By $\mathcal I_{1}$ we denote $\Phi_{i}$, 
by $\mathcal I_{3}$ we denote $\Psi_{i'}$,
and by $\mathcal I_{2}$ we denote the intersection of $\Psi_{i}$ and 
$\Phi_{i'}$.
Then conditions 
(l1), (l2), (r1), and (r2) are satisfied
for $u = u_{i}$ and $v= u_{i'}$.
It remains to prove property (m).

We say that a variable is \emph{a leaf} if it occurs only in one constraint.
We say that a constraint is \emph{a leaf-constraint} if
only one of its variables is not a leaf.
Suppose a leaf-constraint is 
 ${\mathcal {\widetilde W}}_{R}^{j}(u_0,u_1,\dots,u_{j})$ 
 for some $j$. 
 If its nonleaf variable is $u_{j'}$, where $j'<j$,
 then 
 we can use the transformation (s).
If its nonleaf variable is $u_{j}$,
then from the definition 
of ${\mathcal {\widetilde S}}_{R}^{i}$ and ${\mathcal {\widetilde W}}_{R}^{i}$
we derive that this leaf-constraint can be removed from the instance not changing the property that it has no solutions for some $\zv$.
Both situations contradict our assumptions that 
(r) and (s) are not applicable.
Hence, any leaf-constraint of $\mathcal T$
is of the form 
${\mathcal {\widetilde S}}_{R}^{j}(v_0,v_1,\dots,v_{j})$.
Notice that $j$ must be smaller than $n-1$
because 
otherwise ${\mathcal {\widetilde S}}_{R}^{j}={\mathcal {\widetilde W}}_{R}^{j}$ and the transformation (w) does not really change the instance
and can always be applied.
Consider several cases:


Case 1. The constraint $C_{i+1}$ is 
${\mathcal {\widetilde S}}_{R}^{j}(v_0,v_1,\dots,v_{j})$ for some $j$.
Then $C_{i+1}$ is the required constraint to satisfy (m).

Case 2. The constraint $C_{i+1}$ is 
${\mathcal {\widetilde W}}_{R}^{j}(v_0,v_1,\dots,v_{j})$
where $v_{j}\neq u_{i}$ and $v_{j}\neq u_{i+1}$.
If $v_{j}$ is a leaf, we can apply the transformation (s).
Otherwise, consider a part of $\mathcal I_{2}$ containing $v_{j}$. This part must contain a leaf-constraint, which implies property (m).

Case 3. The constraint $C_{i+1}$ is 
${\mathcal {\widetilde W}}_{R}^{j}(v_0,v_1,\dots,v_{j-1},u_{i})$.
Then we apply transformation (j) to $C_{i}$ and $C_{i+1}$, which contradicts our assumptions.

Case 4. The constraint $C_{i+1}$ is 
${\mathcal {\widetilde W}}_{R}^{j}(v_0,v_1,\dots,v_{j-1},u_{i+1})$.
Then we apply transformation (j) to $C_{i+1}$ and $C_{i+2}$, which again contradicts our assumptions.
\end{proof}

\begin{cor}\label{CORBigTreeImpliesMightyTupleI}
Suppose  
\begin{enumerate}
    \item $R\subseteq A^{2n+1}$, where $n>0$;
    \item $\mathcal T$ is a tree-covering
of $\mathcal{\widetilde I}_{R}$ with the minimal number of variables 
such that 
$\prescript{\zv}{}{\mathcal T}$ has no solutions for some 
$\zv$;
\item $|\Var(\mathcal T)|\ge (n\cdot |A|)^{2^{{2|A|}^{|A|+1}}}$.
\end{enumerate}
Then $R$ q-defines a mighty tuple III.
\end{cor}

\begin{proof}
We apply transformations 
(w) and (r) to  $\mathcal T$ while we can maintain the condition 
that $\prescript{\zv}{}{\mathcal T}$ has no solutions for some 
$\zv$.
Also, we apply transformations (s) and (j) when applicable.
Notice that we cannot apply these transformations forever, and we never increase the number of variables.
Thus, the obtained 
tree-instance $\mathcal T'$ still contains the minimal number of variables 
and satisfies the same conditions.

By Lemma \ref{LEMSplitPathIntoThreeParts} we can split 
$\mathcal T'$ into 3 parts $\mathcal I_1$,
$\mathcal I_2$, and
$\mathcal I_3$ satisfying conditions 
(l1), (r1), (l2), (r2), and (m).
Let 
$\mathcal I_{2}'$ be obtained from $\mathcal I_2$ 
by replacing the constraint 
$\mathcal {\widetilde S}_{R}^{i}(v_0,\dots,v_{i})$ coming from 
condition (m)
by 
$\mathcal {\widetilde W}_{R}^{i}(v_0,\dots,v_{i})$.
Notice that 
$\mathcal I_{1}\wedge \mathcal I_{2}'\wedge \mathcal I_{3}$ 
has a solution for every $\zv$.
Let
$\mathcal I_{1}(u)$ define $B$, 
$\mathcal I_{3}(v)$ define $C$,
$\mathcal I_{2}(u,v)$ define $S$, and 
$\mathcal I_{2}'(u,v)$ define $W$.
By Lemma \ref{LEMSTildeIsUniversalInWTilde},
$\mathcal I_{2}\useq \mathcal I_2'$, hence by Lemma \ref{LEMUniversalSubuniverseImplies}
we have 
$S\useq W$.
Let $Q$ be a $(\zv,\alpha)$-parameterized relation q-definable from $R$ such that 
$Q^{\forall}  = W$ and 
$Q^{\forall\forall} =S$.
Let us show that 
$(Q,B,C)$ is a mighty tuple III.
Condition 1 follows from the fact that 
the transformation (r) gives an instance with a solution 
for every $\zv$.
Condition 2 follows from (l2), 
condition 3 follows from (r2), 
condition 4 follows from the existence of 
a solution of $\mathcal I_{1}\wedge \mathcal I_{2}'\wedge \mathcal I_{3}$ for every $\zv$.
Condition 5 follows from the fact that 
$\prescript{\zv_{0}}{}{\mathcal T'} = 
\prescript{\zv_{0}}{}{\mathcal I_{1}}\wedge \prescript{\zv_{0}}{}{\mathcal I_{2}}\wedge \prescript{\zv_{0}}{}{\mathcal I_{3}}$
has no solutions for some $\zv_{0}$.
\end{proof}

Now we are ready to prove two theorems from Section \ref{SECTIONMainProof}.

\begin{THMFindSmallTreeTHM}
Suppose $R\subseteq A^{2n+1}$.
Then one of the following conditions holds:
\begin{enumerate}
    \item there exists a $\zv$-parameterized
nonempty 1-consistent reduction for 
$\mathcal I_{R}$;
    \item  
there exists a subinstance  
$\mathcal J\subseteq \mathcal I_{R}$ with at most 
$(n\cdot |A|)^{2^{{2|A|}^{|A|+1}}}$ variables not having a solution for some $\zv\in A^{|A|}$;
\item there exists
    a mighty tuple III q-definable from $R$.
\end{enumerate}
\end{THMFindSmallTreeTHM}

\begin{proof}
Let us consider a maximal $\zv$-parameterized 
1-consistent reduction for $\mathcal {\widetilde I_{R}}$.
By Lemma \ref{LEMMaximalReductionSimple} 
either this reduction is nonempty, 
or 
there exists a $\zv$-parameterized tree-covering $\mathcal T$ of 
$\mathcal {\widetilde I_{R}}$ such that 
the instance 
$\prescript{\zv}{}{\mathcal T}$ has no solutions for some $\zv$.
In the first case the same reduction is also a
nonempty 1-consistent reduction for $\mathcal I_{R}$,
and we satisfied 
condition 1. 

In the second case we consider a tree-covering $\mathcal T$ with the minimal number of variables. 
If $|\Var(\mathcal T)|\ge (n\cdot |A|)^{2^{{2|A|}^{|A|+1}}}$,
then
Corollary \ref{CORBigTreeImpliesMightyTupleI}
implies that a mighty tuple III is q-definable from $R$.
If $|\Var(\mathcal T)|< (n\cdot |A|)^{2^{{2|A|}^{|A|+1}}}$,
then let 
$\mathcal J$ be the subinstance of $\mathcal I_{R}$ 
containing all the constraints $C$ of $\mathcal I_{R}$ 
such that a child of each variable of $C$ appears in $\mathcal T$.
Notice that if $\prescript{\zv}{}{\mathcal J}$ has a solution then 
$\prescript{\zv}{}{\mathcal T}$ has a solution.
Thus, $\prescript{\zv}{}{\mathcal J}$ has no solutions 
for some $\zv$.
Hence $|\Var(\mathcal J)|\le |\Var(\mathcal T)|< (n\cdot |A|)^{2^{{2|A|}^{|A|+1}}}$.
\end{proof}

\begin{THMNonemptyReductionIsZigzagTHM}
Suppose $R\subseteq A^{2n+1}$, 
$D^{(\top)}$ is an inclusion-maximal $\zv$-parameterized 
1-consistent nonempty reduction 
for $\mathcal I_{R}$.
Then 
$D^{(\top)}$ is a universal reduction.
\end{THMNonemptyReductionIsZigzagTHM}

\begin{proof}
Choose some variable 
$u\in\Var(\mathcal I_{R})$ and 
prove that $D^{(\top)}_{u}\uuus D^{(\top,0)}_{u}$.

By Lemma \ref{LEMMaximalReductionSimple}, 
there exists a tree-covering $\Upsilon_{0}$ of $\mathcal I_{R}$ 
such that 
$\prescript{\zv}{}\Upsilon_{0}(u)$ defines $\prescript{\zv}{}D_{u}^{(\top)}$
for every $\zv$.
We apply the following transformations 
to $\Upsilon_{0}$ similar to the transformations we used before:
\begin{enumerate}
    \item[(w)]
    replace a constraint ${\mathcal {S}}_{R}^{i}(u_0,u_1,\dots,u_{i})$ by 
    ${\mathcal {W}}_{R}^{i}(u_0,u_1,\dots,u_{i})$; 
    \item[(s)] if $u_{i}$ appears only once in the instance 
    in a constraint ${\mathcal W}_{R}^{i}(u_0,u_1,\dots,u_{i-1}, u_i)$, and $u_{i}\neq u$; then replace the constraint by 
    ${\mathcal S}_{R}^{i}(u_0,u_1,\dots,u_{i-1})$;
    \item[(j)] suppose $u_{i}\in\Var(\mathcal T)$ appears in constraints    
    ${\mathcal W}_{R}^{i}(u_0,u_1,\dots,u_{i-1},u_i)$ and
    ${\mathcal W}_{R}^{i}(v_{0},v_{1},\dots,v_{i-1},u_i)$,
    and 
    $u\notin\{v_1,\dots,v_{i-1}\}$; 
    then we identify the variables 
    $v_{k}=u_{k}$ for every $k\in\{0,1,\dots,i-1\}$
    and remove 
    the constraint 
    ${\mathcal W}_{R}^{i}(v_{0},v_{1},\dots,v_{i-1},u_i)$.
\end{enumerate}

Notice that transformation
(w) makes the instance weaker (more solutions) 
but (s) and (j) make the instance stronger (less solutions).

We apply transformations 
(w), (s), and (j) in any order and 
obtain a sequence 
$\Upsilon_{0},\Upsilon_{1},\dots,\Upsilon_{t}$
of 
tree-coverings of $\mathcal I_{R}$.
Notice that at least one transformation is applicable 
unless the lowest variable of $\Upsilon_{t}$ is $u$
and $u$ appears exactly once.
Let the constraint containing $u$ be 
${\mathcal W}_{R}^{i}(u_0,u_1,\dots,u_{i-1},u)$.
Since $D^{(\top)}$ is a 1-consistent reduction
and $\Upsilon_{t}$ is a tree-covering, 
$D_{u_{k}}^{(\top)}\subseteq \Upsilon_{t}(u_{k})$
for every $k\in\{0,1,\dots,i-1\}$.
Hence, 
$\Upsilon_{t}(u)\supseteq D_{u}^{(\top,0)}$.
Let 
the $\zv$-parameterized unary relation 
$C_{j}$ be defined by 
$\Upsilon_{j}(u)$
for $j=0,1,2,\dots,t$.
By the construction and Lemma \ref{LEMUniversalSubuniverseImplies}
we have $C_{j}\supseteq C_{j+1}$ or
$C_{j}\useq C_{j+1}$ for every $j$.
Put $E_{j} = C_{j}\cap C_{j+1}\cap \dots\cap C_{t}\cap D_{u}^{(\top,0)}$ for  
$j=0,1,2,\dots,t$.
Since the reduction $D^{(\top)}$ is 1-consistent and 
each $\Upsilon_{j}$ is a tree-covering, 
we have $C_{j}\supseteq D_{u}^{(\top)}$.
Hence, $E_{0}=C_{0} = \Upsilon_{0}(u) = D_{u}^{(\top)}$
and $E_{t} = D_{u}^{(\top,0)}$.
By Lemma \ref{LEMUniversalSubuniverseImplies}
$E_{j}\useq E_{j+1}$ for every $j\in\{0,1,\dots,t\}$.
Thus, $D_{u}^{(\top)}\uuus D_{u}^{(\top,0)}$
and $D^{(\top)}$ is a universal reduction.
\end{proof}



\subsection{The existence of a universal subset}\label{SUBSECTIONUniversalSubuniverseExists}

In this section we prove that 
for any 1-consistent $\zv$-parameterized universal reduction 
$D^{(\top)}$ 
of $\mathcal I_{R}$ 
there exists a $\zv$-parameterized unary relation $B$ and 
a variable $y_{i}^{a_1,\dots,a_{i}}$ such that
$B\us D_{y_{i}^{a_1,\dots,a_{i}}}^{(\top)}$.

For a sequence $a_1,\dots,a_m$, where 
$m\in\{0,1,\dots,n\}$ and it can be empty,
we put
$${\mathcal I^{a_1,\dots,a_m}_{R}}^{(\top)}
=\bigwedge\limits_{i=m}^{n}\;\;\bigwedge\limits_{a_{m+1},\dots,a_i\in A}
C_{S,\top}^{a_1,\dots,a_i}
.$$
Thus, ${\mathcal I^{a_1,\dots,a_m}_{R}}^{(\top)}$ is the part of 
$\mathcal I_{R}^{(\top)}$ containing the variable 
$y_{m}^{a_1,\dots,a_{m}}$.


\begin{THMFindUniversalSubuniverseTHM}
Suppose $R\subseteq A^{2n+1}$, $\mathcal I_{R}$ has no solutions for some $\zv$, 
$D^{(\top)}$ is a $\zv$-parameterized
universal 1-consistent reduction for 
$\mathcal I_{R}$.
Then one of the following conditions holds:
\begin{enumerate}
    \item 
    there exists a variable $u$ of $\mathcal I_{R}$ and
    a $\zv$-parameterized nonempty unary relation $B$ such that 
    $B\us D_{u}^{(\top)}$;
   \item  there exists a mighty tuple V q-definable from $R$.
\end{enumerate}
\end{THMFindUniversalSubuniverseTHM}

\begin{proof}
Since $\mathcal I_{R}$ has no solutions
for some $\zv$, 
${\mathcal I_{R}}^{(\top)}$ also does not have 
solutions for some $\zv$.
Consider maximal $m$ such that 
${\mathcal I_{R}^{c_1,\dots,c_{m}}}^{(\top)}$ has no solutions for some 
$c_1,\dots,c_m\in A$ and $\zv$.
We fix $c_1,\dots,c_m$
and denote $\mathcal I_0= {\mathcal I_{R}^{c_1,\dots,c_{m}}}^{(\top)}$.
Then we apply the following transformations to the instance $\mathcal I_{0}$ 
while possible
to obtain a sequence of instances
$\mathcal I_{0},\mathcal I_1,\dots,\mathcal I_{T}$.

\begin{enumerate}
    \item[(w)] replace the constraint $C_{S,\top}^{a_1,\dots,a_i}$
    by $C_{W,\top}^{a_1,\dots,a_i}$;
    \item[(e)] if 
    a variable $y_{i}^{a_1,\dots,a_i}$, where $i>m$, appears only once in the instance in a constraint 
    $C_{W,\top}^{a_1,\dots,a_i}$,
    then replace $C_{W,\top}^{a_1,\dots,a_i}$ by $C_{W,\top,|A|}^{a_1,\dots,a_i}$;
    \item[(z)] replace $C_{W,\top,j}^{a_1,\dots,a_i}$ by 
    $C_{W,\top,j-1}^{a_1,\dots,a_i}$;
    \item[(s)] replace $C_{W,\top,0}^{a_1,\dots,a_i}$ by 
    $C_{S,\top}^{a_1,\dots,a_{i-1}}$.
\end{enumerate}

Notice that 
(w) replaces an instance by its universal weakening, 
(s) makes the instance stronger, 
(e) just existentially quantifies a variable that appears only once,
(z) replaces a constraint by its universal weakening. 

It follows from the definition that we cannot apply these transformations forever,
and we can never remove 
the constraint $C_{W,\top}^{c_{1},\dots,c_{m}}$.
Let us show that 
the final instance 
$\mathcal I_{T}$ consists of just one constraint 
$C_{W,\top}^{c_{1},\dots,c_{m}}$.
If $\mathcal I_{T}$ contains some
$C_{S,\top}^{a_{1},\dots,a_{i}}$
or $C_{W,\top,j}^{a_1,\dots,a_i}$, 
then we can apply the transformations (w) , (z), or (s), 
which contradicts the assumption that $\mathcal I_{T}$ is final.
Otherwise,
let $y_{i}^{a_1,\dots,a_i}$ be the lowest variable of $\mathcal I_{T}$.
If $i>m$, then
$y_{i}^{a_1,\dots,a_i}$ appears only in the constraint
$C_{W,\top}^{a_{1},\dots,a_{i}}$ and we can apply transformation (e), which again contradicts the assumption.
Thus, $i=m$ and the only constraint 
in $\mathcal I_{T}$  is 
$C_{W,\top}^{c_{1},\dots,c_{m}}$.

Consider the last instance $\mathcal I_{t}$ in the sequence
$\mathcal I_{0},\mathcal I_{1},\dots,\mathcal I_{T}$
not satisfying the following property: 
$\mathcal I_{t}$ has a solution 
with $y_{i}^{c_{1},\dots,c_{i}}=d$ for 
any $\zv\in A^{|A|}$, any $i\in\{0,\dots,m\}$, and any $d\in \prescript{\zv}{}D_{y_{i}^{c_{1},\dots,c_{i}}}^{(\top)}$.
We refer to this property 
as \emph{the subdirectness property}.
Since $\mathcal I_{0}$ does not satisfy the subdirectness property,
such $t$ always exists.
Notice that 
$\mathcal I_{T}$ and even $\mathcal I_{T-1}$ 
satisfy the subdirectness property, hence $t<T-1$.
By the definition of $t$, the instance 
$\mathcal I_{t+1}$ cannot be stronger than $\mathcal I_{t}$.
Hence, $\mathcal I_{t}\useq \mathcal I_{t+1}$.
Consider two cases:

Case 1. $\mathcal I_{t}$ has a solution for every $\zv\in A^{|A|}$.
Then consider some variable 
$y_{i}^{c_{1},\dots,c_{i}}$ wintessing that 
$\mathcal I_{t}$ does not have the subdirectness property.
Since $\mathcal I_{t+1}$ has 
the subdirectness property,
$\mathcal I_{t+1}(y_{i}^{c_{1},\dots,c_{i}})$ 
defines $D_{y_{i}^{c_{1},\dots,c_{i}}}^{(\top)}$.
Let $\mathcal I_{t}(y_{i}^{c_{1},\dots,c_{i}})$ define 
a $\zv$-parameterized unary relation $B$.
Since $\mathcal I_{t}\useq \mathcal I_{t+1}$, 
Lemma \ref{LEMUniversalSubuniverseImplies}
implies 
$B\us D_{y_{i}^{c_1,\dots,c_{i}}}^{(\top)}$,
which satisfies condition 1 and completes this case.

Case 2. $\mathcal I_{t}$ does not have a solution for some $\zv\in A^{|A|}$.
Put $\mathcal J_{0} = \mathcal I_{t}$.
Then we apply another transformation 
to $\mathcal I_{t}$
and obtain a sequence of 
instances $\mathcal J_{0},\mathcal J_{1},\dots,\mathcal J_{s}$, 
where $\mathcal J_{0} = \mathcal I_{t}$.
If a variable $u$ is a child of $y_{i}^{c_{1},\dots,c_{i}}$, where $i\in\{0,1,\dots,m\}$, 
and $u$ appears several times 
in $\mathcal J_{k}$, then we rename some (but not all) of the variables $u$ into $u'$ and obtain 
a covering $\mathcal J_{k+1}$ of $\mathcal J_{k}$.
If $\mathcal J_{k+1}$ has a solution 
for every $\zv$, we finish the sequence.
Notice that if we split all the children of  each $y_{i}^{c_{1},\dots,c_{i}}$ so that 
each of them appears exactly once, 
then the obtained instance has a solution for every $\zv$ 
by the maximality of $m$.
Thus, we get a sequence $\mathcal J_{0},\mathcal J_{1},\dots,\mathcal J_{s}$
 of coverings of $\mathcal I_{t}$
such that $\mathcal J_{s}$ has a solution for every $\zv$.

Since $\mathcal I_{t+1}$ is a universal weakening of 
$\mathcal I_{t}$,
this universal weakening can be transferred to $\mathcal J_{s}$,
where we replace the child of every constraint of $\mathcal I_{t}$
by the corresponding weakened version.
As a result we get 
a universal weakening $\mathcal J_{s}'$
of $\mathcal J_{s}$.
Notice that $\mathcal J_{s}'$ is a covering of $\mathcal I_{t+1}$, 
which implies that 
$\mathcal J_{s}'$ satisfies the modification of the subdirectness property 
for coverings. 
That is, 
$\mathcal J_{s}'$ has a solution 
with $v=d$, 
if $v$ is a child of $y_{i}^{c_{1},\dots,c_{i}}$,
for any $\zv\in A^{|A|}$, any $i\in\{0,\dots,m\}$, and any $d\in \prescript{\zv}{}D_{y_{i}^{c_{1},\dots,c_{i}}}^{(\top)}$.
Let $u$ be the variable we split while defining 
$\mathcal J_{s}$ and $u'$ be the new variable we added.
Consider two subcases:

Subcase 2A. 
There exist $\zv$ and  
$d\in \prescript{\zv}{}D_{u}^{(\top)}$
such that $\prescript{\zv}{}{\mathcal J_{s}}$ has no solution 
with $u = d$ or has no solution with $u' = d$.
Without loss of generality let it be $u$.
By the subdirectness property for $\mathcal J_{s}'$, 
the formula
$\mathcal J_{s}(u)$ defines the $\zv$-parameterized unary relation 
$D_{u}^{(\top)}$.
Suppose 
$\mathcal J_{s}(u)$ defines
a $\zv$-parameterized unary relation $B$.
Since $\mathcal J_{s}\useq \mathcal J_{s}'$, 
Lemma \ref{LEMUniversalSubuniverseImplies} 
implies that $B\us D_{u}^{(\top)}$ 
and $B$ satisfies condition 1, which  completes this case.

Subcase 2B. 
For every $\zv$ and every 
$d\in D_{u}^{(\top)}$
the instance $\mathcal J_{s}$ has a solution 
with $u = d$ and a solution with $u' = d$.
Let $S$ be the binary $\zv$-parameterized 
relation
defined by $\mathcal J_{s}(u,u')$ and 
$W$ be the binary $\zv$-parameterized 
relation
defined by $\mathcal J_{s}'(u,u')$.
Since $\mathcal J_{s}\useq \mathcal J_{s}'$, 
Lemma \ref{LEMUniversalSubuniverseImplies} 
implies that $S\useq W$. 
Let $Q$ be a $(\zv,\alpha)$-parameterized relation q-definable from $R$ witnessing $S\useq W$, 
that is, 
$Q^{\forall}  = W$ and 
$Q^{\forall\forall} =S$.
Let us show that 
$(Q,D)$ forms 
a mighty tuple V.
Property 1 follows from the subdirectness of 
$\mathcal J_{s}'$.
Property 2 follows from the definition of subcase 2B.
Since $\mathcal J_{s-1}$ has no solutions for some $\zv$, 
$\prescript{\zv}{}{Q}^{\forall\forall}$ is irreflexive for this $\zv$, and we get property 3.
\end{proof}

\subsection{Finding a smaller reduction}\label{SUBSECTIONFindingSmallerReduction}

In this section we will show that 
for any 1-consistent universal reduction 
$D^{(\top)}$ for 
$\mathcal I_{R}$ and a unary $\zv$-parameterized 
relation $B\us D_{y_{i}^{a_1,\dots,a_{i}}}^{(\top)}$
we can build a
smaller 1-consistent universal reduction 
$D^{(\bot)}\subsetneq D^{(\top)}$.

By $\mathcal I_{R}'$ we denote 
the instance $\mathcal I_{R}^{(\top)}$
with additional constraints 
$C_{W,\top}^{a_1,\dots,a_{i}}$, 
$C_{W,\top,j}^{a_1,\dots,a_{i}}$,
for all $i\in\{0,1,\dots,n\}$, 
$a_1,\dots,a_{i}\in A$, 
and $j\in \{0,1,\dots,|A|\}$. 
Notice that all the constraints we added to $\mathcal I_{R}$ to get $\mathcal I_{R}'$ are 
weaker than the constraints that are already there.
Hence, $\mathcal I_{R}'$ has a solution if and only if 
$\mathcal I_{R}$ has a solution, 
and a reduction is 1-consistent for 
$\mathcal I_{R}'$ if and only if it is 1-consistent for $\mathcal I_{R}$.

To simplify presentation we 
fix a highest variable 
$y_{m}^{c_1,\dots,c_{m}}$ such that 
there exists a $\zv$-parameterized unary relation $B$
satisfying 
$B\us D_{y_{m}^{c_1,\dots,c_{m}}}^{(\top)}$.
Denote this variable by $u=y_{m}^{c_1,\dots,c_{m}}$.
By $\mathcal B$ we denote the 
set of all  $\zv$-parameterized nonempty unary relations $B$
satisfying 
$B\us D_{u}^{(\top)}$.


By $\mathcal T$ we denote the set of all tree-coverings 
of $\mathcal I_{R}'$ 
such that some of the children of $u$ are marked as 
leaves and exactly one child of $u$ is marked as
the root. 
Elements of $\mathcal T$ are called \emph{trees}. Notice that a vertex can be simultaneously 
a leaf and the root.
Any path in an instance $\mathcal I_{R}'$ can be viewed as a 
tree-covering. Marking the first element of the path as 
a leaf and the last element as the root we can 
make a tree from any path.
By $\mathcal P$ we denote the set of all paths in $\mathcal T$. 
Notice that by choosing one leaf in a tree we can always make a path from a tree.

Suppose $t\in \mathcal T$ 
with leaves $u_1,\dots,u_{s}$ and the root $u_0$. 
For a $\zv$-parameterized unary relation $B$  
by $B+t$ we denote the $\zv$-parameterized unary relation 
defined by 
$(t\wedge u_{1}\in B\wedge \dots\wedge u_{s}\in B)(u_{0})$. 
Informally, $B+t$ is the restriction we 
get on the root if restrict we all the leaves to $B$.
Notice that, since the reduction $D^{(\top)}$ is 
1-consistent, we have 
$D_{u}^{(\top)}+t = D_{u}^{(\top)}$ 
for any $t\in \mathcal T$.

To prove the existence of a smaller reduction, 
it is sufficient to satisfy the following Lemma.

\begin{lem}\label{LEMIntersectEachImpliesSmallerReduction}
Suppose $B\in\mathcal B$ and $\prescript{\zv}{}{B}+\prescript{\zv}{}t\neq \varnothing$ for every $t\in\mathcal T$ and every $\zv\in A^{|A|}$.
Then there exists 
a $\zv$-parameterized 1-consistent universal reduction 
$D^{(\bot)}$ for $\mathcal I_{R}$ such that
$D^{(\bot)}\subsetneq D^{(\top)}$.
\end{lem}

\begin{proof}
Let $D^{(\triangle)}$ be the $\zv$-parameterized reduction
for $\mathcal I_{R}$ such that 
$D_{u}^{(\triangle)}=B$
and
$D_{v}^{(\triangle)}=D_{v}^{(\top)}$
for every $v\neq u$.
For every $\zv\in A^{|A|}$ let 
$\prescript{\zv}{}D^{(\bot)}$ be the inclusion maximal 1-consistent reduction 
such that $\prescript{\zv}{}D^{(\bot)}\subseteq \prescript{\zv}{}D^{(\triangle)}$.
Consider two cases:

Case 1. Assume that $\prescript{\zv_{0}}{}{D_{v}^{(\bot)}}$ is empty for
some $\zv_{0}\in A^{|A|}$ and some variable $v$.
Then by Lemma \ref{LEMExistenceOfTreeCoverings}  
there exists a tree-covering $\Upsilon$ of $\mathcal I_{R}'$ such that 
$\prescript{\zv_{0}}{}\Upsilon$ has no solutions if all the children of $u$ are restricted to $\prescript{\zv_{0}}{}B$.
Let $t\in\mathcal T$ be the tree obtained from $\Upsilon^{(\top)}$
by marking all the children of $u$ as leaves and marking one of the children 
as the root.
Then $\prescript{\zv_{0}}{}B+\prescript{\zv_{0}}{}t=\varnothing$, which contradicts our assumptions.

Case 2. $\prescript{\zv}{}{D_{v}^{(\bot)}}$ is not empty for every $\zv\in A^{|A|}$ and every $v$.
By Lemma \ref{LEMMaximalReductionGeneral}
for every $v$ there exists a tree-covering $\Upsilon_{v}$ of $\mathcal I_{R}$ 
such that  $\prescript{\zv}{}{\Upsilon_{v}^{(\triangle)}}(v)$ defines 
$\prescript{\zv}{}{D_{v}^{(\bot)}}$ for every $\zv$.
Since $D^{(\top)}$ is 1-consistent and $\Upsilon_{v}$ is a tree-formula, 
$\prescript{\zv}{}{\Upsilon_{v}^{(\top)}}(v)$ defines
$\prescript{\zv}{}{D_{v}^{(\top)}}$.
Since $B\us D^{(\top)}_{u}$,  Lemma \ref{LEMUniversalSubuniverseImplies}
implies that $D_{v}^{(\bot)}\useq D_{v}^{(\top)}$ for every $v$.
Again, by Lemma \ref{LEMUniversalSubuniverseImplies}
$$D_{v}^{(\bot)} = D_{v}^{(\bot)}\cap D_{v}^{(\bot,0)}\useq D_{v}^{(\top)}\cap D_{v}^{(\bot,0)}
\uuus D_{v}^{(\top,0)}\cap D_{v}^{(\bot,0)}=D_{v}^{(\bot,0)}.$$
Hence, $D^{(\bot)}$ is a $\zv$-parameterized 1-consistent
universal reduction for $\mathcal I_{R}$
that is smaller than $D^{(\top)}$.
\end{proof}

We define two directed graphs $G_{\mathcal P}$ and $G_{\mathcal T}$ whose vertices 
are elements of $\mathcal B$.
There is an edge $B_{1}\to B_{2}$ in 
$G_{\mathcal P}$
if there exists 
$p\in \mathcal P$ such that 
$B_{1}+p = B_{2}$.
Similarly, the edge $B_{1}\to B_{2}$ is in 
$G_{\mathcal T}$ 
if there exists 
$t\in \mathcal T$ such that 
$B_{1}+t = B_{2}$.
Since we consider trivial paths/trees, 
both graphs are reflexive (have all the loops).
Since we can compose paths and trees, 
$B_1\to B_2$ and $B_2\to B_3$ implies 
$B_1\to B_3$, that is both graphs are transitive.
Let $\mathcal B_{\mathcal T}$ be a strongly connected component 
of $G_{\mathcal T}$
not having edges going outside of the component.
Let $\mathcal B_{\mathcal P}$ be a strongly connected component 
of $G_{\mathcal P}$ inside $\mathcal B_{\mathcal T}$ 
not having edges going outside of the component.
Thus, 
we have 
$\mathcal B_{\mathcal P}\subseteq \mathcal B_{\mathcal T}\subseteq \mathcal B$.
Put $\mathcal B_{\mathcal P}^{*} = 
B_{\mathcal P}\cup\{D_{u}^{(\top)}\}$.
Then for every $B\in\mathcal B_{\mathcal P}^{*}$,
and $p\in\mathcal P$
we have 
$B+p\in \mathcal B_{\mathcal P}^{*}$.

In this section we prove that there exists a smaller 1-consistent reduction
for $\mathcal I_{R}$ or 
$R$ q-defines a mighty tuple IV.
As we show in Lemma \ref{LEMSmallerReductionOrPSpaceHardness}, to prove this, it is sufficient to satisfy the following property:
there exist $B\in \mathcal B_{\mathcal P}$ and 
$\zv$-parameterized binary relations 
$S$ and $W$ q-definable from $R$ such that
\begin{enumerate}
    \item $S\triangleleft W$; 
    \item $\prescript{\zv}{}{B}+\prescript{\zv}{}{S}=\prescript{\zv}{}{B}$ for every $\zv\in A^{|A|}$;
    \item 
    $\prescript{\zv}{}{B}+\prescript{\zv}{}{W}=\prescript{\zv}{}{D_{u}^{(\top)}}$ for every $\zv\in A^{|A|}$;
    \item 
     $\prescript{\zv}{}{D_{u}^{(\top)}}+\prescript{\zv}{}{S}=\prescript{\zv}{}{D_{u}^{(\top)}}$ for every $\zv\in A^{|A|}$.
\end{enumerate}
In this case we say that the tuple 
$(R,D^{(\top)},u,\mathcal B_{\mathcal P})$ is \emph{strong}.

\begin{lem}\label{LEMSmallerReductionOrPSpaceHardness}
Suppose $(R,D^{(\top)},u,\mathcal B_{\mathcal P})$ is a strong tuple.
Then one of the following conditions holds:
\begin{enumerate}
    \item there exists 
a $\zv$-parameterized 1-consistent universal reduction 
$D^{(\bot)}$ for $\mathcal I_{R}$ such that
$D^{(\bot)}\subsetneq D^{(\top)}$;
    \item there exists 
    a mighty tuple IV q-definable from $R$.
\end{enumerate}
\end{lem}

\begin{proof}
By the definition, there exist $B\in\mathcal B_{\mathcal P}$ and 
$\zv$-parameterized binary relations 
$S$ and $W$ q-definable from $R$ satisfying the required 4 conditions.
    If $\prescript{\zv}{}B+\prescript{\zv}{}t\neq \varnothing$ for any $t\in \mathcal T$ and $\zv$
then 
Lemma \ref{LEMIntersectEachImpliesSmallerReduction} implies that condition 1 holds.
Otherwise, let $t$ be the tree with the minimal number of leaves 
such that $\prescript{\zv}{}B+\prescript{\zv}{}t=\varnothing$. 
Define a new tree $t'$ by moving the root of $t$
to one of the leaves and removing its leaf mark.
By the minimality of the number of leaves
$\prescript{\zv}{}B+\prescript{\zv}{}t'\neq\varnothing$
for any $\zv$.
Denote $C  = B+t'$.
Since $S\us W$, there exists a q-definable relation $Q$
such that $Q^{\forall\forall} = S$ and 
$Q^{\forall} = W$.
Let us check that 
$(Q,D_{u}^{(\top)},B,C)$ forms 
a mighty tuple IV.
All the conditions but 5 follow from the definition
of a strong tuple.
Condition 5 follows from the 
fact that $\prescript{\zv}{}B+\prescript{\zv}{}t =\varnothing$ for some $\zv$ and therefore 
$\prescript{\zv}{}B\cap \prescript{\zv}{}C = \varnothing$ 
for this $\zv$.
\end{proof}

In the first case of the above lemma we obtain
a smaller reduction, and in the second case we can build a mighty tuple and therefore prove PSpace-hardness. 
Thus, whenever the tuple $(R,D^{(\top)},u,\mathcal B_{\mathcal P})$ is strong,  
we can achieve the required result. 
That is why,  
in many further lemmas we have an assumption
that it is not strong.

We say that $B_{1}$ is \emph{a supervised 
universal subset of $B_{2}$} 
if there exist
$B_{0}\in\mathcal B_{\mathcal P}$ and 
$p_{1},p_{2}\in\mathcal P$
such that 
$p_1\useq p_{2}$, 
$B_{0}+p_{1} = B_{1}$, 
and 
$B_{0}+p_{2} = B_{2}$. 
We write it as $B_{1}\suseq B_{2}$.

The following lemma follows immediately from the definition 
and the fact that we can compose paths.

\begin{lem}\label{LEMContinueSupervisedRelaxation}
Suppose $B_{1}\suseq B_{2}$ and
$p\in\mathcal P$.
Then $B_{1}+p\suseq B_{2}+p$.
\end{lem}


The following lemma is the crucial fact in the whole proof. We show that we are done whenever 
$B\suseq D_{u}^{(\top)}$, and in the next lemmas 
we just try to achieve 
the condition
$B\suseq D_{u}^{(\top)}$.

\begin{lem}\label{LEMSupervisedRelaxationImpliesPSpace}
Suppose 
$B\suseq D_{u}^{(\top)}$.
Then $(R,D^{(\top)},u,\mathcal B_{\mathcal P})$ is a strong tuple.
\end{lem}

\begin{proof}
Consider $B_0\in \mathcal B_{\mathcal P}$ and two paths $p_1,p_2\in \mathcal P$ 
such that 
$B_{0}+p_1 = B$, $B_{0}+p_2=D_{u}^{(\top)}$,
and $p_{1}\useq p_{2}$.
Consider a path $p_0\in \mathcal P$ such that 
$B+p_0= B_0$, and
define two new paths by
$p_1' = p_0+p_1$ and
$p_2' = p_0+p_2$.
Then 
$B + p_1' = B$, $B+p_{2}'=D_{u}^{(\top)}$, 
and $p_{1}'\useq p_{2}'$.
Let $u_{1}$ and $u_{2}$ be the two ends of the paths
$p_{1}'$ and $p_{2}'$.
Let $p_{1}'(u_1,u_2)$ define $S$,
$p_{2}'(u_1,u_2)$ define $W$.
Then $W$, $S$, and $B$
witness that $(R,D^{(\top)},u,\mathcal B_{\mathcal P})$ is a strong tuple.
\end{proof}


\begin{lem}\label{ultimatePropOne}
Suppose $B_{1},B_{2}\in\mathcal B_{\mathcal P}$.
Then for every $\zv$ either 
both $\prescript{\zv}{}B_{1}$ and $\prescript{\zv}{}B_{2}$ 
are different from $\prescript{\zv}{}D_{u}^{(\top)}$, 
or both are equal to $\prescript{\zv}{}D_{u}^{(\top)}$.
\end{lem}

\begin{proof}
Since $B_{1},B_{2}\in\mathcal B_{\mathcal P}$, 
there exist paths $p_1,p_2\in\mathcal P$ 
such that $B_{1}+p_1 = B_{2}$ and 
$B_{2} + p_2 = B_1$.
Since the reduction 
$D^{(\top)}$ is 1-consistent, 
$\prescript{\zv}{}D_{u}^{(\top)}+p = 
\prescript{\zv}{}D_{u}^{(\top)}$ for any $p\in\mathcal P$ and 
$\zv\in A^{|A|}$. This implies the required property.
\end{proof}

\emph{A supervised zig-zag} from $B_{1}$ to $B_{2}$
is a sequence 
$C_{0},C_1,\dots,C_{k}\in \mathcal B_{\mathcal P}^{*}$
such that 
\begin{itemize}
\item $C_{0} = B_{1}$, $C_{k} = B_{2}$;
\item $C_{i-1}\supseteq C_{i}$ or 
$C_{i-1}\suseq C_{i}$ for every $i\in[k]$.
\end{itemize}
If there exists a supervised zig-zag of length $k$ from 
$B_{1}$ to $B_{2}$, then we write 
$B_{1}\szigzag_{k} B_{2}$
or just 
$B_{1}\szigzag B_{2}$ if we do not want to specify the length.

\begin{lem}\label{LEMZigzagFromBoneToBtwo}
Suppose $B_{1},B_{2}\in \mathcal B_{\mathcal P}$.
Then $B_{1}\szigzag B_{2}$.
\end{lem}

\begin{proof}
Consider a path $p_{0}\in\mathcal P$ such that
$B_{2}+p_0=B_{1}$.
We will build a sequence of 
paths
$p_0,p_1,\dots,p_{k}$ 
such that 
$C_{i} = B_{2}+p_{i}$.
We have $C_{0} = B_{2}+p_0 = B_1$.

By the choice of the variable $u$ 
there is no 
$\zv$-parameterized unary relation $B\triangleleft D_{v}^{(\top)}$
for any variable $v$ that is above $u$ in $\mathcal I_{R}$.
Therefore, by Lemma \ref{LEMUniversalSubuniverseImplies}, 
these variables cannot appear in the path from the 
leaf to the root but can appear in some constraints. 

We apply the following transformations to  
$p_{0}$ and define a sequence 
$p_{0},p_1,\dots,p_{k}\in\mathcal P$.

\begin{enumerate}
    \item[(w)] replace a child of the constraint $C_{S,\top}^{a_1,\dots,a_i}$
    by the corresponding child of $C_{W,\top}^{a_1,\dots,a_i}$;
    \item[(e)] if the lowest variable of a child 
    of 
    $C_{W,\top}^{a_1,\dots,a_i}$ appears only once,
    then replace it by the corresponding child of $C_{W,\top,|A|}^{a_1,\dots,a_i}$;
    \item[(j)] if a variable from the $i$-th level
    appears in 
    two children of $C_{W,\top}^{a_1,\dots,a_i}$,
    then we replace these children 
    by one child of $C_{W,\top}^{a_1,\dots,a_i}$ identifying 
    the corresponding variables of the children;
    \item[(z)] replace a child of $C_{W,\top,j}^{a_1,\dots,a_i}$ by
    the corresponding child of 
    $C_{W,\top,j-1}^{a_1,\dots,a_i}$;
    \item[(s)] replace a child of $C_{W,\top,0}^{a_1,\dots,a_i}$ by 
    the corresponding child of $C_{S,\top}^{a_1,\dots,a_{i-1}}$.
\end{enumerate}

Notice that 
(w) replaces an instance by its universal weakening, 
(s) makes the instance stronger, 
(e) just existentially quantifies a variable that appears only once,
(j) joins several constraints together and makes the instance stronger,
(z) replaces the instance by its universal weakening. 
Notice that we do not apply (e) 
if the lowest variable is a child of $u$ because this would mean 
removing a root. 

Let us show that we can apply these transformation 
till the moment when we have only one 
constraint and this constraint is a child of 
$C_{W}^{c_1,\dots,c_m}$.
Suppose we already have 
$p_{0},\dots,p_{\ell}$. 
If $p_{\ell}$ 
has a child of $C_{W,j}^{a_1,\dots,a_i}$
or a child of $C_{S}^{a_1,\dots,a_i}$
then we can apply (w), (z), or (s).
Otherwise, let $v$ be the lowest variable of $p_{\ell}$.
Notice that $v$ has to be from a level below $u$, since otherwise 
$p_{\ell}$ already consists of just one constraint.
If $v$ appears only once, then we can apply (e).
Otherwise, we can apply (j).

Since we always reduce the tree and reduce the arity of a constraint, we cannot apply transformations 
forever.
Thus, we have 
the sequence 
$p_{0},p_{1},\dots,p_{k}\in \mathcal P$ 
such that for every $i$
either 
$p_{i+1}$ is stronger than $p_{i}$, or 
$p_{i}\suseq p_{i+1}$.
Since the last path in the sequence consists of
a child of $C_{W}^{c_1,\dots,c_m}$, 
we have 
$B_{2}+p_{k} = B_{2}$
and the sequence 
$C_{0},C_{1},C_{2},\dots,C_{k}$
witnesses that 
$B_{1}\szigzag B_{2}$.
\end{proof}

\begin{lem}\label{LEMCannotKillJustOne}
Suppose $B_{1},B_{2}\in \mathcal B_{\mathcal P}$,
$p\in\mathcal P$,
$B_1+p\neq D_{u}^{(\top)}$.
Then $B_2+p\neq D_{u}^{(\top)}$ 
or the tuple
$(R,D^{(\top)},u,\mathcal B_{\mathcal P})$ is strong.
\end{lem}

\begin{proof}
Assume that $B_2+p=D_{u}^{(\top)}$.
By Lemma \ref{LEMZigzagFromBoneToBtwo} there is 
a supervised zig-zag 
$C_0,C_1,\dots,C_{k}$
from $B_{1}$ to $B_{2}$.
Consider the last element in the sequence 
$C_0+p,C_1+p,\dots,C_{k}+p$ that 
is different from $D_{u}^{(\top)}$.
Let it be $C_{i}+p$.
Then by Lemma \ref{LEMContinueSupervisedRelaxation}, 
$C_{i}+p\suseq C_{i+1}+p=D_{u}^{(\top)}$, which by 
Lemma \ref{LEMSupervisedRelaxationImpliesPSpace} implies 
that 
the tuple
$(R,D^{(\top)},u,\mathcal B_{\mathcal P})$ is strong.
\end{proof}

\begin{lem}\label{LEMreduceZigZag}
Suppose 
$B_{1},B_{2}\in\mathcal B_{\mathcal P}$, $B_{1}\szigzag_{k} B_{2}$, and 
$\prescript{\zv}{}{B_{1}}\not\supseteq \prescript{\zv}{}{B_{2}}$
for some $\zv\in A^{|A|}$.
Then there exists $p\in\mathcal P$ such that
$B_{1}+p\szigzag_{k-1} B_{2}+p\neq D_{u}^{(\top)}$
or
the tuple
$(R,D^{(\top)},u,\mathcal B_{\mathcal P})$ is strong.
\end{lem}
\begin{proof}
Let $C_{0},\dots,C_k$ be a supervised zig-zag 
from $B_1$ to $B_2$.
Since $\prescript{\zv}{}{B_{1}}\not\supseteq \prescript{\zv}{}{B_{2}}$, 
$C_{i}\sus C_{i+1}$ for some $i$.
Choose an inclusion maximal $B\in\mathcal B_{\mathcal P}$
and a path $p\in\mathcal P$ such that 
$C_{i}+p = B$.
Unless the tuple
$(R,D^{(\top)},u,\mathcal B_{\mathcal P})$ is strong,
Lemma \ref{LEMCannotKillJustOne} implies that
$C_{j}+p\neq D_{u}^{(\top)}$ for every $j$.
Since $C_{i}\subseteq C_{i+1}$ and $B$ is inclusion maximal, we
have $C_{i}+p = C_{i+1}+p$.
Then by 
Lemma \ref{LEMContinueSupervisedRelaxation},
$C_{0}+p, C_1+p,\dots,C_{i}+p,C_{i+2}+p,\dots,C_k+p$
is a supervised zig-zag from $B_{1}+p$ to $B_{2}+p$ of length $k-1$.
\end{proof}

\begin{lem}\label{LEMmakeInclusion}
Suppose  
$B_{1},B_{2}\in\mathcal B_{\mathcal P}$.
Then there exists $p\in\mathcal P$ such that 
$ B_{1}+p\subseteq B_2+p\neq D_{u}^{(\top)}$
or
the tuple
$(R,D^{(\top)},u,\mathcal B_{\mathcal P})$ is strong.
\end{lem}

\begin{proof}
By Lemma \ref{LEMZigzagFromBoneToBtwo},  
$B_{2}\szigzag_{k}B_1$ for some $k$.
Applying 
Lemma \ref{LEMreduceZigZag} we obtain 
$B_2+p_1\szigzag_{k-1}B_1+p_1$.
Applying 
Lemma \ref{LEMreduceZigZag} again
we obtain  
$B_2+p_1+p_2\szigzag_{k-2}B_1+p_1+p_2$.
We can do this till the moment when
$B_2+p_1+p_2+\dots+p_s \supseteq B_1+p_1+p_2+\dots+p_s$.
It remains to put 
$p = p_1+p_2+\dots+p_s$.
\end{proof}

\begin{cor}\label{CORMakeTwoEqual}
Suppose  
$B_{1},B_{2}\in\mathcal B_{\mathcal P}$.
Then there exists $p\in\mathcal P$ such that 
$ B_{1}+p= B_2+p\neq D_{u}^{(\top)}$
or
the tuple
$(R,D^{(\top)},u,\mathcal B_{\mathcal P})$ is strong.
\end{cor}

\begin{proof}
By Lemma \ref{LEMmakeInclusion}
there exists $p$ such that 
$B_1+p \supseteq B_2+p$.
Again by 
Lemma \ref{LEMmakeInclusion}
there exists $p'$ such that 
$B_1+p+p'\subseteq B_2+p+p'$.
Combining this with $B_1+p \supseteq B_2+p$,
we obtain $B_1+p+p' = B_2+p+p'$.
\end{proof}

\begin{lem}\label{LEMUniversalTree}
Suppose $B\in\mathcal B_{\mathcal P}$.
Then there exists $p\in \mathcal P$ such that 
$B'+p=B$ for every $B'\in\mathcal B_{\mathcal P}$
or
the tuple
$(R,D^{(\top)},u,\mathcal B_{\mathcal P})$ is strong.
\end{lem}

\begin{proof}
First, let us show that there exists a path sending all 
$B'\in\mathcal B_{\mathcal P}$ to the same element of $\mathcal B_{\mathcal P}$.
Put $\mathcal B_{0} = \mathcal B_{\mathcal P}$.
If $|\mathcal B_{0}|=1$ then we can take a trivial path $p$.
Otherwise, consider different $B_{1},B_{2}\in\mathcal B_0$.
By Corollary \ref{CORMakeTwoEqual}, there exists a
path $p_1$ such that 
$B_1+p_1= B_2+p_1\neq D_{u}^{(\top)}$.
Put 
$\mathcal B_1 = \{B'+p_1\mid B'\in\mathcal B_0\}$.
If $|\mathcal B_{1}|=1$ then we finish.
Otherwise, choose different $B_{1},B_{2}\in\mathcal B_{1}$
and make them equal using a path $p_2$.
Then define 
$\mathcal B_2 = \{B'+p_2\mid B'\in\mathcal B_1\}$.
Proceeding this way 
we get paths $p_{1},\dots,p_{s}$ and 
sets $\mathcal B_{1},\dots,\mathcal B_{s}$ 
such that 
$\mathcal B_{i} = \{B'+p_i\mid B'\in\mathcal B_{i-1}\}$
and $|\mathcal B_{i}|<|\mathcal B_{i-1}|$
for $i=1,2,\dots,s$.
Notice that by Lemma \ref{LEMCannotKillJustOne} 
$D_{u}^{(\top)}\notin \mathcal B_{i}$ for any $i$.
We finish when $|\mathcal B_{s}| = 1$, 
so let $\mathcal B_{s} = \{B_{0}\}$.

Thus, for any $B'\in \mathcal B_{\mathcal P}$ 
we have 
$B'+p_1+p_2+\dots+p_{s} = B_{0}\in\mathcal B_{\mathcal P}$.
By the definition of $\mathcal B_{\mathcal P}$
there exists a path $p_{s+1}\in\mathcal P$ such that 
$B_{0} + p_{s+1} = B$. 
It remains to put 
$p= p_1+p_2+\dots+p_s+p_{s+1}$.
\end{proof}

\begin{lem}\label{LEMFindLargerInB0}
Suppose $B\in \mathcal B_{\mathcal T}$.
Then there exists 
$B'\in \mathcal B_{\mathcal P}$ such that 
$B'\supseteq B$, 
or
the tuple
$(R,D^{(\top)},u,\mathcal B_{\mathcal P})$ is strong.
\end{lem}

\begin{proof}
By the definition of $\mathcal B_{\mathcal T}$
there exist $B_{0}\in\mathcal B_{\mathcal P}$
and a tree $t_{1}\in\mathcal T$
such that $B_0+t_1 = B$.
Let $M\in\mathcal B_{\mathcal T}$ be chosen inclusion maximal.
Choose $t_{0}, t_2\in\mathcal T$ such that 
$M+t_0 = B_0$ and $B+t_2 = M$.
Put $t = t_0+t_1+t_2$. Then $M+t = M$.

Consider the minimal set $L$ of leaves of $t$ we need to restrict to $M$ to obtain a $\zv$-parameterized unary relation in the root
that is different from $D_{u}^{(\top)}$.
Since $M$ is inclusion maximal, 
this unary relation must be $M$.
Let $L = \{u_1,\dots,u_{\ell}\}$ and $u_{0}$ be the root of $t$.
We consider two cases:

Case 1. $\ell>1$. Let
$(t\wedge u_{2}\in C\wedge\dots \wedge u_{\ell}\in C)(u_1,u_0)$
define a $\zv$-parameterized binary relation $S$
and
$t(u_1,u_0)$ define a $\zv$-parameterized binary relation $W$.
By the minimality of $\ell$ 
we have $\prescript{\zv}{}{D_{u}^{(\top)}} + 
\prescript{\zv}{}{S}=\prescript{\zv}{}{D_{u}^{(\top)}}$ for any $\zv$.
Then $W$, $S$, and $B$ witness that 
$(R,D^{(\top)},u,\mathcal B_{\mathcal P})$ is a strong tuple.

Case 2. $\ell=1$.
Let $p$ be the path in $t$ from the leaf $u_1$ to the root $u_{0}$.
Then 
$p = p_0+p_1+p_2$, 
where $p_0$, $p_1$, and $p_2$ are 
parts of $p$ coming from $t_0$, $t_1$, and $t_2$, respectively.
Notice that 
$M+p=M$, $M+p_0\supseteq B_0$,
$M+p_0+p_1\supseteq B$, and $M+p_0+p_1\neq D_{u}^{(\top)}$.
Hence, $B_0+p_1\subseteq M+p_0+p_1\neq D_{u}^{(\top)}$
and $B_0+p_1\supseteq B$.
It remains to put $B' = B_0+p_1$.
\end{proof}

\begin{lem}\label{LEMNonemptyIntersectionInBs}
Suppose 
$B_1\in\mathcal B_{\mathcal P}$ and 
$B_2\in\mathcal B_{\mathcal T}$.
Then $B_{1}\cap B_{2}\neq \varnothing$,
or
the tuple
$(R,D^{(\top)},u,\mathcal B_{\mathcal P})$ is strong.
\end{lem}
\begin{proof}
By Lemma \ref{LEMUniversalTree}
there exists a path 
$p\in \mathcal P$ such that 
$B+p = B_1$ for every $B\in\mathcal B_{\mathcal P}$.
Let $B_{0}= B_2+p'$, where $p'$ is obtained from $p$ 
by switching ends (we could also write $B_{0} = B_2-p$).
If $B_{0}=D_{u}^{(\top)}$ 
then $B_1+p=B_1$ implies 
$B_{1}\cap B_2\neq \varnothing$.
Otherwise, $B_{0}\in\mathcal B_{\mathcal T}$ 
and by Lemma \ref{LEMFindLargerInB0} there exists 
$B_{0}'\in\mathcal B_{\mathcal P}$ such that 
$B_{0}'\supseteq B_0$.
By the definition 
of $p$ we must have 
$B_{0}'+p = B_1$.
Therefore, 
$B_{2}\subseteq B_2-p+p
= B_{0}+p\subseteq B_{0}'+p = B_{1}$,
which completes the proof.
\end{proof}

We are ready to prove the main theorem of this subsection.

\begin{THMFindSmallerReductionTHM}
Suppose $R\subseteq A^{2n+1}$, 
$D^{(\top)}$ is a $\zv$-parameterized
universal 1-consistent reduction for 
$\mathcal I_{R}$, 
$u\in \Var(\mathcal I_{R})$,
$B\us D_{u}^{(\top)}$ is 
    a $\zv$-parameterized nonempty unary relation. Then one of the following conditions holds:
\begin{enumerate}
    \item there exists a $\zv$-parameterized
universal 1-consistent reduction $D^{(\bot)}$ for 
$\mathcal I_{R}$ that is smaller than $D^{(\top)}$;
    \item  there exists a mighty tuple IV q-definable from $R$.
\end{enumerate}
\end{THMFindSmallerReductionTHM}


\begin{proof}
First, we repeat assumptions from the beginning of this section.
We choose the highest variable $u$ with the same property, 
then we define $\mathcal B$ and choose $\mathcal B_{\mathcal T}$ and $\mathcal B_{\mathcal P}$.

Choose some $B\in\mathcal B_{\mathcal P}$.
If $\prescript{\zv}{}B+\prescript{\zv}{}t \neq \varnothing$
for every $\zv$
and $t\in\mathcal T$, then 
by Lemma \ref{LEMIntersectEachImpliesSmallerReduction}
there exists a required smaller reduction.
Otherwise, choose a tree $t\in\mathcal T$ with the minimal number of leaves such that 
$\prescript{\zv_0}{}B+\prescript{\zv_0}{}t = \varnothing$
for some $\zv_0$.
Moving the root to one of the leaves and removing its leaf mark we get another tree $t'$
such that 
$\prescript{\zv_0}{}B+\prescript{\zv_0}{}t' \cap \prescript{\zv_0}{}B = \varnothing$.
Since $t$ has the minimal number of leaves, 
$\prescript{\zv}{}B+\prescript{\zv}{}t'\neq\varnothing$
for any $\zv$ and $\prescript{\zv}{}B+\prescript{\zv}{}t'\in \mathcal B_{\mathcal T}$.
Then Lemma \ref{LEMNonemptyIntersectionInBs} implies that 
the tuple
$(R,D^{(\top)},u,\mathcal B_{\mathcal P})$ is strong.
It remains to apply Lemma \ref{LEMSmallerReductionOrPSpaceHardness}.
\end{proof}

\section{Hardness Claims}\label{SECTIONHardnessResults}
\subsection{Definitions}

Binary relations in this section 
are often viewed as directed graphs,
and we use terminology from the graph theory such as paths and cycles.
A binary relation $R$ is called \emph{transitive} 
if $R + R = R$.
Even though the general domain of any relation $R$ is $A$, 
we often define a subset $D$ such that 
$R\subseteq D^{2}$. 
Then we call the relation $R$ \emph{reflexive} if 
$\{(d,d)\mid d\in D\}\subseteq R$.
Suppose $R$ is a reflexive relation on $D$.
Then \emph{the transitive symmetric closure} of $R$
is the minimal transitive symmetric relation $R'\supseteq R$.
Notice that 
$R' = R-R+R-R+\dots+R-R+R$, for sufficiently many pluses and minuses, 
hence $R'$ is q-definable over $R$.

For a positive integer $m$ and a binary relation $S$ 
denote $m\cdot S =\underbrace{S+S+\dots+S}_{m}$.
\begin{lem}\label{LEMFactorialRepetition} 
Suppose 
$R\subseteq A\times A$, 
$S = (|A|!\cdot |A|^{2})\cdot R$.
Then 
$S + S = S$.
\end{lem}

\begin{proof}
First, put $S_{1} = (|A|!)\cdot R$.
Notice that if $(a,b) \in S_1$, then there is a path from 
$a$ to $b$ in $R$ of length $|A|!$.
Since the domain is of size $|A|$, there must be a cycle of length
$m\le |A|$ in the path. Repeating this cycle 
$|A|!/m$ times we make a path of length $2|A|!$, which
implies $S_1+S_1\supseteq S_1$.
Put $S_{n} = n\cdot S_{1}$.
Since $S_1+S_1\supseteq S_1$, 
we have $S_{i}\subseteq S_{i+1}$.
Moreover, if 
$S_{i}= S_{i+1}$, then $S_{j} = S_{i}$ for 
any $j>i$. Thus, the sequence 
$S_{1},S_{2},\dots,$ stabilises at 
some $S_{i}$, where $i\le |A|^{2}$.
Hence, $S_{|A|^{2}} + S_{|A|^{2}}  = S_{|A|^{2}}$, 
which completes the proof.
\end{proof}

For two equivalence relations $R_{1}$ and $R_{2}$ on some set $D$
by $R_{1}\join R_{2}$ we denote the minimal equivalence relation 
on $D$ containing $R_{1}$ and $R_{2}$.
Thus, it is the usual join of two equivalence relations, but we 
prefer to use this symbol to distinguish it from the disjunction.
Notice that 
$R_{1}\join R_{2}$ is q-definable from $R_{1}$ and $R_{2}$ as 
we can always write a quantified formula 
defining 
$R_{1}+R_{2}+R_{1}+R_{2}+\dots+R_{1}+R_{2}$.

Recall that we agreed that 
$A = \{1,\dots,|A|\}$.
Then we put 
$\kappa = (1,\dots,|A|)$, that is, 
$\kappa$ is a concrete tuple of length $|A|$ with all the elements of 
$A$.



\subsection{PSpace-hardness for a mighty tuple I}
\label{SUBSECTIONPSPACEHARDNESS}
In this section we show 
that the QCSP over a mighty tuple I is PSpace-hard.

For technical reasons 
we will need mighty tuples with an additional property:
\begin{itemize}
    \item[($\kappa$)] $\prescript{\zv}{\delta}R^{\kappa}\subseteq
    \prescript{\zv}{\delta}R^{\alpha}$ for
    every  $\zv\in A^{|A|}$, $\delta\in \prescript{\zv}{}\Delta $, and $\alpha$. 
\end{itemize}

A mighty tuple I satisfying property ($\kappa$) is called \emph{a mighty tuple I'}.

\begin{lem}\label{LEMMightyTupleOnePrime}
Suppose 
$(Q,D,B,C,\Delta)$ is a mighty tuple I. 
Then $\{Q,D,B,C,\Delta\}$ q-defines 
a mighty tuple I'.
\end{lem}

\begin{proof}
Define 
a mighty tuple I' as follows.
Suppose the $\alpha$-parameter is from $A^{k}$.
Put 
$$\prescript{\zv}{\delta}R^{x_1,\dots,x_{|A|}}(y_1,y_2)=
\bigwedge\limits_{i_1,\dots,i_k\in\{1,2,\dots,|A|\}.
}
\prescript{\zv}{\delta}Q^{x_{i_1},\dots,x_{i_k}}(y_1,y_2).$$
Then 
$R^{\kappa} =R^{\forall\forall} = Q^{\forall\forall}$ and 
$R^{\forall} = Q^{\forall}$.
Hence $(R,D,B,C,\Delta)$ is a mighty tuple I'.
\end{proof}

As we mentioned in Section \ref{SectionWithSimpleHardnessReduction} 
we have one reduction covering all the PSpace-hard cases of 
the QCSP. Precisely, we will show that using a mighty tuple we can build relations very similar to 
the relations from Section \ref{SectionWithSimpleHardnessReduction}. Then we use 
the same reduction from 
the Quantified-3-DNF, which is 
the complement of the Quantified-3-CNF.

 
First, for any instance of the Quantified-3-DNF 
 we define a sentence corresponding to this reduction.
Suppose we have two 
relational symbols $\Upsilon_0$ and $\Upsilon_1$ of arity $m+2$
for some $m\ge 1$.
Then 
$\Upsilon_0$ and $\Upsilon_1$ can be viewed as $\xv$-parameterized binary relations, where $\xv\in A^{m}$.
Let $\Phi$ be an instance of the Quantified-3-DNF of the form
 $$Q_{1} x_{1} Q_{2} x_2\dots Q_{n}x_{n}\;
 (x_{a_1}=a_{1}'\wedge x_{b_1}=b_{1}'\wedge x_{c_1}=c_{1}')\vee\dots\vee(x_{a_s}= a_{s}'\wedge x_{b_s}=b_{s}'\wedge x_{c_s}=c_{s}'),$$
 where 
 $Q_{1},Q_{2},\dots,Q_{n}\in\{\forall,\exists\}$,
$a_{i},b_{i},c_{i}\in[n]$ and $a_{i}',b_{i}',c_{i}'\in\{0,1\}$ for every $i\in[n]$.

We define recursively formulas 
$\Psi_{n},\Psi_{n-1},\dots,\Psi_{1},\Psi_{0}$.
Put
\begin{align*}
    \Psi_{n} = \exists y_{1}\dots\exists y_{s-1} \bigwedge_{1 \le i \le s}
   (\Upsilon_{\overline{a_{i}'}}^{\xv_{a_i}}(y_{i-1},y_i)\wedge 
    \Upsilon_{\overline{b_{i}'}}^{\xv_{b_i}}(y_{i-1},y_i) \wedge \Upsilon_{\overline {c_{i}'}}^{\xv_{c_i}}(y_{i-1},y_i)), 
\end{align*}
where $\overline a$ is a negation of $a$ for any $a\in\{0,1\}$.
Notice that for any variable $x_{i}$ of the original instance we introduce 
a variable $\xv_{i}$, which takes values from $A^{m}$.

For every $i$ by $l_{i}$ and $r_{i}$ 
we denote the minimal and the maximal indices of $y$-variables appearing in 
the formula 
$\Psi_{i}$.
Thus, we have $l_{n} = 0$ and $r_{n} = s$.
Let us show how to define 
$\Psi_{k-1}$ from $\Psi_{k}$.
If $Q_{k}=\forall$, then we put 
$\Psi_{k-1} = 
\forall \xv_{k}\; \Psi_{k}$.

If $Q_{k}=\exists$, then we put 
$$\Psi_{k-1}=
\exists y_{r_{k}} \forall \xv_{k} \exists y_{l_{k}} 
 \; \Psi_{k}
\wedge \Upsilon_{0}^{\xv_{k}}(y_{l_{k}-1},y_{l_{k}}) \wedge \Upsilon_{1}^{\xv_{k}}(y_{r_{k}+1},y_{l_{k}}).$$


Notice that in the formula $\Psi_{0}$ all the variables except for 
$y_{l_0}$  and $y_{r_0}$ are quantified.
By $\mathcal Q^{\Phi}$ 
we 
denote the formula $\Psi_{0}$  
whose variables 
$y_{l_0}$  and $y_{r_0}$ are replaced by 
$y$ and $y'$ respectively.

 For $\xv$-parameterized relations $R_{0}$, $R_{1}$,
 and an instance $\Phi$ of the Quantified-3-DNF by 
 $\mathcal Q^{\Phi}(R_{0},R_{1})$ we denote the formula 
 obtained from  
 $\mathcal Q^{\Phi}$ by substituting 
 $R_{0}$ for $\Upsilon_{0}$ 
 and $R_{1}$ for $\Upsilon_{1}$.
 By 
 $\mathcal T^{\Phi}(R_{0},R_{1})$ 
 we denote the transitive symmetric closure of $\sigma$, where
 $\sigma(y,y') = \mathcal  Q^{\Phi}(R_{0},R_{1})$.



Arguing as in Section \ref{SectionWithSimpleHardnessReduction} we can prove the following lemma.
\begin{lem}\label{LEMCanonicalPSpaceHardness}
Suppose 
$\Phi$ is an instance of the Quantified-3-DNF,
$A = \{+,-,0,1\}$,
$V_{0}^{x}(y_1,y_2) = (y_1,y_2\in\{+,-\})\wedge (x=0\rightarrow y_1=y_2)$,
$V_{1}^{x}(y_1,y_2) = (y_1,y_2\in\{+,-\})\wedge (x=1\rightarrow  y_1=y_2)$.
Then 
$\mathcal T^{\Phi}(V_{0},V_{1}) = \{(+,+),(-,-)\}$ if $\Phi$ does not hold;
$\mathcal T^{\Phi}(V_{0},V_{1}) = \{+,-\}^{2}$ if $\Phi$ holds.
\end{lem}

The next lemmas describe important properties of the 
operator $\mathcal T^{\Phi}$.

\begin{lem}\label{LEMNoInstanceCondition}
Suppose 
\begin{enumerate}
	\item $\Phi$ is a No-instance of the Quantified-3-DNF;
	    \item $R_{0}$ and $R_{1}$ are $\xv$-parameterized equivalence relations on $D$;
    \item $(B\times C)\cap (R_{0}^{\beta_{0}}\join R_{1}^{\beta_{1}})=\varnothing$
    for some $\beta_{0}$ and $\beta_1$.
\end{enumerate}
    Then $(B\times C)\cap \mathcal T^{\Phi}(R_0,R_1)=\varnothing$.
\end{lem}

\begin{proof}
Let $\delta=R_{0}^{\beta_{0}}\join R_{1}^{\beta_{1}}$.
Let $B'$ be the union of all classes of $\delta$ having a nonempty intersection with $B$.
Let $C' = D\setminus B'$,
$\delta' = B'^{2}\cup C'^{2}$,
$L_{0}^{\xv} = 
\begin{cases}
\delta' & \text{ if $\xv=\beta_{0}$}\\
D\times D& \text{ otherwise}
\end{cases}$, $L_{1}^{\xv} = 
\begin{cases}
\delta'
& \text{ if $\xv=\beta_{1}$}\\
D\times D& \text{ otherwise}
\end{cases}$.

Notice that $L_{0}\supseteq R_{0}$ and $L_{1}\supseteq R_{1}$,
hence replacement of $R_{0}$ by $L_{0}$ and $R_{1}$ by $L_1$ would make it even harder for 
the UP to win. 
Also, we may assume that the UP only plays $\beta_{0}$ and $\beta_{1}$.
Interpreting $\beta_{0}$ as $0$ and $\beta_{1}$ as 1, 
and interpreting the domain $D/\delta'$ as $\{+,-\}$ 
we derive that 
$\{(+,-)\} \cap\mathcal T^{\Phi}(V_0,V_{1})=\varnothing$ if and only if
$(B\times C)\cap \mathcal T^{\Phi}(L_{0},L_{1})=\varnothing$,
where $V_{0}$ and $V_{1}$ are the canonical relations from
Lemma \ref{LEMCanonicalPSpaceHardness}.
This implies 
$(B\times C)\cap \mathcal T^{\Phi}(R_{0},R_{1})=\varnothing$.
\end{proof}

\begin{lem}\label{YesInstanceCondition}
Suppose 
\begin{enumerate}
	\item $\Phi$ is a Yes-instance of the Quantified-3-DNF;
 	    \item $R_{0}$ and $R_{1}$ are $\xv$-parameterized equivalence relations on $D$;
   \item $(b,c)\in  R_{0}^{\alpha}$ or
$(b,c)\in R_{1}^{\alpha}$ for every $\alpha$.
\end{enumerate}
    Then 
    $(b,c)\in \mathcal T^{\Phi}(R_0,R_1)$.
\end{lem}
 
\begin{proof}
Notice that if $b=c$, then there is an obvious winning strategy for the EP where
she always plays the element $b$.
Thus, we assume that $b\neq c$.
To make it harder for the EP to win we let her play only elements $b$ and $c$.
That is, we replace
the relation $R_{0}^{\xv}$ by   
the relation
$L_{0}^{\xv} =
 R_{0}^{\xv}\cap \{b,c\}^{2}$
and the relation $R_{1}^{\xv}$
by $L_{1}^{\xv} =
 R_{1}^{\xv}\cap \{b,c\}^{2}$.
By condition 3
for any $\xv$ one of the two relations $L_{0}^{\xv}$ and $L_{1}^{\xv}$
is equal to 
$\{b,c\}^{2}$ and another is either $\{(b,b),(c,c)\}$, or $\{b,c\}^{2}$.

Since $\Phi$ is a Yes-instance, 
$(+,-)\in \mathcal T^{\Phi}(V_{0},V_{1})$ 
for the canonical relations $V_{0}$ and $V_{1}$ in  
 Lemma \ref{LEMCanonicalPSpaceHardness}.
Interpreting $b$ and $c$ as $-$ and $+$
we can derive that 
$(b,c)\in \mathcal T^{\Phi}(L_{0},L_{1})$. 
In fact, for any choice of $\xv$ 
either $L_{0}^{\xv}$ connects $b$ and $c$, 
or $L_{1}^{\xv}$ connects $b$ and $c$, 
or both connect.
Hence, if the UP cannot win
in $\mathcal T^{\Phi}(V_{0},V_{1})$, then he cannot win in $\mathcal T^{\Phi}(L_{0},L_{1})$.
This implies that
$(b,c)\in\mathcal T^{\Phi}(R_{0},R_{1})$.
\end{proof}


We will need parameterized relations having arbitrary many parameters. 
Formally, we say that $S$ is 
\emph{a multi-parameter equivalence relation} if
it assigns an equivalence relation $S^{\alpha_1,\dots,\alpha_{n}}$  to every 
sequence
$\alpha_1,\dots,\alpha_{n}\in A^{m}$
and satisfies the following properties:
\begin{enumerate}
\item[(s)] $S^{\alpha_1,\dots,\alpha_{n_1}}=S^{\beta_1,\dots,\beta_{n_2}}$ whenever 
$\{\alpha_1,\dots,\alpha_{n_1}\}=\{\beta_1,\dots,\beta_{n_2}\}$
\item[(m)] $S^{\alpha_1,\dots,\alpha_{n}}\subseteq 
S^{\alpha_1,\dots,\alpha_{n},\alpha_{n+1}}$ for any
$\alpha_1,\dots,\alpha_{n},\alpha_{n+1}\in A^{m}$
\end{enumerate}

Since the set $A^{m}$ is finite, (s) implies 
that we may think of $S$ as a relation of an arity 
$N:= m\cdot |A|^{m}+2$ such that 
$S^{\alpha_1,\dots,\alpha_{N}}$ depends 
only on the set $\{\alpha_1,\dots,\alpha_{N}\}$.
Thus, $S$ is still a finite relation of a fixed arity, 
but it will be convenient for us to assume that 
it can have arbitrary many parameters.
We say that a multi-parameter equivalence relation $S_{1}$ \emph{is larger than}
a multi-parameter equivalence relation $S_{2}$ if 
$S_{1}^{\alpha_{1},\dots,\alpha_{n}}\supseteq S_{2}^{\alpha_{1},\dots,\alpha_{n}}$
for every 
$\alpha_{1},\dots,\alpha_{n}\in A^{m}$.
If additionally $S_1\neq S_2$, we say that 
$S_{1}$ \emph{is strictly larger than} $S_{2}$.

We extend $\mathcal T^{\Phi}$ 
to multi-parameter equivalence relations.
For a multi-parameter equivalence relation $S$ and a parameterized equivalence relation $R$ 
by $\mathcal T^{\Phi}(S,R)$ we denote the multi-parameter equivalence relation
$S_{0}$  defined as follows.
To define $S_{0}^{\uv_1,\dots,\uv_{n}}$ we  
take the formula 
 $\mathcal Q^{\Phi}$,
 replace each $\Upsilon_{0}^{\xv_{i}}$ by 
$S^{\uv_1,\dots,\uv_n,\xv_{i}}$, replace 
each
$\Upsilon_{1}^{\xv_{i}}$ by 
$S^{\uv_1,\dots,\uv_n}\join R^{\xv_{i}}$.
For fixed $\uv_1,\dots,\uv_n$ the obtained formula has only two free variables $y$ and $y'$ and defines 
a binary relation $\sigma$.
Then $S_{0}^{\uv_1,\dots,\uv_{n}}$ is the transitive symmetric closure of
$\sigma$.




\begin{lem}\label{NotReduceUniversalRelation}
Suppose 
\begin{enumerate}
    \item $S$ is a multi-parameter relation on a set $D$;
    \item $R$ is an $\xv$-parameterized equivalence relation on $D$;
    \item $\Phi$ is an instance of the Quantified-3-DNF.
\end{enumerate}
Then $\mathcal T^{\Phi}(S,R)$ is a multi-parameter equivalence relation that is larger than $S$.
\end{lem}

\begin{proof}
Suppose 
$\mathcal T^{\Phi}(S,R) = S_0$.
Since $S$ satisfies properties (s) and (m), it immediately follows from the 
definition that $S_0$ also satisfies properties (s) and (m). 
Let us show that $S_{0}$ is larger than $S$.
For any 
$\uv_1,\dots,\uv_{n}$ 
we have 
$S^{\uv_1,\dots,\uv_{n}}\subseteq 
S^{\uv_1,\dots,\uv_{n},\xv_{i}}$
and
$S^{\uv_1,\dots,\uv_{n}}\subseteq 
S^{\uv_1,\dots,\uv_{n}}\join
R^{\xv_i}$.
Hence, 
the interpretations of both 
$\Upsilon_{0}^{\xv_i}$ and $\Upsilon_{1}^{\xv_i}$
contain every pair $(b,c)$
from  
$S^{\uv_1,\dots,\uv_{n}}$. 
Thus,  for any play of the UP
the EP can always play $b$ 
to confirm that 
$(b,c)\in S_{0}^{\uv_1,\dots,\uv_{n}}$.
\end{proof}

\begin{lem}\label{IncreaseUniversalRelation}
Suppose 
\begin{enumerate}
    \item[(1)] $S$ is a multi-parameter equivalence relation on a set $D$;
    \item[(2)] $R$ is an $\xv$-parameterized equivalence relation on $D$;
    \item[(3)] $(b,c)\in S^{\alpha}\join R^{\alpha}$ for every $\alpha\in A^{m}$;
    \item[(4)] $(b,c)\notin S^{\alpha}$ for some $\alpha\in A^{m}$;
    \item[(5)] $\Phi$ is a Yes-instance  of the Quantitied-3-DNF.
\end{enumerate}
Then $\mathcal T^{\Phi}(S,R)$ is strictly larger than $S$.
\end{lem}

\begin{proof}
Suppose 
$\mathcal T^{\Phi}(S,R) = S_0$.
Consider a maximal set of 
tuples 
$\alpha_1,\dots,\alpha_{n}$ such 
that $(b,c)\notin S^{\alpha_1,\dots,\alpha_{n}}$.
By condition (4) such a set exists.
Note that this set may contain all 
tuples.
Let us show that 
$(b,c)\in S_0^{\alpha_1,\dots,\alpha_{n}}$, 
which together with 
Lemma \ref{NotReduceUniversalRelation} would mean 
that 
$\mathcal T^{\Phi}(S,R)$ is strictly larger than $S$.

Let us consider the interpretations of
$\Upsilon_{0}^{\xv_i}$ and $\Upsilon_{1}^{\xv_i}$
in the definition of $\mathcal T^{\Phi}(S,R)$.
If $\xv_i$ is not 
from the set $\{\alpha_{1},\dots,\alpha_{n}\}$,
then by the maximality of the set 
the relation 
$\Upsilon_{0}^{\xv_i} = S^{\alpha_{1},\dots,\alpha_{n},\xv_i}$
contains $(b,c)$.
If $\xv_i$ is 
from the set $\{\alpha_{1},\dots,\alpha_{n}\}$,
then 
$\Upsilon_{1}^{\xv_i} = S^{\alpha_{1},\dots,\alpha_{n}}
\join R^{\xv_i}\supseteq S^{\xv_i}\join R^{\xv_i}$,
which contains $(b,c)$ by condition (3).
Hence, 
by Lemma \ref{YesInstanceCondition}
$(b,c)\in S_{0}^{\alpha_{1},\dots,\alpha_{n}}$.
\end{proof}

\begin{lem}\label{LEMNoInstanceImpliesNoBC}
Suppose 
\begin{enumerate}
    \item[(1)] $S$ is a multi-parameter equivalence relation on a set $D$;
    \item[(2)] $R$ is an $\xv$-parameterized equivalence relation on $D$;
    \item[(3)] $(B\times C)\cap S^{\beta}=\varnothing$ for some $\beta\in A^{m}$;
    \item[(4)] there exists $\alpha$ such that 
$R^{\alpha}\subseteq S^{\xv}$ for every $\xv$;
    \item[(5)] $\Phi$ is a No-instance  of the Quantitied-3-DNF.
\end{enumerate}
Then $(B\times C)\cap S_{0}^{\beta}=\varnothing$,
where $S_{0} = \mathcal T^{\Phi}(S,R)$.
\end{lem}

\begin{proof}
Recall that $S_{0}^{\beta}$ is defined using the formula
$\mathcal Q^{\Phi}$, where we substitute  
$S^{\beta,\xv_{i}}$
for each $\Upsilon_{0}^{\xv_{i}}$ and 
$S^{\beta}\join R^{\xv_{i}}$
for each 
$\Upsilon_{1}^{\xv_{i}}$.
We derive from 
(3) and (4) that
$
(B\times C)\cap (S^{\beta,\beta}\join (S^{\beta}\join R^{\alpha}))=\varnothing$.
Then Lemma \ref{LEMNoInstanceCondition}
implies that
$(B\times C)\cap S_{0}^{\beta}=\varnothing$.
\end{proof}

\begin{lem}\label{LEMMightyTupleImplies}
Suppose $(R,D,B,C,\Delta)$ is a mighty tuple I'. Then there exist $(\zv,\delta,\xv)$-parameterized equivalence relations $R_0$ and $R_{1}$ on $D$  q-definable from 
$\{R,D,B,C,\Delta\}$ and satisfying the following conditions: 
\begin{enumerate}
\item[(1)] $\forall \zv\in A^{|A|} \exists \delta\in \prescript{\zv}{}\Delta  \; \forall \xv \;
(\prescript{\zv}{\delta}B\times \prescript{\zv}{\delta}C\subseteq \prescript{\zv}{\delta}R_{0}^{\xv}\join \prescript{\zv}{\delta}R_{1}^{\xv})$;
\item[(2)] 
$\exists \zv\in A^{|A|} \forall \delta \in \prescript{\zv}{}\Delta \;\exists \xv
 ((\prescript{\zv}{\delta}B\times \prescript{\zv}{\delta}C)\cap \prescript{\zv}{\delta}R_{0}^{\xv}=\varnothing)$;
 \item[(3)] there exists $\alpha$ such that 
$\prescript{\zv}{\delta}R_{1}^{\alpha}\subseteq
\prescript{\zv}{\delta}R_{0}^{\xv}$
for every 
$\zv\in A^{|A|}$, $\delta\in\prescript{\zv}{}\Delta $, and $\xv$.
\end{enumerate}
\end{lem}

\begin{proof}
Let 
$\sigma_1,\dots,\sigma_{N}$ 
be the set of all injective mappings from 
$\{1,2,\dots,|A|\}$ to $\{1,2,\dots,|A|^{2}\}$.
Let 
$$\prescript{\zv}{\delta}U_{n}^{x_1,\dots,x_{|A|^{2}}} =
\prescript{\zv}{\delta}R^{x_{\sigma_{1}(1)},\dots,x_{\sigma_{1}(|A|)}}
\join
\prescript{\zv}{\delta}R^{x_{\sigma_{2}(1)},\dots,x_{\sigma_{2}(|A|)}}
\join\dots\join
\prescript{\zv}{\delta}R^{x_{\sigma_{n}(1)},\dots,x_{\sigma_{n}(|A|)}}$$
Since at least $|A|$ elements in the set 
$x_1,\dots,x_{|A|^{2}}$ are equal, 
there exists $i\in\{1,2,\dots,N\}$ 
such that 
$x_{\sigma_{i}(1)}=x_{\sigma_{i}(2)}=\dots=x_{\sigma_{i}(|A|)}$.
Since
$\prescript{\zv}{\delta}R^{\forall}=\prescript{\zv}{\delta}D\times
    \prescript{\zv}{\delta}D$
    for every $\zv\in A^{|A|}$ and $\delta\in \prescript{\zv}{}\Delta $,
the relation
$\prescript{\zv}{\delta}U_{N}^{x_1,\dots,x_{|A|^{2}}}$
is equal to $\prescript{\zv}{\delta}D\times \prescript{\zv}{\delta}D$.

Consider the maximal $n$ such that the following condition holds

$$\exists \zv\in A^{|A|} \forall \delta\in \prescript{\zv}{}\Delta \; \exists \xv_0
 ((\prescript{\zv}{\delta}B\times \prescript{\zv}{\delta}C)\cap (\prescript{\zv}{\delta}U_n^{\xv_0})=\varnothing).$$
Put 
$\prescript{\zv}{\delta}R_{0}^{x_1,\dots,x_{|A|^{2}}}
=\prescript{\zv}{\delta}U_{n}^{x_1,\dots,x_{|A|^{2}}}$
and 
$\prescript{\zv}{\delta}R_{1}^{x_1,\dots,x_{|A|^{2}}}
=\prescript{\zv}{\delta}R^{x_{\sigma_{n+1}(1)},\dots,x_{\sigma_{n+1}(|A|)}}$ and show that they satisfy the required properties.

Property (1) follows from the fact that $n$ was chosen maximal and
the corresponding condition for $\prescript{\zv}{\delta}U_{n+1}^\xv= \prescript{\zv}{\delta}R_{0}^\xv\join \prescript{\zv}{\delta}R_{1}^\xv$ does not hold.
Property (2) again follows from the choice of $n$.
To prove property (3) consider a tuple 
$\alpha=(a_1,\dots,a_{|A|^{2}})$
such that $(a_{\sigma_{n+1}(1)},\dots,a_{\sigma_{n+1}(|A|)})=\kappa$.
Then $\prescript{\zv}{\delta}R_1^{\alpha}=\prescript{\zv}{\delta}R^{\kappa}\subseteq \prescript{\zv}{\delta}R_0^{\xv_0}$
for every 
$\zv\in A^{|A|}$, $\delta\in\prescript{\zv}{}\Delta $, and $\xv_0$.
\end{proof}

\begin{THMMightyTupleIPSpaceHardnessTHM}
Suppose $(Q,D,B,C,\Delta)$ is a mighty tuple I.
Then $\QCSP(\{Q,D,B,C,\Delta\})$ is PSpace-hard.
\end{THMMightyTupleIPSpaceHardnessTHM}

\begin{proof}

By Lemma \ref{LEMMightyTupleOnePrime} 
there exists a mighty tuple I'
$(R,D,B,C,\Delta)$ 
q-definable from the set $\{Q,D,B,C,\Delta\}$. 
By Lemma \ref{LEMMightyTupleImplies}
there exist $R_0$ and $R_{1}$ satisfying the corresponding conditions (1)-(3).
For every $\zv$ and $\delta$ we define  
a multi-parameter equivalence relation $\prescript{\zv}{\delta}S_{0}$ by 
 $$\prescript{\zv}{\delta}S_{0}^{\xv_1,\dots,\xv_k} = 
 \prescript{\zv}{\delta}R_{0}^{\xv_1}\join \dots\join \prescript{\zv}{\delta}R_{0}^{\xv_k}.$$

Using the operator $\mathcal T^{\Phi}$
we will build a sequence of multi-parameter equivalence relations
$\prescript{\zv}{\delta}S_{0},\prescript{\zv}{\delta}S_1,\dots,\prescript{\zv}{\delta}S_N$.
The idea is to reduce an instance $\Phi$ of the Quantified-3-DNF to $\QCSP(\Gamma)$ by 
substituting $S_{N}$ into the formula $\mathcal Q^{\Phi}$.
If this reduction works, then we proved the PSpace-hardness.
If it does not work, we  define a new bigger multi-parameter equivalence relation 
$S_{N+1}$ and continue.
Thus, we want to build 
a sequence $S_{0}, \dots, S_{N}$ maintaining the following properties:

\begin{enumerate}
\item[(s1)] $S_{i}$ is q-definable over $\{R,D,B,C,\Delta\}$;
\item[(s2)] for every $\zv\in A^{|A|}$ and $\delta\in\prescript{\zv}{}\Delta $
the universal relation $\prescript{\zv}{\delta}S_{i+1}$ is larger than 
 $\prescript{\zv}{\delta}S_{i}$;
\item[(s3)] there exist $\zv\in A^{|A|}$ and $\delta\in\prescript{\zv}{}\Delta $ such that
$\prescript{\zv}{\delta}S_{i+1}$ is strictly larger than  
 $\prescript{\zv}{\delta}S_{i}$; 
\item[(s4)] $\forall \zv\in A^{|A|} \exists \delta \in\prescript{\zv}{}\Delta \; \forall \xv \;
(\prescript{\zv}{\delta}B\times \prescript{\zv}{\delta}C\subseteq \prescript{\zv}{\delta}S_{i}^{\xv}\join \prescript{\zv}{\delta}R_1^{\xv})$;
\item[(s5)] 
$\exists \zv\in A^{|A|} \forall \delta\in\prescript{\zv}{}\Delta  \exists \xv
 ((\prescript{\zv}{\delta}B\times \prescript{\zv}{\delta}C)\cap \prescript{\zv}{\delta}S_{i}^{\xv}=\varnothing)$.
\end{enumerate}

Let us check that $S_{0}$ satisfies 
conditions (s1), (s4), and (s5). 
Condition (s1) follows from the definition.
Condition (s4) and (s5) come from (1) and (2) in Lemma \ref{LEMMightyTupleImplies}.

Properties (s2) and (s3) guarantee that the sequence will not be infinite.
Assume that we have a sequence 
$S_{0},S_1,\dots,S_N$.
Let us build $S_{N+1}$ satisfying the above properties or prove the PSpace-hardness using $S_{N}$. 
For every instance $\Phi$ of the Quantified-3-DNF by $\prescript{\zv}{\delta}S_{N+1,\Phi}$
we denote $\mathcal T^{\Phi}(\prescript{\zv}{\delta}S_{N},\prescript{\zv}{\delta}R_{1})$.
%
Consider two cases:

Case 1. There exists a Yes-instance $\Phi$ of 
the Quantified-3-DNF such that $S_{N+1,\Phi}$
satisfies condition (s5).
Put $S_{N+1} = S_{N+1,\Phi}$ and check that each of the properties 
(s1)-(s5) holds.
Property (s1) follows from the definition.
Property (s2) follows from Lemma 
\ref{NotReduceUniversalRelation}.
To prove property (s3) consider $\zv$ from condition (s5) for $S_{N}$, 
and 
the corresponding $\delta$ from condition (s4) for $S_{N}$.
Then 
$\prescript{\zv}{\delta}B\times \prescript{\zv}{\delta}C\subseteq \prescript{\zv}{\delta}S_{N}^{\xv}\join \prescript{\zv}{\delta}R_1^{\xv}$
for every $\xv$ 
and  
$(\prescript{\zv}{\delta}B\times \prescript{\zv}{\delta}C)\cap \prescript{\zv}{\delta}S_{N}^{\xv}=\varnothing$
 for some $\xv$.
Then 
Lemma \ref{IncreaseUniversalRelation} implies 
that 
$\prescript{\zv}{\delta}S_{N+1}$ is strictly larger than 
$\prescript{\zv}{\delta} S_{N}$ which proves 
condition (s3).
Property (s4) follows from 
the fact that $S_{N+1}$ is larger that $S_{N}$ and 
$S_{N}$ satisfies (s4).
Property (s5) is just the definition of Case 1.
Thus, we defined $S_{N+1}$ satisfying the required properties (s1)-(s5).

Case 2. $S_{N+1,\Phi}$ does not satisfy property (s5) for any Yes-instance $\Phi$ of 
the Quantified-3-DNF. Thus, 
for every Yes-instance $\Phi$
we have 
\begin{align}\label{CaseOneDefinition}\forall \zv\in A^{|A|} \exists \delta\in\prescript{\zv}{}\Delta  \; \forall \xv\;
(\prescript{\zv}{\delta}B\times \prescript{\zv}{\delta}C\subseteq \prescript{\zv}{\delta}S_{N+1,\Phi}^{\xv}).\end{align}
Let us show that 
any (Yes- or No-) instance 
$\Phi$ of the Quantified-3-DNF
is equivalent to 
\begin{align}\label{QCSPRecuction}
\forall \zv\in A^{|A|} \exists \delta\in\prescript{\zv}{}\Delta  \; \forall \uv\;
(S_{N+1,\Phi}^{\uv}(y,y')
\wedge y\in \prescript{\zv}{\delta}B \wedge 
y'\in \prescript{\zv}{\delta}C).
\end{align}
Notice that the above formula can be efficiently built 
from the instance $\Phi$, which gives us a polynomial reduction
from the Quantified-3-DNF.
If $\Phi$ is a Yes-instance, then 
it follows from (\ref{CaseOneDefinition}). Suppose $\Phi$ is a No-instance.
Recall that (see condition (3) in Lemma \ref{LEMMightyTupleImplies})
there exists $\alpha$ such that 
$\prescript{\zv}{\delta}R_{1}^{\alpha}\subseteq
\prescript{\zv}{\delta}R_{0}^{\xv}\subseteq \prescript{\zv}{\delta}S_{N}^{\xv}$
for every 
$\zv\in A^{|A|}$, $\delta\in\prescript{\zv}{}\Delta $, and $\xv$.
Combining this with 
property (s5) for $S_{N}$
and using Lemma \ref{LEMNoInstanceImpliesNoBC} 
we obtain that 
$$\exists \zv\in A^{|A|} \forall \delta\in \prescript{\zv}{}\Delta  \;\exists \xv \;
((\prescript{\zv}{\delta}B\times \prescript{\zv}{\delta}C)\cap 
 \prescript{\zv}{\delta}S_{N+1}^{\xv}=\varnothing).$$
Hence, (\ref{QCSPRecuction}) does not hold
and the instance $\Phi$ is equivalent to (\ref{QCSPRecuction}).
Thus we built a reduction from the Quantified-3-DNF
and proved PSpace-hardness of
$\QCSP(\{Q,D,B,C,\Delta\})$.
\end{proof}

\subsection{Mighty tuples II, III, and IV}
\label{SUBSECTIONMightyTuplesTwoThreeFour} 

It this section we show that 
mighty tuples II, III, and IV are equivalent in the sense that 
any of them q-defines any other. 

Below, 
$R$ is always a $(\zv,\alpha)$-parameterized binary relation, 
$D,B$ and $C$ are $\zv$-parameterized unary relation, 
where $\zv\in A^{|A|}$ and $\alpha\in A^{k}$.
The tuple $(R,D,B,C)$ is called a \emph{quadruple} in this section.
We will need the following properties of a quadruple:






\begin{enumerate}
\item[($\kappa$)] $k=|A|$ and $\prescript{\zv}{}R^{\kappa}\subseteq
    \prescript{\zv}{}R^{\alpha}$ for
    every  $\zv\in A^{|A|}$ and $\alpha\in A^{k}$;
\item[($d+$)] $\prescript{\zv}{}D + \prescript{\zv}{}R^{\forall\forall} = \prescript{\zv}{}D$
    for every $\zv\in A^{|A|}$;
    
\item[(un)] $\prescript{\zv}{}B\neq \varnothing$, $\prescript{\zv}{}C\neq \varnothing$, 
    $\prescript{\zv}{}B\subseteq \prescript{\zv}{}D$, $\prescript{\zv}{}C\subseteq \prescript{\zv}{}D$
    for every $\zv\in A^{|A|}$;
\item[($bc$)] $\prescript{\zv}{}R^{\forall}\cap (\prescript{\zv}{}B\times \prescript{\zv}{}C)\neq \varnothing$ for every $\zv\in A^{|A|}$; 
\item[($\varnothing$)] $\prescript{\zv}{}B\cap \prescript{\zv}{}C=\varnothing$ 
    for some $\zv\in A^{|A|}$;

\item[($b+$)] $\prescript{\zv}{}B + \prescript{\zv}{}R^{\forall\forall} = \prescript{\zv}{}B$
    for every $\zv\in A^{|A|}$; 
\item[($+c$)] $\prescript{\zv}{}R^{\forall\forall} +\prescript{\zv}{}C = \prescript{\zv}{}C$
    for every $\zv\in A^{|A|}$;
\item[(t)] $\prescript{\zv}{}R^{\alpha} + \prescript{\zv}{}R^{\alpha} = \prescript{\zv}{}R^{\alpha}$ for every $\zv\in A^{|A|}$ and $\alpha\in A^{k}$;
\item [(sd)] $\proj_{1}(\prescript{\zv}{}R^{\alpha})=\proj_{2}(\prescript{\zv}{}R^{\alpha})=\prescript{\zv}{}D$ for 
    every  $\zv\in A^{|A|}$ and $\alpha\in A^{k}$;
\item[(r)]  $\{(d,d)\mid d\in \prescript{\zv}{} D\}\subseteq \prescript{\zv}{}R^{\alpha}$
    for every $\zv\in A^{|A|}$ and $\alpha\in A^{k}$;

\item[($bd$)] $\prescript{\zv}{}B + \prescript{\zv}{}R^{\forall} = \prescript{\zv}{}D$    
    for every $\zv\in A^{|A|}$;
\item[($cd$)] $\prescript{\zv}{}R^{\forall} +  \prescript{\zv}{}C = \prescript{\zv}{}D$      for every $\zv\in A^{|A|}$; 

\item[($c+$)] $\prescript{\zv}{}C+\prescript{\zv}{}R^{\forall\forall}  = \prescript{\zv}{}C$      for some $\zv\in A^{|A|}$;
\item [(s)]$\prescript{\zv}{}R^{\alpha}$ is 
    symmetric for 
    every  $\zv\in A^{|A|}$ and $\alpha\in A^{k}$.

\end{enumerate}

Notice that 
a mighty tuple II is just a quadruple satisfying 
all the 
above properties except for ($\kappa$),
a mighty tuple III $(R,B,C)$ forms a quadruple 
$(R,A,B,C)$, where $D = A$,
satisfying properties $\{\text{un},bc,\varnothing,b+,+c\}$,
a mighty tuple IV
is a quadruple 
satisfying properties
$\{d+,\text{un},bc,\varnothing,b+,bd\}$.
Let 
$\mathrm{II}$ be the set of all the above properties except
for ($\kappa$),
$\mathrm{III} = \{\text{un},bc,\varnothing,b+,+c\}$,
and 
$\mathrm{IV} = \{d+,\text{un},bc,\varnothing,b+,bd\}$.

Below we prove many claims that allow us to moderate the 
quadruple to satisfy more properties from the above list.
Usually the claims are of the form 
$P_{1}\vdash P_{2}$, where $P_{1}$ and $P_{2}$ are 
some sets of properties of a quadruple, and should be understood 
as follows.
Suppose a quadruple satisfies properties $P_{1}$,
then there exists a quadruple q-definable from the first one
and satisfying properties $P_{2}$.
Also sometimes we add  
``$+\text{reduce } \sum_{\zv\in A^{|A|}}|\prescript{\zv}{}D|$'' 
meaning that the sum $\sum_{\zv\in A^{|A|}}|\prescript{\zv}{}D|$
calculated for the new quadruple is 
smaller than the sum calculated for the old one.
We write \emph{increase} or \emph{keep} instead of reduce 
if the sum is increased or stays the same, respectively.
Most of the properties are from the above list but some of them are 
given by a quantified formula.

First, we want to 
be able to add the additional property ($\kappa$) to existing properties
from III or IV.

\begin{sublem}\label{CLAIMAddKappa}
Suppose $P\subseteq \mathrm{III}\cup \mathrm{IV}$.
Then 
$P\vdash P\cup\{\kappa\}$.
\end{sublem}

\begin{proof}
We change only the $(\zv,\alpha)$-parameterized relation 
$R$. The new relation $R_0$ is defined by
$$\prescript{\zv}{}R_0^{x_1,\dots,x_{|A|}}(y_1,y_2)=
\bigwedge\limits_{i_1,\dots,i_k\in\{1,2,\dots,|A|\}.
}
\prescript{\zv}{}R^{x_{i_1},\dots,x_{i_k}}(y_1,y_2).$$
Then 
$R_0^{\kappa} =R_0^{\forall\forall} = R^{\forall\forall}$ and 
$R_0^{\forall} = R^{\forall}$.
It is straightforward to check that 
the quadruple $(R_0,D,B,C)$ satisfies all 
the properties satisfied by $(R,D,B,C)$, which completes the proof.
\end{proof}

Notice that property ($\kappa$) implies 
that $R^{\kappa} =R^{\forall\forall}$ 
and in the following claims we usually write 
$R^{\kappa}$ instead of $R^{\forall\forall}$.

\begin{sublem}\label{CLAIMFirstTransitive}
$\mathrm{III}\cup\{\kappa\}\vdash \mathrm{III}\cup\{\kappa,t\}$.

\end{sublem}
\begin{proof}
Put
$\prescript{\zv}{}R_{0}^{\beta} = N\cdot \prescript{\zv}{}R^{\beta}$, 
where $N = |A|!\cdot |A|^{2}$.
Note that we have
$\prescript{\zv}{}R_{0}^{\forall}\supseteq N\cdot\prescript{\zv}{}R^{\forall}$.
We claim that $R_{0}$, $B$, $C$, and $D$ satisfy 
properties 
$\{\kappa,\text{un},bc,\varnothing,b+,+c,\text{t}\}$.
For all the properties but ($bc$) and (t) it follows immediately 
from the same properties for $R$.
To prove  ($bc$)
we choose some
$(\prescript{\zv}{}b_{1},\prescript{\zv}{}c)\in \prescript{\zv}{}R^{\forall}\cap (\prescript{\zv}{}B\times \prescript{\zv}{}C)$.
Since ($b+$), we can find a sequence 
$\prescript{\zv}{}b_{N}-\prescript{\zv}{}b_{N-1}-\dots-
\prescript{\zv}{}b_{2}-
\prescript{\zv}{}b_{1}$ 
such that 
each $\prescript{\zv}{}b_{i}$ is from $\prescript{\zv}{}B$ and 
$(\prescript{\zv}{}b_{i+1},\prescript{\zv}{}b_{i})\in \prescript{\zv}{}R^{\kappa}\subseteq \prescript{\zv}{}R^{\forall}$.
Then $(b_{N},c)\in \prescript{\zv}{}R_{0}^{\forall}$, 
which implies ($bc$).
Lemma \ref{LEMFactorialRepetition}
implies property (t).
\end{proof}

\begin{sublem}\label{CLAIMSecondTransitive}
$\mathrm{IV}\cup\{\kappa\}\vdash \mathrm{IV}\cup\{\kappa,t\}$.

\end{sublem}
\begin{proof}
The proof repeats the proof of the previous claim word for word. 
Additional properties ($bd$) and ($d+$) follow from ($bd$) and ($d+$) for $R$.  
\end{proof}

\begin{sublem}\label{CLAIMSecondPlusC}
$\mathrm{IV}\cup\{\kappa,\text{t}\}\vdash \mathrm{IV}\cup\{\kappa,t,+c\}$.

\end{sublem}

\begin{proof}
    Put $\prescript{\zv}{}C_{0} = 
    (\prescript{\zv}{}R^{\kappa} +\prescript{\zv}{}C)\cap 
    \prescript{\zv}{}D$ and $\prescript{\zv}{}R_{0}^{\alpha}=
    \prescript{\zv}{}R^{\alpha}\cap (\prescript{\zv}{}D\times \prescript{\zv}{}D)$.
We claim that 
$R_0$, $B$, $C_{0}$, and $D$ satisfy the required properties.
Restriction of $R$ does not affect any properties as we have 
property (un).
Changing $C$ could affect only properties 
(un) and  ($\varnothing$).
(un) follows from ($d+$) for $R$  and 
property (un) for $C$.
By property (t) 
we have 
$$\prescript{\zv}{}C_{0}=
\prescript{\zv}{}R_{0}^{\kappa} +\prescript{\zv}{}C=
\prescript{\zv}{}R_{0}^{\kappa} +\prescript{\zv}{}R_{0}^{\kappa} +\prescript{\zv}{}C=\prescript{\zv}{}R_{0}^{\kappa} +\prescript{\zv}{}C_{0}
$$ 
hence we have ($+c$).
To prove property ($\varnothing$) 
we use this property for $C$ and 
consider $\zv\in A^{|A|}$ such that 
$\prescript{\zv}{}B\cap \prescript{\zv}{}C=\varnothing$.
Then using property ($b+$) for $R$ we derive
$$\prescript{\zv}{}B\cap \prescript{\zv}{}C=\varnothing
\Rightarrow 
(\prescript{\zv}{}B+\prescript{\zv}{} R^{\kappa})\cap \prescript{\zv}{}C=\varnothing
\Rightarrow 
\prescript{\zv}{}B\cap (\prescript{\zv}{} R^{\kappa}+ \prescript{\zv}{}C)=\varnothing
\Rightarrow 
\prescript{\zv}{}B\cap \prescript{\zv}{}C_{0}=\varnothing.
$$%
\end{proof}

Notice that 
$\mathrm{IV}\cup\{\kappa,\text{t},+c\}=
\mathrm{III}\cup\{\kappa,\text{t},d+,bd\}$.

\begin{sublem}\label{CLAIMAddSDR}
$\mathrm{III}\cup\{\kappa,\text{t}\}\vdash \mathrm{III}\cup\{\kappa,\text{t},d+,\text{r},\text{sd}\}$.
\end{sublem}
\begin{proof}
Put $\prescript{\zv}{}D_{0}(x) = 
\prescript{\zv}{}R^{\kappa}(x,x)$,
$\prescript{\zv}{}B_{0} =
\prescript{\zv}{}B\cap 
\prescript{\zv}{}D_{0}$,
$\prescript{\zv}{}C_{0} =
\prescript{\zv}{}C\cap 
\prescript{\zv}{}D_{0}$,
$\prescript{\zv}{}R_{0}^{\alpha} =
\prescript{\zv}{}R_{0}^{\alpha}\cap 
(\prescript{\zv}{}D_{0}\times \prescript{\zv}{}D_{0})$.

Let us prove that $R_{0}$, $B_{0}$, $C_{0}$, and $D_{0}$ satisfy 
the required properties.
Properties ($\kappa$) 
and $(\varnothing)$ follow from the corresponding properties 
for $R,B,C$, and $D$. 
Properties (r), ($d+$), and (sd) follow immediately from the definition.

Consider a tuple $(b_{0},c_{0})\in \prescript{\zv}{}R^{\forall}$,
which exists by property ($bc$).
By property ($b+$)
we can find a 
path 
$b_{N}-b_{N-1}-\dots-b_{1}-b_0$ of any length $N$ such that each $b_{i}$ is from 
$\prescript{\zv}{}B$ and
$(b_{i+1},b_{i})\in\prescript{\zv}{}R^{\kappa}$.
Similarly, by property ($+c$) we can find
$c_0-c_1-\dots-c_{N}$ of length $N$ such that each $c_{i}$ is from 
$\prescript{\zv}{}C$ and
$(c_{i},c_{i+1})\in\prescript{\zv}{}R^{\kappa}$.
If $N$ is large enough 
then 
both sequences will have repetitive elements.
Let these elements be $b_{i}$ and $c_{j}$.
By property (t), these repetitive elements $b_{i}$ and $c_{j}$ should be from $\prescript{\zv}{}B_{0}$
and $\prescript{\zv}{}C_{0}$,
respectively, which implies (un).
Again (t) for $R$
implies that 
$(b_{i},c_{j})\in\prescript{\zv}{}R^{\forall}$
and therefore 
$(b_{i},c_{j})\in\prescript{\zv}{}R_{0}^{\forall}$,
which confirms ($bc$).
Properties ($b+$) and ($+c$) for $R_{0}$ follows from 
($b+$) and ($+c$) for $R$ and 
(r) for $R_{0}$.

By reflexivity (property (r)) of $\prescript{\zv}{}R_{0}^{\alpha}$ we have
$\prescript{\zv}{}R^{\alpha} + \prescript{\zv}{}R^{\alpha} \supseteq  \prescript{\zv}{}R^{\alpha}$ 
and by property (t) for $R$ we have
$\prescript{\zv}{}R_{0}^{\alpha} + \prescript{\zv}{}R_{0}^{\alpha} \subseteq  \prescript{\zv}{}R_{0}^{\alpha}$.
Thus, we have property (t) for $R_{0}$.
\end{proof}

Denote 
$J = 
\{\kappa,d+,\text{un},bc,\varnothing,b+,+c,\text{t},\text{r},\text{sd}\}$.
Then 
$J= 
\mathrm{III} \cup \{\kappa,d+,\text{t},\text{r},\text{sd}\}=
\mathrm{IV} \cup \{\kappa,\text{t},+c,\text{r},\text{sd}\}\setminus \{bd\}
=\mathrm{II}\cup\{\kappa\}\setminus \{bd,cd,c+,\text{s}\}$.

\begin{sublem}\label{reduceDBright}
$J\cup\{\neg bd\}\vdash 
J\cup \{bd\} +\text{reduce } \sum_{\zv\in A^{|A|}}|\prescript{\zv}{}D|$.
\end{sublem}

\begin{proof}
Put $\prescript{\zv}{}D_{0} =
\prescript{\zv}{}B + 
\prescript{\zv}{}R^{\forall}$,
$\prescript{\zv}{}B_{0} =
\prescript{\zv}{}B\cap 
\prescript{\zv}{}D_{0}=\prescript{\zv}{}B$,
$\prescript{\zv}{}C_{0} =
\prescript{\zv}{}C\cap 
\prescript{\zv}{}D_{0}$,
$\prescript{\zv}{}R_{0}^{\alpha} =
\prescript{\zv}{}R^{\alpha}\cap 
(\prescript{\zv}{}D_{0}\times \prescript{\zv}{}D_{0})$.
Let us prove that 
$R_{0}$, $B_{0}$, $C_{0}$, $D_{0}$ satisfy the required properties.
Properties 
($\kappa$), ($d+$), ($\varnothing$), ($b+$), ($+c$),
and (r) follow immediately from the definition 
and the corresponding properties for $R$.
Property (t) follows from (t) and (r) for $R$.
By property (r) for $R$ we have $\prescript{\zv}{}B_{0}=\prescript{\zv}{}B$.
Hence, properties ($bd$) and ($bc$) follow from the definition.
Property (un) follows from (un) and ($bc$) for $R$.
\end{proof}

\begin{sublem}\label{reduceDCleft}
$J\cup\{\neg cd\}\vdash 
J\cup \{cd\} +\text{reduce } \sum_{\zv\in A^{|A|}}|\prescript{\zv}{}D|$.
\end{sublem}

\begin{proof}
Put $\prescript{\zv}{}D_{0} = 
\prescript{\zv}{}R^{\forall} + \prescript{\zv}{}C $,
$\prescript{\zv}{}B_{0} =
\prescript{\zv}{}B\cap 
\prescript{\zv}{}D_{0}$,
$\prescript{\zv}{}C_{0} =
\prescript{\zv}{}C\cap 
\prescript{\zv}{}D_{0}=\prescript{\zv}{}C$,
$\prescript{\zv}{}R_{0}^{\alpha} =
\prescript{\zv}{}R^{\alpha}\cap 
(\prescript{\zv}{}D_{0}\times \prescript{\zv}{}D_{0})$
and repeat the proof of the previous claim switching $B$ and $C$.
\end{proof}

\begin{sublem}\label{reduceDCright}
$J\cup\{bd,cd,\forall \zv (\prescript{\zv}{}C+\prescript{\zv}{}R^{\kappa})\cap \prescript{\zv}{}B\neq \varnothing,
\exists \zv \prescript{\zv}{}C+\prescript{\zv}{}R^{\kappa}\neq \prescript{\zv}{}D\}
\vdash
J 
+\text{reduce } \sum_{\zv\in A^{|A|}}|\prescript{\zv}{}D|$.

\end{sublem}

\begin{proof}
 Put $\prescript{\zv}{}D_{0} = 
 \prescript{\zv}{}C+\prescript{\zv}{}R^{\kappa}$,
 $\prescript{\zv}{}B_{0} =
 \prescript{\zv}{}B\cap 
 \prescript{\zv}{}D_{0}$,
 $\prescript{\zv}{}C_{0} =
 \prescript{\zv}{}C\cap 
 \prescript{\zv}{}D_{0}=\prescript{\zv}{}C$,
 $\prescript{\zv}{}R_{0}^{\alpha} =
 \prescript{\zv}{}R^{\alpha}\cap 
 (\prescript{\zv}{}D_{0}\times \prescript{\zv}{}D_{0})$.
 Notice that by (r)
 we have $\prescript{\zv}{}D_{0}\supseteq \prescript{\zv}{}C$, 
 and by the property 
 $\forall \zv (\prescript{\zv}{}C+\prescript{\zv}{}R^{\kappa})\cap \prescript{\zv}{}B\neq \varnothing$
 we have $\prescript{\zv}{}B_0\neq \varnothing$.
 Then properties
($\kappa$), ($d+$), (un), ($\varnothing$), ($b+$), ($+c$), (t), and (r) 
follow from the corresponding properties for $R$.
Property ($bc$) follows from ($cd$) for $R$.
 \end{proof}

\begin{sublem}\label{increaseC}
$J\cup\{\neg c+,\exists \zv (\prescript{\zv}{}C+\prescript{\zv}{}R^{\kappa})\cap \prescript{\zv}{}B= \varnothing\}
\vdash
J +\text{keep } \sum_{\zv\in A^{|A|}}|\prescript{\zv}{}D|
+\text{increase } \sum_{\zv\in A^{|A|}}|\prescript{\zv}{}C|$.

\end{sublem}

\begin{proof}
We put $\prescript{\zv}{}C_{0} = 
\prescript{\zv}{}C+\prescript{\zv}{}R^{\kappa}$,
$\prescript{\zv}{}C_{1} = 
\prescript{\zv}{}R^{\kappa}+\prescript{\zv}{}C_{0}$,
and claim that 
$R$, $D$, $B$, and $C_{1}$ satisfy the required properties.
Since we only increased $C$ the only properties we need to check are ($\varnothing$) and ($+c$).
Property ($+c$) follows from property (r) and (t) for $R$. 
To prove property ($\varnothing$) choose $\zv\in A^{|A|}$ 
from the property 
$\exists \zv (\prescript{\zv}{}C+\prescript{\zv}{}R^{\kappa})\cap \prescript{\zv}{}B= \varnothing$.
Then 
$\prescript{\zv}{}C_{0}\cap \prescript{\zv}{}B= \varnothing$.
By property ($b+$) $\prescript{\zv}{}C_{1}\cap \prescript{\zv}{}B= \varnothing$, 
which gives us property ($\varnothing$).
\end{proof}

\begin{sublem}\label{makeSymmetricOne}
$J\cup\{bd,cd, 
\forall \zv \prescript{\zv}{}C+\prescript{\zv}{}R^{\kappa}= \prescript{\zv}{}D\}
\vdash
J\cup\{c+,\text{s}\}+\text{keep } \sum_{\zv\in A^{|A|}}|\prescript{\zv}{}D|$.
\end{sublem}

\begin{proof}
Put 
$\prescript{\zv}{}R_{0}^{\alpha}(x,y) =
\prescript{\zv}{}R^{\alpha}(x,y)\wedge 
\prescript{\zv}{}R^{\alpha}(y,x)$
and prove the required properties for
$R_{0}$, $B$, $C$, and $D$.
Properties 
($\kappa$), ($d+$), (un), ($\varnothing$),
($b+$), ($+c$), ($c+$), (t), (r), and (s) follow 
from the definition and the respective properties for $R$.
It remains to prove property ($bc$).

Notice that 
it follows from 
(t) for $R$ that 
$\prescript{\zv}{}R^{\forall}$ is transitive.
Let us
build an infinite path 
$d_{0}-d_{1}-d_{2}-d_{3}-d_{4}-d_{5}\dots$
such that 
each $d_{2i}$ is from $\prescript{\zv}{}B$,
each $d_{2i+1}$ is from $\prescript{\zv}{}C$,
each $(d_{i+1},d_{i})$ is from $\prescript{\zv}{}R^{\forall}$.
Choose some $d_{0}\in \prescript{\zv}{}B$.
By the property 
$\forall \zv \prescript{\zv}{}C+\prescript{\zv}{}R^{\kappa}= \prescript{\zv}{}D$,
there exists $d_{1}\in \prescript{\zv}{}C$ 
such that $(d_1,d_0)\in \prescript{\zv}{}R^{\kappa}\subseteq \prescript{\zv}{}R^{\forall}$.
By
property ($bd$) 
there exists 
$d_2\in\prescript{\zv}{}B$  such that 
$(d_2,d_1)\in \prescript{\zv}{}R^{\forall}$.
Proceeding this way we can make an infinite sequence.
Since $\prescript{\zv}{}B$ is finite,
we have 
$d_{2i}=d_{2j}$ for some $i< j$.
By transitivity of $\prescript{\zv}{}R^{\forall}$
we have $(d_{2i},d_{2i+1}),(d_{2i+1},d_{2i})\in\prescript{\zv}{}R^{\forall}$, 
which gives us property ($bc$).
\end{proof}


\begin{sublem}\label{makeSymmetricTwo}
$J\cup\{bd,cd,c+\}
\vdash
J\cup\{bd,cd,c+,\text{s}\}$.
\end{sublem}

\begin{proof}
Put
$B_0=B$ 
and 
$R_{0} = R$.
Define sequences
$\prescript{\zv}{}B_{i+1} =
 \prescript{\zv}{}B_{i}- 
\prescript{\zv}{}R^{\kappa}+\prescript{\zv}{}R^{\kappa}$ and
$\prescript{\zv}{}R_{i+1}^{\alpha} =
\prescript{\zv}{}R_{i}^{\alpha}- 
\prescript{\zv}{}R^{\alpha}+\prescript{\zv}{}R^{\alpha}$.

Since $\prescript{\zv}{}R^{\kappa}$ is reflexive, 
these sequences of relations are growing.
From the finiteness we conclude that 
these sequences will stabilize at some $N$.

Let us prove that $R_{N}$, $B_{N}$, $C$, and $D$ satisfy 
the required properties.
Properties 
($\kappa$), ($d+$), (un), (bc), ($+c$), (r), (bd), (cd), and ($c+$)
easily follow from the corresponding properties for $R$.
Properties ($b+$) and (t) follow from the fact that sequences stabilized. 
By ($c+$) and ($+c$) for $R$ we derive that 
we can never escape from $\prescript{\zv}{}C$ and therefore, 
by property ($\varnothing$) for $R$, we can never come to 
$\prescript{\zv}{}B$.
Hence we have ($\varnothing$).
Property (s) follows from properties (r) and (t) for $R$ and 
from the fact that the sequences stabilized.
\end{proof}

\begin{sublem}\label{CLAIMPropertiesDerivation}
$\mathrm{III}\vdash \mathrm{II}\cup\{\kappa\}$.

\end{sublem}

\begin{proof}
By Claim \ref{CLAIMAddKappa} we get property 
($\kappa$).
By Claim \ref{CLAIMFirstTransitive} we additionally get 
property (t).
By Claim \ref{CLAIMAddSDR} we additionally get 
($d+$), (r), and (sd). Thus, we get all the properties from $J$.

Iteratively applying Claims \ref{reduceDBright}, \ref{reduceDCleft}, 
\ref{reduceDCright}, and
\ref{increaseC} 
whenever possible
we either achieve additional properties ($bd$), ($cd$), and ($c+$),
or 
we get additional properties ($bd$), ($cd$) and 
$\forall \zv\prescript{\zv}{}C+\prescript{\zv}{}R^{\kappa} = 
\prescript{\zv}{}D$.
Notice that the process cannot last forever because at every step 
we either reduce 
$\sum_{\zv\in A^{|A|}}|\prescript{\zv}{}D|$ or increase 
$\sum_{\zv\in A^{|A|}}|\prescript{\zv}{}C|$.
If we get additional properties ($bd$), ($cd$), and ($c+$) 
then the statement follows from 
Claim \ref{makeSymmetricTwo}.
If we get additional properties ($bd$), ($cd$) and 
$\forall \zv \prescript{\zv}{}C+\prescript{\zv}{}R^{\kappa} = 
\prescript{\zv}{}D$ then 
we apply Claim \ref{makeSymmetricOne}.
If the obtained quadruple satisfies 
($bd$) and ($cd$) then we satisfied all the required properties.
Otherwise we apply 
Claims \ref{reduceDBright} and \ref{reduceDCleft}, 
and reduce the sum 
$\sum_{\zv\in A^{|A|}}|\prescript{\zv}{}D|$. Then we again 
apply Claims \ref{reduceDBright}, \ref{reduceDCleft}, 
\ref{reduceDCright}, and
\ref{increaseC} whenever possible and so on.
\end{proof}

Note that properties (r),(t),(s), and (sd) imply that 
$\prescript{\zv}{}R^{\alpha}$ is an equivalence relation
on $\prescript{\zv}{}D$,
and properties ($b+$) and $(+c)$ imply that 
$\prescript{\zv}{}B$ and $\prescript{\zv}{}C$
are unions of some equivalence classes of 
$\prescript{\zv}{}R^{\kappa}$.

\begin{lem}\label{LEMMakeAMightyTupleFromProperties}
Suppose a
quadruple $(R,D,B,C)$ satisfies all the properties 
from $\mathrm{II}\cup\{\kappa\}$. 
Then there exists a mighty tuple I q-definable from 
$R,D,B,C$.
\end{lem}
\begin{proof}
Define a mighty tuple $(R_{1},D_{1},B_1,C_1,\Delta_{1})$ 
by
 \begin{align*}
     \prescript{\zv}{}\Delta_{1}(u,v) =& \exists x(
\prescript{\zv}{}B(u)\wedge \prescript{\zv}{}C(v)\wedge \prescript{\zv}{}R^{\forall}(u,x) \wedge \prescript{\zv}{}R^{\forall}(v,x))\\
    \prescript{\zv}{uv}D_{1}(x) =& 
\prescript{\zv}{}B(u)\wedge \prescript{\zv}{}C(v)\wedge \prescript{\zv}{}R^{\forall}(u,x) \wedge \prescript{\zv}{}R^{\forall}(v,x)\\
 \prescript{\zv}{uv}R_1^{\alpha} =&
 \prescript{\zv}{}R^{\alpha}
 \cap
 (\prescript{\zv}{uv}D_{1}\times \prescript{\zv}{uv}D_{1})\\
 \prescript{\zv}{uv}B_{1}(x) =&
\prescript{\zv}{}B(u)\wedge \prescript{\zv}{}C(v)\wedge \prescript{\zv}{}R^{\kappa}(u,x) \wedge \prescript{\zv}{}R^{\forall}(v,x)\\
 \prescript{\zv}{uv}C_{1}(x) =&
\prescript{\zv}{}B(u)\wedge \prescript{\zv}{}C(v)\wedge \prescript{\zv}{}R^{\forall}(u,x) \wedge \prescript{\zv}{}R^{\kappa}(v,x)
 \end{align*}

By property ($bc$) for the quadruple 
we derive 
$\prescript{\zv}{}\Delta_{1}\neq\varnothing$,
$\prescript{\zv}{}D_{1}\neq\varnothing$,
$\prescript{\zv}{}B_{1}\neq\varnothing$, 
and $\prescript{\zv}{}C_{1}\neq\varnothing$ 
for every $\zv\in A^{|A|}$.
Thus, we already satisfied first two properties of 
a mighty tuple.
Property 3 from (r), (t), (s).
Property 4 follows from 
the definition and the fact that 
$\prescript{\zv}{}R^{\forall}$ is 
an equivalence relation.
Property 5 follows from the definition,
and property 6 follows from $(\varnothing$) for $R$.
\end{proof}

We will prove the following 
two lemmas from Section \ref{SUBSECTIONMightyTuplesDefinitions}
simultaneously.



\begin{LEMMightyTupleTwoThreeFourEquivalenceLEM}
Suppose $\Sigma$ is a set of relations on $A$.
Then the following conditions are equivalent:
\begin{enumerate}
    \item $\Sigma$ q-defines a mighty tuple II;
    \item $\Sigma$ q-defines a mighty tuple III;
    \item $\Sigma$ q-defines a mighty tuple IV.
\end{enumerate} 
\end{LEMMightyTupleTwoThreeFourEquivalenceLEM}

\begin{LEMMightyTupleTwoImpliesLEM}
Suppose $T$ is a mighty tuple of type II, III, or IV.
Then relations of $T$ q-define a mighty tuple I.
\end{LEMMightyTupleTwoImpliesLEM}
\begin{proof}

We want to prove that the existence of a mighty tuple II, III, or IV implies 
the existence of I, II, III, and IV.
First, let us show that we can derive a 
mighty tuple III.
A mighty tuple II is also a mighty tuple III, hence for 
II and III it is obvious.
If $(Q,D,B,C)$ is a mighty tuple IV, 
then $(Q,D,B,C)$ satisfies 
conditions from IV.
By Claims \ref{CLAIMAddKappa}, \ref{CLAIMSecondTransitive} and \ref{CLAIMSecondPlusC}
we derive a quadruple
satisfying all the properties of III,
which gives us a mighty tuple III.


Let $(Q,B,C)$ be a mighty tuple III. Put 
$\prescript{\zv}{}{D} = A$ for every $\zv$.
Then $(Q,D,B,C)$ satisfies all
properties of III.
Claim \ref{CLAIMPropertiesDerivation} implies 
the existence of a quadruple satisfying properties
$\mathrm{II}\cup \{\kappa\}$.
Notice that this quadruple is 
simultaneously a mighty tuple II, III, and IV.
Additionally, Lemma \ref{LEMMakeAMightyTupleFromProperties} 
implies that 
this quadruple q-defines a mighty tuple I.
\end{proof}

\subsection{Mighty tuple V}\label{SUBSECTIONMightyTupleV}

First, let us define a modification of a mighty tuple V, which we call
a mighty tuple V'.

A tuple $(Q,D,\Delta)$, where $\Delta$ is a $\zv$-parameterized $m$-ary relation, 
$Q$ is a$(\zv,\delta,\alpha)$-parameterized binary relation,
and $D$ is a $(\zv,\delta)$-parameterized unary relation, 
is called \emph{a mighty tuple V'} if 
\begin{enumerate}
    \item $\prescript{\zv}{}\Delta\neq\varnothing$ for every $\zv\in A^{|A|}$;
    \item $\prescript{\zv}{\delta}Q^{\kappa}\subseteq
    \prescript{\zv}{\delta}Q^{\alpha}$ for
    every  $\zv\in A^{|A|}$, $\delta\in \prescript{\zv}{}\Delta $, and $\alpha\in A^{|A|}$;

    \item  $\{(d,d)\mid d\in \prescript{\zv}{\delta}D\}\subseteq \prescript{\zv}{\delta}Q^{\forall}$ 
    for every $\zv\in A^{|A|}$ and $\delta\in \prescript{\zv}{}\Delta$; \;\quad($\prescript{\zv}{\delta}Q^{\forall}$ is reflexive)
\item $\proj_{1}(\prescript{\zv}{\delta}Q^{\alpha})= 
    \proj_{2}(\prescript{\zv}{\delta}Q^{\alpha})=\prescript{\zv}{\delta}D$ for 
    every  $\zv\in A^{|A|}$, $\delta\in \prescript{\zv}{}\Delta $, and $\alpha\in A^{|A|}$.
    \item $\prescript{\zv}{\delta}Q^{\forall\forall}\cap \{(d,d)\mid d\in A\}=\varnothing$  
    for some $\zv\in A^{|A|}$ and every $\delta\in\prescript{\zv}{}\Delta$. \quad ($\prescript{\zv}{\delta}Q^{\forall\forall}$ has no loops)
\end{enumerate}

Notice that we allow $\Delta$ to be of arity $0$. Then 
condition 1 means that $\Delta=\{\Lambda\}$, where 
$\Lambda$ is an empty tuple/word.
In this case we can omit a parameter $\delta$ in relations.

\begin{lem}\label{LEMMightyVImpliesPrime}
Suppose $(Q,D)$ is a mighty tuple V.
Then $\{Q,D\}$ q-defines a mighty tuple V' $(R,D,\{\Lambda\})$.
\end{lem}

\begin{proof}
Define 
a mighty tuple V' as follows.
Let $\Delta$ be the relation of arity $0$ containing the empty tuple.
The relation $R$ is defined by
$$\prescript{\zv}{}R^{x_1,\dots,x_{|A|}}(y_1,y_2)=
\bigwedge\limits_{i_1,\dots,i_k\in\{1,2,\dots,|A|\}.
}
\prescript{\zv}{}Q^{x_{i_1},\dots,x_{i_k}}(y_1,y_2)
\wedge (y_1\in \prescript{\zv}{}D)
\wedge (y_2\in \prescript{\zv}{}D).$$
Then 
$R^{\kappa} =R^{\forall\forall} = Q^{\forall\forall}$ and 
$R^{\forall} = Q^{\forall}$.
It is straightforward to check that $(R,D,\Delta)$ is a mighty tuple V'.
\end{proof}

\begin{lem}\label{LEMSymmetricMightyVImply}
Suppose $(R,D,\Delta)$ is a mighty tuple V', 
$\prescript{\zv}{\delta}R^{\alpha}$ is symmetric 
    for
    every  $\zv\in A^{|A|}$, $\delta\in \prescript{\zv}{}\Delta $, and $\alpha$.
Then there exists a mighty tuple I  q-definable from 
$R$, $D$, and $\Delta$.
\end{lem}


\begin{proof}
First, we assign a 
pair to every mighty tuple V' and 
evaluations of $\zv$ and $\delta$.
Put $\phi_{R,D,\Delta}(\zv,\delta) =
(m,|\prescript{\zv}{\delta}D|)$, where 
$m$ is the minimal odd positive integer such that 
$\prescript{\zv}{\delta}R^{\kappa}$ has cycles
of length $m$.
By $\phi^{1}_{R,D,\Delta}(\zv,\delta)$
we denote the first element of the pair, that is $m$.
Notice that $m$ can be $\infty$ if 
$\prescript{\zv}{\delta}R^{\kappa}$ does not have odd cycles.
Then we define a linear order on pairs by 
$(m_1,s_1)\le (m_2,s_{2})\Leftrightarrow
(m_1<m_2)\vee (m_1=m_2\wedge s_{1}\ge s_{2})$.
We put
$\phi_{R,D,\Delta} =\max\limits_{\zv\in A^{|A|}} \min\limits_{\delta\in\prescript{\zv}{}\Delta} 
\phi_{R,D,\Delta}(\zv,\delta)$.

We prove the lemma by induction on 
$\phi_{R,D,\Delta}$.
Assume that it does not hold.
Choose a mighty tuple V' 
$(R,D,\Delta)$ such that 
we cannot q-define a mighty tuple I from it and 
the pair $\phi_{R,D,\Delta}$ is maximal.
Thus, to complete the proof it is sufficent to q-define 
a mighty-tuple V' with larger pair or 
q-define a mighty tuple I.
Suppose 
$(m,s) = \phi_{R,D,\Delta}$.
By property 5 of a mighty tuple V' 
we have $m\ge 3$.
We consider two cases.

Case 1 (base of the induction). $m=\infty$, i.e., for some $\zv_{0}\in Z$ 
and every $\delta\in \prescript{\zv_0}{}{\Delta}$ 
the relation 
$\prescript{\zv_{0}}{\delta}R^{\kappa}$
has only cycles of an even length.
Then define
$\prescript{\zv}{\delta}R^{\alpha}_{0} = N\cdot \prescript{\zv}{\delta}R^{\alpha}$, 
where $N = |A|!\cdot |A|^{2}$.
Since $\prescript{\zv}{\delta}R^{\alpha}$ is symmetric, 
the relation 
$\prescript{\zv}{\delta}R^{\alpha}_{0}$ is reflexive and symmetric.
By Lemma \ref{LEMFactorialRepetition}  
$\prescript{\zv}{\delta}R^{\alpha}_{0}$ is transitive,
that is $\prescript{\zv}{\delta}R^{\alpha}_{0}+\prescript{\zv}{\delta}R^{\alpha}_{0} = \prescript{\zv}{\delta}R^{\alpha}_{0}$.
Therefore, $\prescript{\zv}{\delta}R^{\alpha}_{0}$
is an equivalence relation on $\prescript{\zv}{\delta}D$
for every $\zv$, $\delta$, and $\alpha$.
This implies that 
$\prescript{\zv}{\delta}R^{\forall}_{0}$ is also an equivalence relation.
Since 
$\prescript{\zv_{0}}{\delta}R^{\kappa}$ has no cycles of an odd length, 
we have 
$\prescript{\zv_{0}}{\delta}R^{\kappa}_{0}\cap
\prescript{\zv_{0}}{\delta}R^{\kappa}=\varnothing$.
By the reflexivity of 
$\prescript{\zv}{\delta}R^{\forall}$
we have
$\prescript{\zv}{\delta}R^{\kappa}\subseteq 
\prescript{\zv}{\delta}R_{0}^{\forall}$.
Then a mighty tuple I $(R_{1},D_{1},B_1,C_1,\Delta_{1})$ 
can be defined as follows:
 \begin{align*}
  \prescript{\zv}{}\Delta_{1}(\delta,u,v) =& 
\prescript{\zv}{\delta}R^{\kappa}(u,v)\wedge 
\prescript{\zv}{}\Delta(\delta)\\
    \prescript{\zv}{\delta uv}D_{1}(x) =& 
\prescript{\zv}{\delta}R^{\kappa}(u,v)\wedge 
\prescript{\zv}{\delta}R_{0}^{\forall}(u,x)\\
 \prescript{\zv}{\delta uv}R_1^{\alpha} =&
 \prescript{\zv}{\delta}R_0^{\alpha}
 \cap
 (\prescript{\zv}{\delta uv}D_{1}\times \prescript{\zv}{\delta uv}D_{1})\\
 \prescript{\zv}{\delta uv}B_{1}(x) =&
\prescript{\zv}{\delta}R^{\kappa}(u,v)\wedge\prescript{\zv}{\delta}R_0^{\kappa}(u,x)\\
 \prescript{\zv}{\delta uv}C_{1}(x) =&
\prescript{\zv}{\delta}R^{\kappa}(u,v)\wedge\prescript{\zv}{\delta}R_0^{\kappa}(v,x)
 \end{align*}
 Since we can choose any tuple 
 $(a,b) \in \prescript{\zv}{\delta}R^{\kappa}$ as $(u,v)$,
 we  obtain 
 $\prescript{\zv}{}{\Delta_{1}}\neq\varnothing$, 
 $a\in \prescript{\zv}{\delta ab}B_1\subseteq \prescript{\zv}{\delta ab}D_1$,
 and $b\in \prescript{\zv}{\delta ab}C_1\subseteq \prescript{\zv}{\delta ab}D_1$. 
 The relation $\prescript{\zv_{0}}{\delta uv}R_1^{\alpha}$ is an equivalence relation 
 because it is just a restriction of the equivalence relation 
 $\prescript{\zv_{0}}{\delta}R_0^{\alpha}$ to $\prescript{\zv}{\delta u v}D_1$.
 It follows 
 from the definition of $\prescript{\zv}{\delta}D_1$
 that $\prescript{\zv}{\delta}R_{1}^{\forall} = \prescript{\zv}{\delta}D_1\times \prescript{\zv}{\delta}D_1$.
The relations 
$\prescript{\zv}{\delta uv}B_{1}$ and 
$\prescript{\zv}{\delta uv}C_{1}$ are the equivalence classes of 
$\prescript{\zv}{\delta uv}R_{1}^{\kappa}$ 
containing $u$ and $v$, respectively.
Moreover, we 
have
$\prescript{\zv_0}{\delta uv}B_{1}\cap \prescript{\zv_0}{\delta uv}C_{1}=\varnothing$ 
because otherwise we would get a path of an even length from $u$ to $v$
in $\prescript{\zv_{0}}{\delta}R^{\kappa}$, 
which together with the edge $(u,v)$ would give us a cycle of an odd length
from $u$ to $u$ 
and contradicts our assumption about $\zv_{0}$.

Case 2 (inductive step). $m<\infty$.
Then there exists $\zv$ such that 
for any  
$\delta\in \prescript{\zv}{}\Delta$ 
the relation 
$\prescript{\zv}{\delta} R^{\kappa}$ has no cycles of length smaller than $m$
but for some 
$\delta\in \prescript{\zv}{}\Delta$
the relation 
$\prescript{\zv}{\delta} R^{\kappa}$ has cycles of length $m$.
Let $\prescript{\zv}{\delta}R_{0}^{\alpha}= 
\lfloor\frac{m}{2}\rfloor\cdot \prescript{\zv}{\delta}R^{\alpha}$.
Define new relations by
\begin{align*}
    \prescript{\zv}{}\Delta_1(\delta,y) &= 
\exists x\exists x'
\prescript{\zv}{\delta}R_{0}^{\kappa}(y,x)
\wedge
\prescript{\zv}{\delta}R_{0}^{\kappa}(y,x')
\wedge
\prescript{\zv}{\delta}R^{\kappa}(x,x')\wedge \Delta(\delta),\\
    \prescript{\zv}{\delta y} D_{1}(x) &= 
\exists x'
\prescript{\zv}{\delta}R_{0}^{\kappa}(y,x)
\wedge
\prescript{\zv}{\delta}R_{0}^{\kappa}(y,x')
\wedge
\prescript{\zv}{\delta}R^{\kappa}(x,x'),\\
\prescript{\zv}{\delta y}R_{1}^{\alpha}&=
\prescript{\zv}{\delta}R^{\alpha}\cap 
(\prescript{\zv}{\delta y}D_{1}\times \prescript{\zv}{\delta y}D_{1}).
\end{align*}
That is, 
$\prescript{\zv}{}{\Delta_1}(\delta,y)$ holds if
$\Delta(\delta)$ holds and 
$y$ is on some cycle of $\prescript{\zv}{\delta}R^{\kappa}$ of length $m$.
By the definition of maximality of $m$ over $\zv$ the relation 
$\prescript{\zv}{}{\Delta_{1}}$ is not empty for any $\zv$.
Also, 
$\prescript{\zv}{\delta y} D_{1}$ is the set of all elements 
such that there exist paths from $y$ to it of lengths 
$\lfloor\frac{m}{2}\rfloor$ and $\lceil\frac{m}{2}\rceil$.
Hence, $\prescript{\zv}{\delta} D_{1}$ is not empty for any $\zv$ and 
$\delta\in \prescript{\zv}{} \Delta_1$.
Let us show that 
$\phi_{R_{1},D_{1},\Delta_1}>\phi_{R,D,\Delta}$.
Notice that $\prescript{\zv}{\delta y}R_{1}^{\kappa}$ is just a restriction of 
$\prescript{\zv}{\delta}R^{\kappa}$.
Also if 
$\phi^{1}_{R,D,\Delta}(\zv,\delta) = m$
then 
$a\notin \prescript{\zv}{\delta a} D_{1}$
for any $\delta a\in \prescript{\zv}{}\Delta_{1}$.
Hence, $\prescript{\zv}{\delta}D\supsetneq 
\prescript{\zv}{\delta a}D_{1}$ in this case and 
we have 
\begin{align}
\label{EQUATIONgreater}
\forall \zv\; \forall \delta a\in\prescript{\zv}{}\Delta_{1} \;
\left(\phi^{1}_{R,D,\Delta}(\zv,\delta) = m \rightarrow 
\phi_{R,D,\Delta}(\zv,\delta)< 
\phi_{R_1,D_1,\Delta_{1}}(\zv,\delta a)\right)
\end{align}

Let $Z$ be the set of all 
$\zv$ such that 
$\prescript{\zv}{\delta}R^{\kappa}$ has no cycle of length
smaller than 
$m$ for every $\delta\in\prescript{\zv}{}\Delta$, i.e. 
$\min\limits_{\delta \in\prescript{\zv}{}\Delta} \phi^{1}_{R,D,\Delta}(\zv,\delta)= m$.
Let 
$\prescript{\zv}{}\Delta'$ be the set of all 
$\delta\in\prescript{\zv}{}\Delta$ such that 
$\prescript{\zv}{\delta}R^{\kappa}$ has a cycle of length 
$m$.
Notice that 
$\prescript{\zv}{}\Delta'$ is a projection of 
$\prescript{\zv}{}\Delta_{1}$
onto all the coordinates but the last one.
Then 
$\phi^{1}_{R,D,\Delta}(\zv,\delta)= m$ 
for any $\zv\in Z$ and $\delta\in \prescript{\zv}{}{\Delta'}$.
By 
(\ref{EQUATIONgreater})
 we have 
 \begin{align*}
 \phi_{R,D,\Delta} =
  \max\limits_{\zv\in A^{|A|}} \min\limits_{\delta \in\prescript{\zv}{}\Delta} \phi_{R,D,\Delta}(\zv,\delta)=
  \max\limits_{\zv\in Z} \min\limits_{\delta \in\prescript{\zv}{}\Delta} \phi_{R,D,\Delta}(\zv,\delta)
=
 \max\limits_{\zv\in Z} \min\limits_{\delta \in\prescript{\zv}{}\Delta'} \phi_{R,D,\Delta}(\zv,\delta)
 < &\\
 \max\limits_{\zv\in Z} \min\limits_{\delta \in\prescript{\zv}{}\Delta_1} \phi_{R_1,D_1,\Delta_1}(\zv,\delta)\le \max\limits_{\zv\in A^{|A|}} \min\limits_{\delta \in\prescript{\zv}{}\Delta_1} \phi_{R_1,D_1,\Delta_1}(\zv,\delta) =&\phi_{R_1,D_1,\Delta_1}.
 \end{align*}
It remains to check
that $(R_{1},D_{1},\Delta_{1})$ is a mighty tuple V'.
Property 1 was already mentioned. Properties 
2 and 3 follow from the respective properties for $(R,D,\Delta)$.
Property 5 follows from the fact that we only restrict
the relation $R$.
To prove property 4 notice that 
by the definition of $D_{1}$ for every element $x$ there is an element
$x'$ connected to $x$ in $\prescript{\zv}{\delta}R^{\kappa}$
and both $x$ and $x'$ are in $\prescript{\zv}{\delta}D_{1}$.
Hence, by the inductive assumption, $(R_{1},D_{1},\Delta_{1})$ q-defines a mighty tuple I, which completes the proof.
\end{proof}

\begin{LEMMightyTupleFiveImpliesLEM}
Suppose $(R,D)$ is a mighty tuple V.
Then there exists a mighty tuple I q-definable from $\{R,D\}$.
\end{LEMMightyTupleFiveImpliesLEM}

\begin{proof}
By Lemma \ref{LEMMightyVImpliesPrime}
there exists a mighty tuple V'  $(R,D,\{\Lambda\})$ 
q-definable from $Q$ and $D$.
Put 
$\prescript{\zv}{} R_{0}^{\alpha}= N\cdot \prescript{\zv}{}R^{\alpha}$, 
where $N = |A|!\cdot |A|^{2}$.
By Lemma \ref{LEMFactorialRepetition},
$\prescript{\zv}{} R_{0}^{\alpha}$
is transitive for any $\zv$ and $\alpha$.


Put $\prescript{\zv}{} D_{1}(x) = \prescript{\zv}{}R_{0}^{\kappa}(x,x)$. 
Since 
any element from a cycle of $\prescript{\zv}{} R^{\kappa}$ is in $\prescript{\zv}{} D_{1}$ and 
$\prescript{\zv}{} R^{\kappa}$ has cycles
by property 4 of a mighty tuple V', the set 
$\prescript{\zv}{} D_{1}$ is not empty.

Let 
$\prescript{\zv}{} R_{1}^{\alpha}(x,y) = 
\prescript{\zv}{} R_{0}^{\alpha}(x,y)
\wedge \prescript{\zv}{} R_{0}^{\alpha}(y,x)\wedge \prescript{\zv}{}D_{1}(x)\wedge \prescript{\zv}{}D_{1}(y)$.
Notice that 
the relation 
$\prescript{\zv}{} R_{1}^{\alpha}$ is reflexive on 
$\prescript{\zv}{} D_{1}$.
Then 
$\prescript{\zv}{} R_{1}^{\alpha}+\prescript{\zv}{} R_{1}^{\alpha}\supseteq \prescript{\zv}{} R_{1}^{\alpha}$
and by transitivity of 
$\prescript{\zv}{} R_{0}^{\alpha}$ 
we get the transitivity of $\prescript{\zv}{} R_{1}^{\alpha}$.
Thus,
$\prescript{\zv}{} R_{1}^{\alpha}$ is an equivalence relation
for every $\zv$, $\delta$, and $\alpha$. 
Consider two cases.

Case 1. There exists $\zv_0$ such that
we have $\prescript{\zv_0}{} R_{1}^{\kappa}\cap 
\prescript{\zv_0}{} R^{\kappa}=\varnothing$.
Then we define 
a mighty tuple I $(R_{2},D_{2},B_2,C_2,\Delta_{2})$  as follows:
 \begin{align*}
    \prescript{\zv}{}\Delta_2(u,v) =& 
\prescript{\zv}{}R^{\kappa}(u,v)
\\
    \prescript{\zv}{uv}D_{2}(x) =& 
\prescript{\zv}{}R^{\kappa}(u,v)\wedge 
\prescript{\zv}{}R_{1}^{\forall}(x,u)\\
 \prescript{\zv}{uv}R_2^{\alpha} =&
 \prescript{\zv}{}R_1^{\alpha}
 \cap
 (\prescript{\zv}{uv}D_{2}\times \prescript{\zv}{uv}D_{2})\\
 \prescript{\zv}{uv}B_{2}(x) =&
\prescript{\zv}{}R^{\kappa}(u,v)\wedge\prescript{\zv}{}R_1^{\kappa}(u,x)\\
 \prescript{\zv}{uv}C_{2}(x) =&
\prescript{\zv}{}R^{\kappa}(u,v)\wedge\prescript{\zv}{}R_1^{\kappa}(v,x)
 \end{align*}

Let us check that all the properties of a mighty tuple I are satisfied.
Since we can take any $(u,v)$ on a cycle (of length at most $|A|$) in $\prescript{\zv}{}R^{\kappa}$, 
we have 
 $\prescript{\zv}{}{\Delta_{2}}\neq\varnothing$, 
 $u\in \prescript{\zv}{uv}B_2\subseteq \prescript{\zv}{uv}D_2$,
 and $v\in \prescript{\zv}{uv}C_2\subseteq \prescript{\zv}{uv}D_2$. 
Property
$\prescript{\zv_{0}}{uv}B_{2}
\cap \prescript{\zv_{0}}{uv}C_{2}=\varnothing$
follows from the definition of case 1. Other properties are straightforward.

Case 2. 
For every 
$\zv$  we have 
$\prescript{\zv}{} R_{1}^{\kappa}\cap 
\prescript{\zv}{} R^{\kappa}\neq\varnothing$.
This means that we have 
$(b,c) \in \prescript{\zv}{} R^{\kappa}$
such that
there exists a path from $c$ to $b$ of length $N$.
Hence, $b$ is on a cycle of length $N+1$.
Since a minimal cycle going through $b$ is of size at most $|A|$, 
by repeating this cycle we can get a cycle of length $|A|!$.
Combining cycles of lengths $|A|!$ and $|A|!+1$
we can build a cycle of any sufficiently large length. 
Let $k\ge 1$ 
be the minimal 
number such that
for every $\zv$
the graph
$\prescript{\zv}{} R^{\kappa}$ has cycles of length $2^{k}$.
Since $\prescript{\zv}{} R^{\kappa}$ has no loops for some $\zv$ and 
has all sufficiently large cycles, $k$ is well-defined.
Put 
$\prescript{\zv}{} R_{3}^{\alpha}= 2^{k-1}\cdot \prescript{\zv}{}R^{\alpha}$,
$\prescript{\zv}{} D_{4}(x) = \exists y \prescript{\zv}{}R_{3}^{\kappa}(x,y)\wedge \prescript{\zv}{}R_{3}^{\kappa}(y,x)$,
$\prescript{\zv}{} R_{4}^{\alpha}(x,y) = 
\prescript{\zv}{} R_{3}^{\alpha}(x,y)
\wedge \prescript{\zv}{} R_{3}^{\alpha}(y,x)$.
Notice that 
$\prescript{\zv}{} D_{4}$ is the set of all elements appearing in cycles
of length $2^k$ in $\prescript{\zv}{}R^{\kappa}$, 
which is nonempty by our assumptions.
Then it is straightforward to verify that  
$(R_{4},D_{4},\{\Lambda\})$ is a mighty tuple V'.
Since 
the relation $\prescript{\zv}{} R_{4}^{\alpha}$ is symmetric, we can apply Lemma 
\ref{LEMSymmetricMightyVImply} to derive a mighty tuple I.
\end{proof}

\subsection{Classification for constraint languages with all constants}
\label{SUBSECTIONClassificationIdempotentProof}

\begin{LEMTHMIdempotantClassificationLEM}
Suppose $\Gamma\supseteq\{x=a\mid a\in A\}$ is a set of relations on $A$.
Then the following conditions are equivalent:
\begin{enumerate}
    \item $\Gamma$ q-defines a mighty tuple I;
    \item $\Gamma$ q-defines a mighty tuple II;
\item  there exist an equivalence relation $\sigma$ on $D\subseteq A$ and $B,C\subsetneq A$
such that 
$B\cup C = A$ and 
    $\Gamma$
    q-defines the relations 
    $(y_{1},y_{2}\in D)\wedge(\sigma(y_1,y_2)\vee (x\in B))$ and 
    $(y_{1},y_{2}\in D)\wedge(\sigma(y_1,y_2)\vee (x\in C))$.
\end{enumerate}
\end{LEMTHMIdempotantClassificationLEM}
\begin{proof}

Let us prove that 1 implies 3.
By Lemma \ref{LEMMightyTupleOnePrime} 
there is a mighty tuple I'
$(Q_0,D_0,B_0,C_0,\Delta)$ q-definable from $\Gamma$.
We derive the required relations in several steps.

\textbf{Get rid of unnecessary parameters.} 
Choose $\zv_0 = (b_1,\dots,b_{|A|})$ 
satisfying condition 6 of a mighty tuple I.
Then choose any $\delta_{0}\in \prescript{\zv_0}{}{\Delta}$.
We get rid of parameters $\zv$ and $\delta$ substituting 
the corresponding values:
$Q_{1} = \prescript{\zv_0}{\delta_{0}}{} Q_{0}$, 
$D_{1} = \prescript{\zv_0}{\delta_{0}}{} D_{0}$.

Notice that we do not care about the sets
$B_{0}$ and $C_{0}$ anymore because they are
not necessary for the case with all constant relations.
Later we only change $Q_{1}$ and $D_{1}$.

\textbf{Every equivalence relation must be trivial.}
Assume that 
for some $\alpha$ 
the binary relation 
$Q_{1}^{\alpha}$ is different from 
$D_{1}\times D_{1}$ and $Q_{1}^{\kappa}$.
Then choose $b\in D_{1}$ such that 
$\{b\}+Q_{1}^{\alpha}$ is not an equivalence class of 
$Q_{1}^{\kappa}$ and not $D_{1}$.
Put $D_{2} = \{b\}+Q_{1}^{\alpha}$.
Notice that $D_{2}$ is a unary relation q-definable from $\Gamma$.
Put $Q_{2}^{\xv}(y_1,y_2) =
Q_{1}^{\xv}(y_1,y_2) \wedge (y_{1}\in D_{2})\wedge (y_2\in D_{2})$ to reduce the domain to $D_{1}$.
We can repeat this while 
some of the equivalence relation $Q_{2}^{\alpha}$ is not 
$D_{2}\times D_{2}$ and not $Q_{2}^{\kappa}$.
Thus, we assume that  
$Q_{2}^{\alpha}$ is either 
$D_{1}\times D_{1}$ or $Q_{2}^{\kappa}$
for any $\alpha$.

\textbf{Find appropriate $B$ and $C$.}
Here we use the idea from the proof of Lemma \ref{LEMMightyTupleImplies}
but for a much easier case.
Let 
$\sigma_1,\dots,\sigma_{N}$ 
be the set of all injective mappings from 
$\{1,2,\dots,|A|\}$ to $\{1,2,\dots,|A|^{2}\}$.
Let 
$$U_{n}^{x_1,\dots,x_{|A|^{2}}} =
Q_{2}^{x_{\sigma_{1}(1)},\dots,x_{\sigma_{1}(|A|)}}
+
Q_{2}^{x_{\sigma_{2}(1)},\dots,x_{\sigma_{2}(|A|)}}
+\dots +
Q_{2}^{x_{\sigma_{n}(1)},\dots,x_{\sigma_{n}(|A|)}}$$
Since at least $|A|$ elements in the set 
$x_1,\dots,x_{|A|^{2}}$ are equal, 
there exists $i\in\{1,2,\dots,N\}$ 
such that 
$x_{\sigma_{i}(1)}=x_{\sigma_{i}(2)}=\dots=x_{\sigma_{i}(|A|)}$.
Since
$Q_{2}^{\forall}=D_2\times D_2$
the relation $U_{N}^{x_1,\dots,x_{|A|^{2}}}$
is equal to $\prescript{\zv}{\delta}D\times \prescript{\zv}{\delta}D$.
Consider maximal $n$ such that 
$U_{n}^{\alpha}\neq D_{2}\times D_{2}$ for some $\alpha$.
Put $L = U_{n}$ and 
$R^{x_1,\dots,x_{|A|^{2}}}
=Q_{2}^{x_{\sigma_{n+1}(1)},\dots,x_{\sigma_{n+1}(|A|)}}$.
We know that 
$L^{\alpha}\neq D_{2}\times D_{2}$ for some $\alpha$, 
$R^{\alpha}\neq D_{2}\times D_{2}$ for some $\alpha$, 
but   
$L^{\alpha}+ R^{\alpha}=D_{2}\times D_{2}$
for every $\alpha$.
Let $B_0\subseteq A^{|A|^2}$ be the set of all $\alpha$ such that
$L^{\alpha} = D_{2}\times D_{2}$,
$C\subseteq A^{|A|^2}$ be the set of all $\alpha$ such that
$R^{\alpha} = D_{2}\times D_{2}$.
Thus, we have 
$B, C\subsetneq A^{m}$, where 
$m=|A|^{2}$, such that 
$B\cup C = A^{m}$, 
$L^{\xv}(y_1,y_2) = (y_1,y_2\in D_{2})\wedge (Q_{2}^{\kappa}(y_1,y_2)\vee (\xv\in B))$,
and
$R^{\xv}(y_1,y_2) = (y_1,y_2\in D_{2})\wedge (Q_{2}^{\kappa}(y_1,y_2)\vee (\xv\in C))$.

\textbf{Reduce the arity of $B$ and $C$}.
Let $B,C,L$ and $R$ be the relations 
of the minimal arity satisfying the above properties, 
that is, 
$B, C\subsetneq A^{m}$ for some $m$, 
$B\cup C = A^{m}$, 
$L^{\xv}(y_1,y_2) = (y_1,y_2\in D_{2})\wedge (Q_{2}^{\kappa}(y_1,y_2)\vee (\xv\in B))$,
and
$R^{\xv}(y_1,y_2) = (y_1,y_2\in D_{2})\wedge (Q_{2}^{\kappa}(y_1,y_2)\vee (\xv\in C))$.
If $m=1$ then $L$ and $R$ are two relations we needed to define.
Thus, we assume that $m>1$.
Put 
\begin{align*}
    L_{0}^{x_1}(y_1,y_2)  &= \forall x_2 \dots \forall 
x_m L^{x_1,\dots,x_m}(y_1,y_2)\\
R_{0}^{x_1}(y_1,y_2)  &= \forall x_2 \dots \forall 
x_m R^{x_1,\dots,x_m}(y_1,y_2)\\
B_{0}(x_1) &= \forall x_2\dots \forall x_m B(x_1,\dots,x_{m})\\
C_{0}(x_1) &= \forall x_2\dots \forall x_m C(x_1,\dots,x_{m})
\end{align*}
Consider two cases:

Case 1. $B_{0}\cup C_{0} = A$. 
Then
$L_{0}$ and $R_{0}$ are ternary relations 
satisfying all the required properties, 
which contradicts our assumption about the minimality of $m$.

Case 2. $B_0\cup C_0 \neq A$. Choose $a\in A\setminus (B_0\cup C_0)$. 
Put
\begin{align*}
    L_{1}^{x_2,\dots,x_m}(y_1,y_2)  &= 
L^{a,x_2\dots,x_m}(y_1,y_2)\\
    R_{1}^{x_2,\dots,x_m}(y_1,y_2)  &= 
R^{a,x_2\dots,x_m}(y_1,y_2)\\
B_{1}(x_2,\dots,x_m) &=  B(a,x_2,\dots,x_{m})\\
C_{1}(x_2,\dots,x_m) &=  C(a,x_2,\dots,x_{m})
\end{align*}
Notice that $B_{1},C_{1}\subsetneq A^{m-1}$,
$B_{1}\cup C_{1}= A^{m-1}$, and 
the relations 
$L_{1}$ and $R_{1}$ 
again satisfy all the required properties but have 
smaller arity, which contradicts our assumptions.

2 implies 1 by Lemma \ref{LEMMightyTupleTwoImplies}.
It remains to prove that 
3 implies 1.
Suppose we have $B,C\subsetneq A$, an equivalence relation $\sigma$ on $D$ and two q-definable relations 
    $L(y_1,y_2,x)=(y_1,y_2\in D)\wedge (\sigma(y_1,y_2)\vee (x\in B))$ and 
    $R(y_1,y_2,x)=(y_1,y_2\in D)\wedge (\sigma(y_1,y_2)\vee (x\in C))$.
    Let us define a mighty tuple II $(Q,D',B',C')$.
Choose two elements $b$ and $c$ from $D$ that are not equivalent modulo $\sigma$. Put 
$\prescript{\zv}{\delta}{D'}(y) = D(y)$, and 
\begin{align*}
\prescript{\zv}{}{Q}^{x_1,x_{2}}(y_1,y_2) &= 
\exists y \; L(y_1,y,x_1)\wedge R(y,y_2,x_2)\\
\prescript{\zv}{}{B'}(y) &= 
\exists y' \forall x \; (y'=b) \wedge L(y,y',x)\\
\prescript{\zv}{}{C'}(y) &= 
\exists y' \forall x \; (y'=c) \wedge L(y,y',x)
\end{align*}
Notice that the parameter $\zv$ is fictitious.
Since $B\cup C = A$, we have
$\prescript{\zv}{}{Q}^{\forall} = D\times D$.
Since $B\neq A$ and $C \neq A$, we have
$\prescript{\zv}{}{Q}^{\forall\forall} = \sigma$.
Thus, $(Q,D',B',C')$
is a mighty tuple II, which completes the proof.
\end{proof}

\bibliographystyle{plain}
\bibliography{refs}

\end{document}